\documentclass[preprint,12pt]{elsarticle}

\usepackage{amsmath}
\usepackage{amssymb}
\usepackage{booktabs} 
\usepackage{lipsum}
\usepackage{amsfonts}
\usepackage{graphicx}
\usepackage{epstopdf}
\usepackage{algorithmic}
\usepackage[justification=centering]{subcaption}
\usepackage{graphicx}
\usepackage{multicol}
\usepackage{cellspace}
\usepackage{cancel}
\usepackage{multirow}
\usepackage{xcolor}
\usepackage{framed}
\usepackage{varwidth}
\usepackage{amsthm}

\usepackage{url}
\allowdisplaybreaks
\renewcommand{\vec}[1]{\boldsymbol{#1}}

\ifpdf
  \DeclareGraphicsExtensions{.eps,.pdf,.png,.jpg}
\else
  \DeclareGraphicsExtensions{.eps}
\fi
\DeclareDocumentCommand{\prot}{m o}  
{%
	\IfNoValueTF{#2}{#1}
	{$\text{#1}^{\text{#2}}$}
}

\newcommand{\myfrac}[3][0pt]{\genfrac{}{}{}{}{\raisebox{#1}{$#2$}}{\raisebox{-#1}{$#3$}}}

\newtheorem{lemma}{Lemma}

\newtheorem*{remark}{Remark}

\usepackage{amsopn}

\makeatletter
\newcommand*{\addFileDependency}[1]{
  \typeout{(#1)}
  \@addtofilelist{#1}
  \IfFileExists{#1}{}{\typeout{No file #1.}}
}
\makeatother

\usepackage{amssymb}

\journal{Journal of Theoretical Biology}

\begin{document}

\begin{frontmatter}


\title{Phenotypic variation modulates the growth dynamics and response to radiotherapy of solid tumours under normoxia and hypoxia}
 \author[label4,label1]{G.L. Celora}
 \author[label1]{ H. Byrne} 
 \author[label3]{C. Zois }
 \author[label5]{P.G. Kevrekidis}
 \fntext[label4]{celora@maths.ox.ac.uk}
\address[label1]{Mathematical Institute, University of Oxford, Oxford, UK}
\address[label3]{Molecular Oncology Laboratories, Department of Oncology, Oxford University, Weatherall Institute of Molecular Medicine, John Radcliffe Hospital, Oxford, United Kingdom.}
\address[label5]{Department of Mathematics \& Statistics,
University of Massachusetts, Amherst 01003 USA \fnref{label3}}




\begin{abstract}
In cancer, treatment failure and disease recurrence have been associated with small subpopulations of cancer cells with a \textit{stem}-like phenotype.
In this paper, we develop and investigate 
a phenotype-structured model of solid tumour growth in which cells are structured by a \textit{stemness} level, which varies continuously between stem-like and
terminally differentiated behaviours. 
Cell evolution is driven by proliferation and apoptosis, as well as
advection and diffusion with respect to the stemness structure variable.  We use the model to investigate how the environment, in particular oxygen levels, affects the tumour's population dynamics and composition, and its response to radiotherapy. We use a combination of numerical and analytical techniques to quantify how under physiological oxygen levels the cells evolve to a differentiated phenotype and  under low oxygen level (i.e., hypoxia) they de-differentiate. Under normoxia, the proportion of cancer stem cells is typically negligible and the tumour may ultimately become extinct whereas under hypoxia cancer stem cells comprise a dominant proportion of the tumour volume, enhancing radio-resistance and favouring the tumour's 
long-term survival. 
We then investigate how such phenotypic heterogeneity impacts the tumour's response to treatment with radiotherapy under normoxia and hypoxia. Of particular interest is establishing how the presence of
radio-resistant cancer stem cells can facilitate a tumour's regrowth following radiotherapy.
We also use the model to show how radiation-induced changes in tumour oxygen levels can give rise to complex re-growth dynamics. For example, transient periods of hypoxia induced by damage to tumour blood vessels may rescue the cancer cell population from extinction and drive secondary regrowth. Further model extensions to account for spatial variation are also discussed briefly. 
\color{black}
\end{abstract}




\end{frontmatter}



\section{Introduction}

\noindent Understanding of the mechanisms by which cancer is initiated and progresses continues to increase, and, yet, cancer remains one of the leading causes of premature mortality worldwide and a major barrier to increasing average life-expectancy. For example, in 2018, 9.6 million people are estimated to have died of cancer \cite{Bray2018}. Furthermore, treatment outcomes can differ markedly between patients with the same cancer type, with the emergence of resistance being one of the major causes of treatment failure.

Over the past twenty years, there has been a major shift in our perception of solid tumours; they are now regarded as heterogeneous tissues in which malignant cells interact with
normal cells and shape their environment in ways that favour malignant growth
\cite{Hanahan2000}. Cancer stem cells (CSCs) were introduced to explain intra-tumour heterogeneity via the \emph{CSC hypothesis} \cite{Reya2001}. 
This hypothesis proposes that, while CSCs may comprise only a small fraction of the total cell population, their high clonogenic potential and their ability to produce more mature, or specialised, cancer cells enables them to create an entire tumour \cite{Rycaj2014}. As CSCs are found to be resistant to standard treatments, 
they are recognised as a major cause of disease recurrence and treatment failure \cite{Baumann2008,Rycaj2014,cancersreview}. These observations have stimulated the development of novel therapeutic strategies which aim to eradicate CSCs \cite{ende2,Kong2020,Shibata2019,frontiers}. 
In practice, the plasticity of CSCs represents a major obstacle to such treatments. Additionally, CSCs can adapt to their local micro-environment, and remodel it to create and maintain a niche which supports their survival~\cite{architects}.

Increasingly, researchers are turning to mathematical models in order to understand how CSCs affect the growth and composition of tumours, particularly their heterogeneity and response to treatment. These models often decompose 
the tumour into a series of compartments, each representing a particular
cell subtype. For example, in~\cite{ende2}, Enderling distinguishes cancer stem cells (CSCs) and cancer
cells, whereas Saga and coworkers distinguish radio-resistant and radio-sensitive cells~\cite{Saga}, and Scott and colleagues distinguish
tumour-initiating cells (or CSCs), transit-amplifying cells and terminally differentiated cells (TDCs)~\cite{Scott169615}.
Thus, most compartmental models are based on the CSC hypothesis which assumes that it is possible to distinguish between cancer stem cells and the tumour bulk. 
However, this paradigm has been challenged by recent experimental
studies~\cite{dirkse, Soleymani2018} that highlight the phenotypic heterogeneity and plasticity of cancer cells, whose clonogenic (or \emph{stemness}) potential can be altered by the surrounding micro-environment (extrinsic forces). These findings have 
led to a new hypothesis for intra-tumoural heterogeneity, based on \emph{adaptive CSC plasticity} \cite{Fanelli2020}. Under this hypothesis, cancer cells move between stem-like and terminally differentiated states in response to extrinsic (environmental) and/or intrinsic (random epigenetic mutation) forces. Remarkably, the development of state-of-the-art
experimental tools, such as single-cell RNA-seq, means that it is now possible to track the evolution of stemness traits ~\cite{tirosh,venteicher},
rendering this an ideal time to develop mathematical models that can explore these concepts.

Compartmental models can be used to study adaptive CSC plasticity , by allowing transitions between different compartments. However, since they assume that the tumour comprises distinct cell populations, with distinct properties, they are unable to account for continuous variation in cell properties. An increasingly popular mathematical approach for describing population heterogeneity and plasticity 
characterises tumour cells by their position on a continuous phenotypic axis. Position on the phenotypic axis determines cell properties such as resistance to treatment~\cite{Chisholm2016,Clairambault2020,hodgkinson,Lorenzi2016,Lorz2014} and/or metabolic state \cite{Ardaseva2019,Villa2019}. This approach is motivated by concepts from evolutionary ecology, such as risk-spreading
through spontaneous (epigenetic or genetic) variations and evolutionary pressure  \cite{Thomas2013}. The resulting models are typically formulated as systems of reaction-diffusion equations~\cite{Ardaseva2019,Lorenzi2016,Villa2019}, with an advective transport term sometimes included to account for biased mutation dynamics \cite{hodgkinson} or adaptive phenotypic switches \cite{Chisholm2016,Lorenzi2015,Stace2020}. 

In this paper, we formulate a mathematical model that accounts for the evolution of a cancer cell population along such a stemness axis in response to extrinsic and intrinsic stimuli. Initially, we focus on the plastic response of cells to changes in  nutrient levels, in particular oxygen. This is motivated by recent experimental studies~\cite{Garnier2019,Liu2014,Pistollato2010,Pistollato2009} suggesting that hypoxia (i.e. low oxygen levels) is a key driver of cell de-differentiation. 
From this point of view,  spatial heterogeneity
may introduce significant additional complications: as oxygen diffuses into a tumour and is consumed by cells, spatial gradients in the oxygen levels are established.
In this way, local
micro-environments characterised by normoxia, hypoxia and necrosis form as the distance to the nearest nutrient supply (i.e., blood vessels) increases~\cite{hodgkinson,Lorz2014,Villa2019}. For simplicity, we postpone consideration of such spatial complexity to future work and focus, instead, on a \emph{well-mixed} setting where oxygen levels are homogeneous and prescribed. This idealised scenario allows us to investigate how cell properties, such as proliferation, apoptosis
and adaptive response to environmental signals, contribute to the emergence of heterogeneous stemness levels in the population and the long term tumour composition. In this regard, we are interested in identifying conditions under which CSCs are favoured. We then extend the model to account for treatment via a phenotypically-modulated linear-quadratic model
of radiotherapy (see, e.g.,~\cite{lewin2,lewin,Saga} for recent
discussions) which accounts for differential radio-sensitivity of CSCs~\cite{frontiers}.
This allows us to investigate how different radiotherapy protocols perturb the phenotypic distribution and subsequent regrowth of the tumour. 

In practice, stemness is just one of multiple traits that regulate cell behaviour and heterogeneity. We, therefore, anticipate that future models will combine multiple phenotypic axes or \emph{synthetic dimensions}, such as stemness and metabolic state \cite{Ardaseva2019,hodgkinson}. Given the complexity of such multi-dimensional models, it is important first to understand these aspects separately. Noting that considerable mathematical effort has been devoted to investigating cancer
metabolism~\cite{cancermetab}, we choose here to focus on population heterogeneity with respect to a continuously varying stemness axis. We hope that in the long term this work will help motivate a systematic experimental characterization of cell plasticity and phenotype.

The remainder of the article is organised as follows. 
In Section~\ref{model presentation},
we present a well-mixed, spatially homogeneous, model of solid tumour growth in response to a prescribed 
oxygen concentration. We first investigate the population dynamics in the absence of treatment, considering both normoxic and hypoxic conditions. Numerical results are presented in Section~\ref{notreatment}. As a partial validation of the numerical results, we use spectral stability analysis to characterise the long time behaviour of the solutions. Section~\ref{radio_result} focuses on tumour cell responses to different radiotherapy protocols. As in Section~\ref{notreatment}, we simulate responses under normoxia and hypoxia, but we also consider situations in which the environment alternates between periods of hypoxia and normoxia in order to explore the different ways that radiotherapy can alter tissue oxygenation. Finally in Section~\ref{conclusion}, we summarize our key findings and propose possible 
directions for future work. We also present preliminary results showing how accounting for spatial and phenotypic variation may affect a tumour's growth and response to radiotherapy. 

\section{Model Formulation}
\label{model presentation}

\noindent We consider the temporal evolution of a heterogeneous population of tumour cells, $N(z,t)$, where $t \geq 0$ denotes time and $z$ ($0 \leq z \leq 1$) represents their stemness or \emph{clonogenic capacity}.
As shown in Figure \ref{schematic}, $z=0$ corresponds to cancer stem cells (CSCs) which have the maximum level of stemness, and $z=1$ corresponds to terminally differentiated cells (TDCs), which have lost their proliferative capacity 
and which can either enter replicative senescence or undergo cell death \cite{Lee2014}. %
\color{black} We assume that the population dynamics may by described by a reaction-advection-diffusion equation~(see Eq.~\ref{full1} below) which accounts for two essential physical/ecological processes.
First, cells \emph{move} along the stemness axis (i.e., in the $z$-direction) in response to extrinsic (micro-environment) and intrinsic (random epimutation) \emph{forces} \cite{Scott169615}, which give rise to advective and diffusive fluxes respectively. 
Finally, the effect of natural selection on the population is represented by the fitness function $F$, which models the net growth rate of the cells. 

\begin{figure}[h!]
	\centering
	\includegraphics[width=0.85\textwidth]{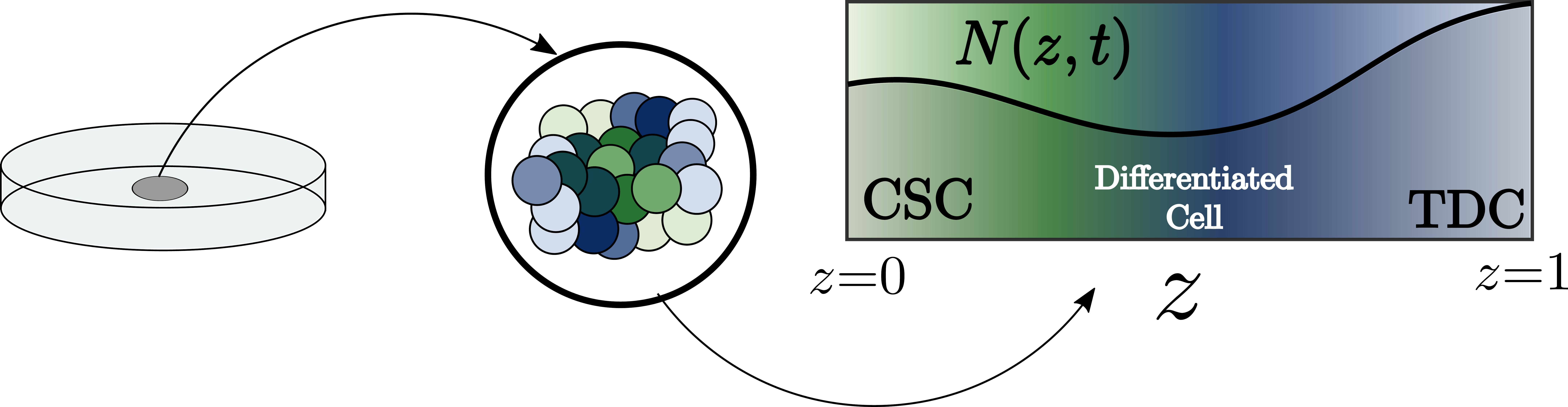}
	\caption{Schematic representation of the well-mixed, phenotypic model. We associate with each cell a stemness level $z$, which varies continuously between the cancer stem cell state (CSCs, with $z\sim0$), the differentiated cell state (with $z\sim 0.5$) and the terminally differentiated cell state (TDCs, with $z\sim 1$).}
	\label{schematic}
\end{figure}

While multiple nutrients and growth factors regulate the growth rate (or fitness function $F$) and phenotypic adaptation (i.e., the advective velocity $v_z$) of the tumour cells, 
here, for simplicity, we focus on a single nutrient, specifically oxygen. 
The critical role of low oxygen levels, or \textit{hypoxia}, in cancer has long been recognised due to its association with 
cell quiescence and poor therapeutic outcomes~\cite{hodgkinson,lewin,Saga}. 
Recent experimental results \cite{Soleymani2018} have shown that hypoxia also plays a role in de-differentiation by regulating pathways associated with a stem-like phenotype. 
We account for these phenomena in our model by assuming 
that all cells are exposed to the same level of oxygen, $c=c(t)$,
which mediates the values of the fitness function, $F$, and the advection velocity, $v_z$; the 
latter feature distinguishes our work from existing theoretical models in which intrinsic forces are assumed to dominate phenotypic variation (i.e., $v_z = 0$) \cite{Ardaseva2019, Villa2019}. 
By combining the processes mentioned above, we deduce that the evolution over time $t$ and along the phenotypic axis $z$ 
of the cell concentration, $N(z,t)$, is governed by the following non-local partial differential equation (PDE) and associated boundary and initial conditions:

\begin{subequations}
	\begin{align}
	\frac{\partial N}{\partial t}= \frac{\partial }{\partial z} \underbrace{\left(\theta \frac{\partial N}{\partial z}-N v_z(z,c)\right)}_{\text{structural flux}}+ \underbrace{F(z,\Phi,t;c)}_{\text{fitness}}N,\label{full1}\\
	\theta \frac{\partial N}{\partial z}-N v_z = 0, \quad z=\left\{0,1\right\}, \, t>0,\\[1pt]
	N(z,0)= N_0(z) \quad  z\in(0,1),\\[1pt]
	\Phi(t)=\int_0^1 N(z,t) \, dz.
	\end{align}%
In Equation~(\ref{full}), the non-negative constant $\theta$ represents the rate at which cells diffuse along the phenotypic axis, due to random epigenetic mutations, $\Phi(t)$ denotes the
	density of cells in the domain at time $t$, and $N_0(z)$ is the initial distribution of cells along the phenotypic axis.
	In ecology, the function $F$ is referred to as fitness landscape which is a mathematical representation of natural, or \textit{Darwinian}, selection \cite{Pisco2015}. We suppose it has the following form:
	\begin{align}
	\begin{aligned}
	F(z,\Phi,t;c)= \underbrace{p(z,c)\left(1-\frac{\Phi}{\Phi_{max}}\right)}_{\text{proliferation}}-\overbrace{f(z)}^{\text{\shortstack{natural cell\\ death}}}-\underbrace{\sum^{M}_{i=1} \log\left(\frac{1}{SF(z,c)}\right)\delta(t-t_i)}_{\text{radiotherapy}}.\label{F}
	\end{aligned}
	\end{align}\label{full}
\end{subequations}
In Equation~(\ref{F}), $p=p(z,c)$ denotes the phenotype-dependent growth rate of the cells (see Section~\ref{fit} for details). It is multiplied by a non-local (in the phenotypic sense) logistic term, with constant carrying capacity $\Phi_{max}$, to capture intra-population competition for space and other resources. We assume that oxygen levels remain sufficiently high so that necrosis can be neglected. Hence, the death rate, $f$, accounts only for natural cell death, or apoptosis, which is assumed to occur at a rate which is independent of the oxygen concentration, $c(t)$. Radiotherapy (RT) also contributes to cell death and, in so doing, reduces cell fitness. We suppose that $M$ rounds of RT are administered at discrete times $t_i$ ($i=1,2,\ldots, M$). After each treatment dose, the proportion of cells of phenotype $z$ that survive is denoted by the survival fraction $SF(z,c)$. By allowing $SF$ to depend on $z$, we can account for phenotypic-dependent radio-sensitivity, and, for example, view the CSCs (i.e. $z=0$) as the most radio-resistant tumour subpopulation \cite{Rycaj2014}. Additionally, the dependence of $SF(z,c)$ on $c(t)$ enables us to account for differential radio-sensitivity under normoxia and hypoxia \cite{Hockel1996,Sorensen2020}.  
In contrast to~\cite{lewin2}, where the
term $(1-SF)$ is used to capture cell death due to radiotherapy, here we use the term
$\log(1/SF)$, to ensure that the jump in tumour cells following each dose of radiotherapy is consistent with the Linear-Quadratic (LQ) model.  

We now partially rescale our model by recasting the dependent variables $N$ and $\Phi$ in the following way:
\begin{equation}
 n = \frac{N}{\Phi_{max}}, \qquad \phi = \frac{\Phi}{\Phi_{max}},
\end{equation}
where the units of time, $t$ [hr] are preserved in a dimensional form to facilitate the interpretation of the results. Under this rescaling, equations~(\ref{full}) become
\begin{subequations}
	\begin{align}
	\hspace{-10mm}
	\frac{\partial n}{\partial t}= \frac{\partial }{\partial z} \left(\theta \frac{\partial n}{\partial z}-n v_z(z,c)\right)+ F(z,\phi,t;c) n,\\
	\theta \frac{\partial n}{\partial z}-n v_z = 0, \qquad z\in\left\{0,1\right\}, \, t>0,\\
	n(z,0)= n_0(z) \quad z\in(0,1),\\
	\phi(t)=\int_0^1 n(z,t) \, dz,\\
	\begin{aligned}
	F(z,\phi,t;c)= p(z,c)\left(1-\phi\right) -f(z)-\sum^{N}_{i} \log\left(\frac{1}{SF(z,c)}\right)\delta(t-t_i).\label{Fitness}
	\end{aligned}
	\end{align}

In order to complete the model, it remains to specify several functional forms; this will be done in Sections \ref{fit} and \ref{vz_sec}. 
Extending the model to account for spatial variation is presented in~\ref{AppendixA}, and preliminary results are included in Section~\ref{conclusion} (a full investigation of the spatially-extended model is postponed to future work). 

In what follows, we assume that oxygen concentration $c$ has been rescaled so that $c=1$ corresponds to physiological oxygen levels, namely \textit{physoxia}, which is about $8\%$ oxygen \cite{McKeown2014}. 
When considering hypoxia, we focus on mild hypoxia, fixing $c=0.2$ which corresponds to $1.6\%$ oxygen in standard units 
(see~\ref{App_param} for details). 
At this oxygen concentration, necrosis can be neglected; it typically occurs at lower oxygen tensions (approximately $0.1\%$ oxygen in standard units).

Unless otherwise stated, we assume that the tumour initially comprises a small population of differentiated cells so that 
\begin{align}
n_0(z)=\frac{\phi_0}{\sqrt{2\pi \sigma^2}} e^{-\frac{\left(z-0.5\right)^2}{2\sigma^2}}\label{initial_cond},
\end{align}\label{mixedmodel}%
where the positive constants $\phi_0$ and $\sigma$ specify the initial size and phenotypic variance of the population.
\end{subequations}

The proportion of CSCs is often
used to characterise heterogeneous populations of cancer cells. 
CSCs are typically identified by their expression of specific markers (such as CD44/CD24 and ALDH1, depending on the tumour type \cite{frontiers}); thresholds in these markers are used to distinguish stem from differentiated cancer cells. 
Since our model treats stemness as a continuously varying cell property, we introduce a threshold $z^* \in (0,1)$ in our simulations, and classify cells with $0<z<z^*$ as CSCs. 
We therefore define the proportion of stem cells at time $t$ to be:
\begin{equation}
\phi_{CSC}(t,z^*)=\frac{\int_0^{z^*} n(t,s) \,ds}{\phi(t)}.
\label{cum_dist}
\end{equation} 
As a further statistical feature of the cell population, we introduce the phenotypic mean, $\mu(t)$, which is defined as follows:
\begin{equation}
\mu(t)=\frac{1}{\phi(t)} \int_0^1 z n(z,t) dz. \label{mean}
\end{equation}
In the absence of suitable experimental data, it is difficult to specify many of the parameters and functional forms in Equations~(\ref{mixedmodel}). 
For this reason, we focus on identifying the qualitative behaviours that the model exhibits across a range of `biologically-reasonable' situations.

\subsection{Fitness Landscape}
\label{fit}
\noindent When considering  the fitness landscape, we assume that, for fixed values of $c$, the proliferation rate, $p(z,c)$ has a multi-peaked profile, with local maxima centred around $z=0$ and $z=0.55$, representing respectively cells with stem-like ($z=0$) and intermediate phenotypes ($z=0.55$, this value being arbitrary). As shown in Figure \ref{fitness landscape}, this choice reduces the overlap of the two Gaussian profiles while maintaining the proliferation rate at $z=1$ close to zero. This asymmetry also emphasises that, under normoxia, more stem-like cells (i.e. $z<0.5$) proliferate at lower rates than more differentiated cells (i.e. $z>0.5$). 
Different environmental conditions (i.e., oxygen concentrations), will create distinct ecological \textit{niches} each of which will favour a particular phenotype. We account for this effect by assuming that the amplitude of the peaks in the proliferation rate are oxygen dependent. Accordingly we  write:
\begin{subequations}
	\begin{align}
	p(z;c)=p_0(c)\exp\left[-\frac{z^2}{g_0}\right]+p_1(c)\exp\left[-\frac{(z-0.55)^2}{g_1}\right],\\
	p_i(c)=p_i^{max}\myfrac[2pt]{c^4}{K_{i}^4+c^4}, \quad i=0,1,\label{pO2}
	\end{align}
	\label{eqnetprol}%
\end{subequations}
where $p_0(c)$ and $p_1(c)$ are Hill–Langmuir type equations with fourth order exponents, so that the growth rate decays rapidly when $c\sim K_{i}$. We assume that differentiated cells are fitter than CSCs under normoxia and, therefore, choose $p_1^{max}>p_0^{max}$. At the same time, we note that chronic hypoxia is widely considered to favour CSCs \cite{Ayob2018, Conley2012, Lan2018}. The plasticity of CSCs enables them more easily to adapt their metabolism to changing nutrient levels than differentiated cells \cite{Garnier2019,Snyder2018} and, therefore, to survive and proliferate in challenging conditions. This behaviour contrasts with that of differentiated cancer cells which tend to become quiescent when exposed to hypoxia. We account for these effects by assuming $K_0\ll K_1$. 

When we consider the rate of cell death due to apoptosis, $f(z)$, 
we note that apoptosis occurs predominantly when cells lose their clonogenic capacity. 
As such, it predominantly affects only TDCs with $z\sim 1$. 
Motivated by the mathematical models developed in \cite{ende2,Scott169615}, we propose
the following monotonically increasing function for $f(z)$: 
\begin{equation}
f(z)=d_f \,e^{-k_f(1-z)}.\label{apoptosis}
\end{equation}
Even though they may not proliferate, TDCs compete for space and resources and, thus, impact the tumour dynamics.
In what follows, we consider two different cases. First, guided by experimental results reported by Driessens et al.~\cite{Driessens2012}, we assume that apoptosis of TDCs occurs on a much longer timescale than that on which 
cells proliferate so that $d_f<<\max_{z} p(z;1)$. In the second case, the rates of cell proliferation and apoptosis are assumed to be comparable. This situation represents a tumour with high cell turnover and, as we will see, gives rise to a tumour population with higher clonogenic capacity.

In Figure \ref{fitness landscape}, we sketch the fitness landscape $F(z,0,t;c)$ for different environmental conditions in the absence of treatment and competition. In doing so, we have neglected competition and radiotherapy in Equations~(\ref{Fitness}), where $p$ and $f$ are defined by Equations~(\ref{eqnetprol})-(\ref{apoptosis}).  

\begin{figure}[h!]	
	\begin{subfigure}{0.32\textwidth}
		\centering
		\includegraphics[width=0.95\textwidth]{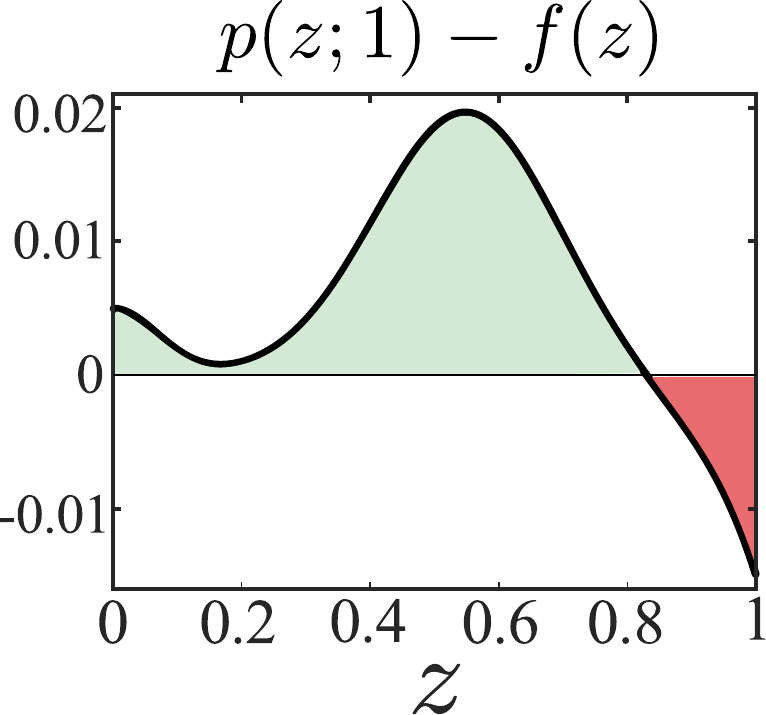}
		\caption{high $d_f$}
		\label{landnorm2}
	\end{subfigure}
	\begin{subfigure}{0.32\textwidth}
		\centering
		\includegraphics[width=0.95\textwidth]{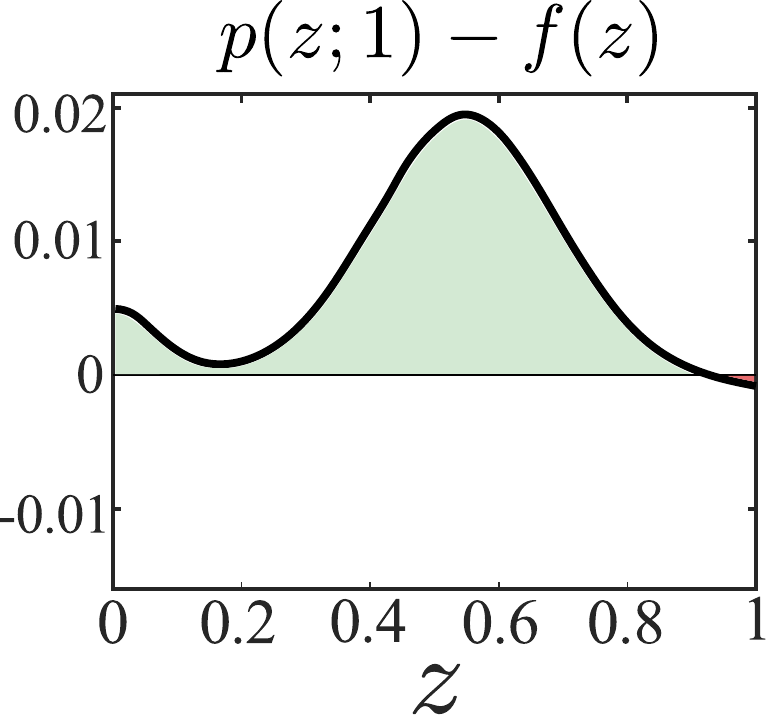}
		\caption{low $d_f$}
		\label{landnorm1}
	\end{subfigure}	
	\begin{subfigure}{0.32\textwidth}
		\centering
		\includegraphics[width=0.95\textwidth]{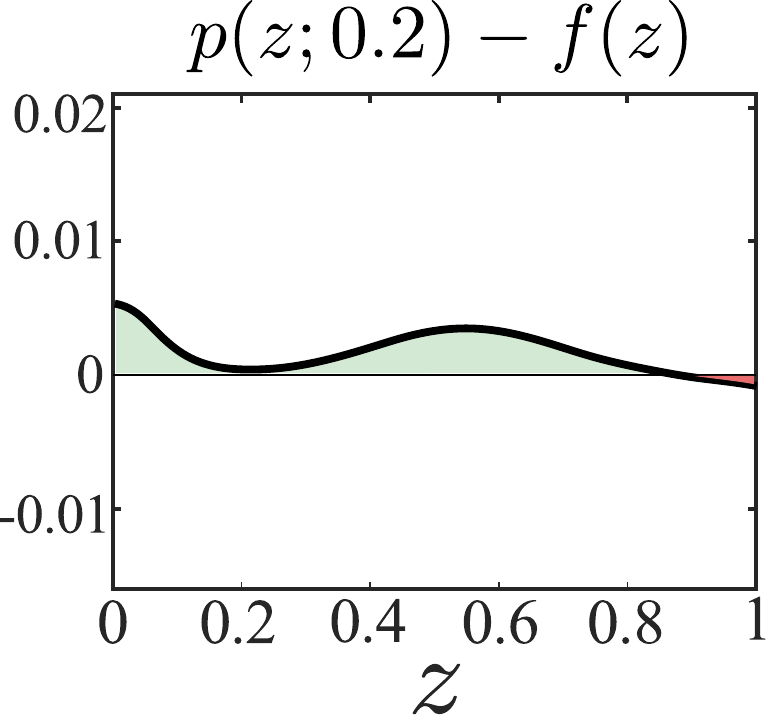}
		\caption{low $d_f$}
		\label{landhyp}
	\end{subfigure}	
	\caption{Series of sketches showing how the maximum growth rate $p(z,c)-f(z)$, as defined by Equations~(\ref{eqnetprol})-(\ref{apoptosis}) changes in different micro-environments: 
		(a)-(b) under normoxia ($c=1$), the progenitor cells ($z=0.55$) are the fittest phenotype, and the death rate may be either high (a) or low (b); (c) under hypoxia ($c=0.2$), the CSCs ($z=0$) are the fittest phenotype. The parameter values used to produce the sketches
		are listed in Table~\ref{par_fitness}. Regions of positive and negative fitness are highlighted in green and red, respectively.}
	\label{fitness landscape}
\end{figure}
\begin{table}[h!]
	\centering
	\subfloat[proliferation]{
		\begin{tabular}{c|c c c }
			\toprule[1.5pt]\addlinespace[2pt]
			& $p^{max}_i\, (hr^{-1})$ &$K_i$&$g_i$ \\
			\hline\addlinespace[2pt]
			$i$=0&$0.005$&$0.05$&$0.01$\\
			$i$=1& $0.02$ &$0.3$&$0.04$\\
			\bottomrule[1.5pt]
	\end{tabular}}
	\hspace{10mm}
	\subfloat[apoptosis]{
		\begin{tabular}{ c c}
			\toprule[1.5pt]\addlinespace[2pt]
			$d_f\, (hr^{-1})$ & $k_f$\\
			\hline\addlinespace[2pt]
			$\left\{0.001,0.015\right\} $& $10$\\
			\bottomrule[1.5pt]
	\end{tabular}}
	\caption{Range of parameter values used in the sensitivity analysis. More information on the specific parameter choice can be found 
	in~\ref{AppendixA}.}
	\label{par_fitness}
\end{table}

We now consider the impact of radiotherapy on cell fitness.
As mentioned above, CSCs possess protective mechanisms
that enable them to withstand damage caused by
radiation and oxidative stresses~\cite{Radioresistance,Clark2016, Diehn2009,Rycaj2014,cancersreview,frontiers,Vassalli2019}. They are, therefore, more resistant to treatment than their differentiated counterparts. It is well known that local oxygen
concentration levels also affect treatment outcomes~\cite{Horsman2012,Moulder1987}. While we account for this effect in the full spatial model (see~\ref{AppendixA}), here we focus on the role of phenotype-dependent radio-sensitivity.
In particular, we adapt the standard Linear-Quadratic (LQ) model so that the tissue specific coefficients, $\alpha (Gy^{-1})$ and $\beta(Gy^{-2})$, are phenotype dependent:
\begin{subequations}
	\begin{align}
	-\log(SF)=\alpha(z) d + \beta(z) d^2,\label{lq_model}
	\end{align}%
	where $d$ is the radiation dose in Grays (Gy). Equation~(\ref{lq_model}) is the natural, continuum extension of previous works \cite{Leder2014,Saga}, in which two-compartment models are used to describe the time-evolution of cancer cells and cancer stem cells exposed to radiotherapy, and CSCs are assumed to be radio-resistant. Accordingly, here, we assume $\alpha$ and $\beta$ are increasing functions of the phenotype $z$~\cite{Saga,cancersreview,frontiers} of the following form:  
	\begin{align}
	\alpha(z)= \alpha_{min}+(\alpha_{max}-\alpha_{min})\tanh\left(\frac{z}{\xi_R}\right),\label{rad_alpha}\\[2pt]
	\beta(z)=\beta_{min}+(\beta_{max}-\beta_{min})\tanh\left(\frac{z}{\xi_R}\right).\label{rad_beta}
	\end{align}\label{radiosensitivity}
\end{subequations}

In Equations~(\ref{rad_alpha})-(\ref{rad_beta}), $\xi_R$, $\alpha_{min,max}$ and $\beta_{min,max}$ are non-negative constants with $\alpha_{min}<\alpha_{max}$ and $\beta_{min}<\beta_{max}$. 
Where possible, parameter estimates are taken from the literature (see~\cite{Saga} for estimates of $\alpha_{min,max}$ and $\beta_{min,max}$); the value of $\xi_R=0.2$ is instead chosen so that differentiated cells (i.e. $z>0.5$) have maximum sensitivity to treatment (i.e., $\alpha(z)\sim \alpha_{max}$ for $z>0.5$).

\begin{table}[h!]
	\centering
	\begin{tabular}{c|l  l | l l }
		\toprule[1.5pt]\addlinespace[3pt]
		& \multirow{2}{*}{$[\alpha_{min},\alpha_{max}](Gy^{-1})$}&\multirow{2}{*}{$[\beta_{min},\beta_{max}](Gy^{-2})$}&\multirow{2}{*}{$\displaystyle\frac{\alpha_{min}}{\beta_{min}}(Gy)$}&\multirow{2}{*}{$\displaystyle\frac{\alpha_{max}}{\beta_{max}}(Gy)$}\\
		&&&&\\
		\hline\addlinespace[2pt]
		R1& $[0.005,0.15]$ &  $[0.002,0.10]$ &2.5&1.5\\[2pt]
		R2& $[0.050,0.20]$ & $[0.020,0.05]$   &2.5&4\\[2pt]
		R3& $[0.005,0.40]$ & $[0.002,0.05]$ &2.5&8\\[1pt]
		\bottomrule[1.5pt]
	\end{tabular}
	\vspace{3mm}
	\caption{Summary of the parameter values used in Equation~(\ref{radiosensitivity}) to describe the three different RT responses used in model simulations. In all cases, we fix $\xi_R=0.2$.}
	\label{rad_tab2}
\end{table}

We consider three different parameter sets (see Table \ref{rad_tab2}); they may represent three cell populations which differ in their sensitivity to radiotherapy (RT).
For cases R1 and R3, CSCs (with $z\sim 0$) respond in the same way to RT, whereas differentiated cancer cells (with $z> 0.5$) respond differently. For case R1, the small value of $\alpha_{max}/\beta_{max}$ for the sensitive cells ($z=1$) corresponds to a late responding tissue, whereas for case R3, the large value of $\alpha_{max}/\beta_{max}$ corresponds to an early responding tissue, with a low repair capacity, for which fractionation is known to be beneficial \cite{McMahon_2018}. Finally, case R2 is intermediate between cases R1 and R3. By assuming heterogeneity in the cell response to RT, we allow consideration of the selective pressure of RT. For a given dosage and LQ model, differences in the radio-sensitivity of CSCs and differentiated cells are determined by the ratios $\alpha_{min}/\alpha_{max}\in (0,1)$ and $\beta_{min}/\beta_{max}\in(0,1)$. When both fractions are small, CSCs are more likely to survive RT than their differentiated counterparts and, therefore, the selective pressure of RT on the population is high. By contrast, as $\alpha_{min}/\alpha_{max}$ and $\beta_{min}/\beta_{max}$ approach the value of unity, RT offers no selective advantage to CSCs as, at leading order, the response is independent of phenotype.  The latter also depends on the specific dose applied. For example, for high doses the quadratic term in Equation~(\ref{lq_model}) is dominant and the selective pressure is only associated with the value of $\beta_{min}/\beta_{max}$. By contrast, for lower doses, the linear and non-linear effects contribute to cell killing and, so, the selective pressure of RT is associated with both $\alpha_{min}/\alpha_{max}$ and $\beta_{min}/\beta_{max}$. For these reasons, we will consider two different RT protocols: either a single dose of $10\,Gy$ is delivered or a fractionated schedule is used (here five doses of $2\,Gy$ are delivered over five consecutive days \cite{Brenner1991,Dale1985}). While R2 is expected to have the least RT selective pressure in both scenarios, this might be higher in R1 or R3 depending on the treatment protocol considered. 
\subsection{Structural Flux}
\label{vz_sec}

Plasticity is an essential feature of phenotypic adaptation to changing environmental conditions~\cite{dirkse,Pisco2015}. It assumes that cells with the same genome can acquire distinct phenotypes depending on their epigenetic status, which is also inheritable. Phenotypic variation may be mediated by random (spontaneous) \textit{epigenetic mutations} \cite{Lorenzi2016}, which we assume to be rare. We account for this effect by including in the structural flux a diffusion term with a constant diffusion coefficient $\theta=5\times 10^{-6}$ hr$^{-1}$(see
Equation \ref{full1}). Such random mutations should not favour any specific phenotype, and \textit{Darwinian} selection (i.e. the fitness function $F$) drives phenotypic evolution of the population. This aspect has been widely studied in previous work in order to investigate how cells adapt to different environments \cite{Ardaseva2019,Lorenzi2016,Villa2019}. At the same time, there is evidence that phenotypic switching may be mediated by environmental factors via \textit{Lamarckian} selection (or induction) \cite{Pisco2015}. In this framework, cells adapt to their environment~\cite{dirkse,schaider} by following a preferential (\emph{biased}) trajectory in phenotypic space. We can, therefore, envisage situations in which a subpopulation may be prevalent in a population without being the fittest subpopulation (i.e. the population with the highest proliferation rate). 
For example, recent studies have identified cell de-differentiation and CSC maintenance as stress responses to harsh environmental conditions \cite{Pisco2015}, including hypoxia. 
More specifically, cells respond to hypoxic stress by up-regulating Hypoxia Inducible Factors (HIFs) which, in turn, promote the expression of stem-related genes~\cite{Garnier2019,Liu2014,Pistollato2010,Pistollato2009}. HIF suppression has also been linked to cell differentiation and reduced levels of stemness~\cite{Shiraishi2017}. 
We account for such micro-environment mediated adaptation by incorporating an advective term in the structural flux. Cells 
are assumed to evolve along the stemness axis with a velocity $v_z=v_z(z,c)$, that depends on the oxygen concentration $c$ and cell phenotype $z$. Under normoxia, cells tend to differentiate, and $v_z>0$. From this point of view, the model is similar to classical age-structured models \cite{Perthame2007,Webb2008}, with $v_z$ being analogous to a \emph{maturation} velocity. In our model, however, \emph{ageing} (i.e. differentiation or loss of clonogenic potential \cite{Scott169615}) may be reversible. For example, under hypoxia (i.e. $c \leq c_H$), we assume  $v_z<0$ (see Figure \ref{vel_prof}) and a more stem-like character is promoted.

\begin{figure}[h!]
\begin{subfigure}{0.32\textwidth}
	\centering
	\includegraphics[width=0.95\textwidth]{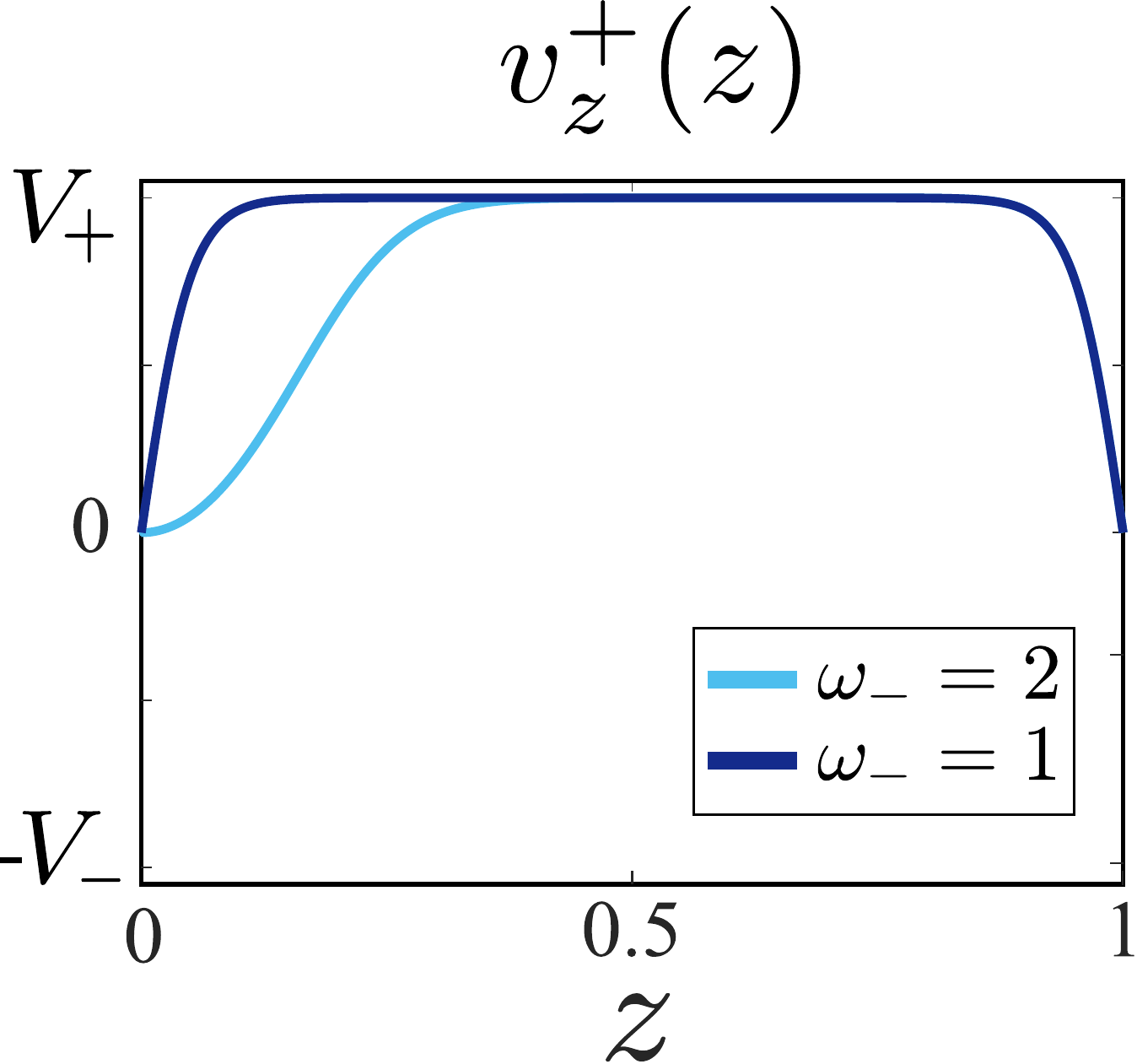}
	\caption{ $\xi_+=0.05$}
	\label{vel_norm1}
\end{subfigure}
\begin{subfigure}{0.32\textwidth}
	\centering
	\includegraphics[width=0.95\textwidth]{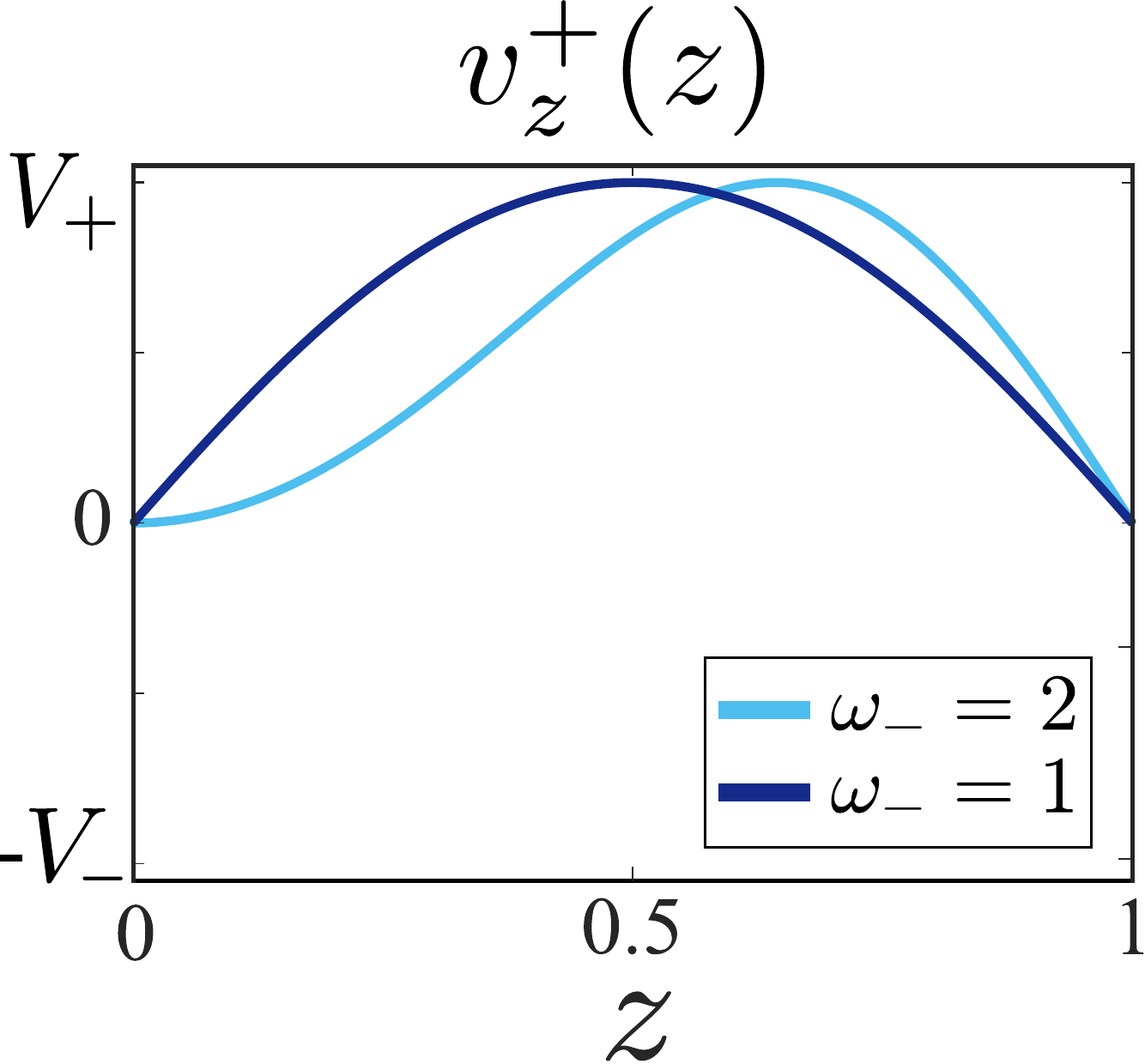}
	\caption{ $\xi_+=0.5$}
	\label{vel_norm2}
\end{subfigure}
\begin{subfigure}{0.32\textwidth}
	\centering
	\includegraphics[width=0.95\textwidth]{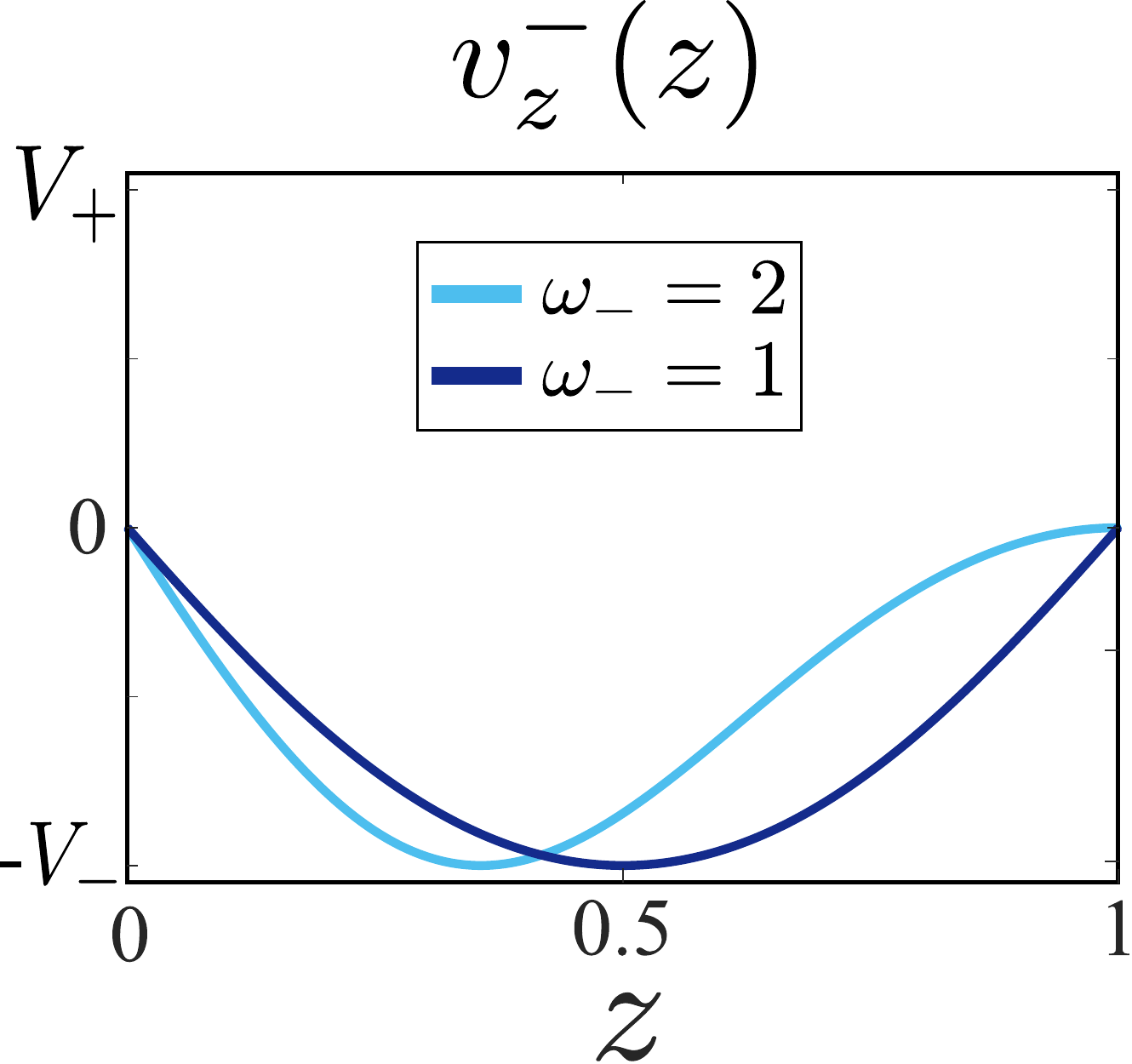}
	\caption{$\xi_-=0.5$}
\end{subfigure}
\caption{Series of sketches showing how $v_z^+$ and $v_z^-$, as defined by Equations~(\ref{eq_ad_norm}) and~(\ref{eq_ad_hyp}) respectively, change as
the parameters $\xi_\pm$ and $\omega_\pm$ vary.}
\label{vel_prof}
\end{figure}

Combining the above observations, and motivated  in part by recent, similar considerations~\cite{hodgkinson}, we  propose
the following functional forms for the phenotypic drift term, $v_z$:
\begin{subequations}
\begin{align}
v_z(z;c)=v^+_z(z) H_{\epsilon}(c-c_H)-v^-_z(z)H_{\epsilon}(c_H-c),\\[3pt]	
v^+_z(z)= \myfrac[2pt]{V_+}{V^*_+} \tanh\left(\myfrac[3pt]{z^{\omega_+}}{\xi_+}\right)\tanh\left(\myfrac[2pt]{(1-z)}{\xi_+}\right),\label{eq_ad_norm}\\[3pt]
v^-_z(z)=\myfrac[2pt]{V_-}{V^*_-} \tanh\left(\myfrac[2pt]{z}{\xi_-}\right)\tanh\left(\myfrac[3pt]{(1-z)^{\omega_-}}{\xi_-}\right).\label{eq_ad_hyp}
\end{align}\label{eq_ad}
\end{subequations}
where $H_\epsilon$ is a smooth variant of the Heaviside function approaching the latter in the limit of $\epsilon \rightarrow 0$ (i.e., $H_\epsilon(x)= {(1+\tanh(\epsilon^{-1}x))}/{2}$). In Equations~(\ref{eq_ad}), the normalising factors $V_\pm^*$ ensure that $\left(\max_z v^\pm_z\right)/V_\pm=1$ and $V_\pm\,(hr^{-1})$ corresponds to the magnitude of the velocity. Further, by controlling
the advection speed along the stemness axis, $V^{-1}_\pm$ determines the timescales for maturation and de-differentiation. 
The parameters $\xi_\pm$ regulate the slopes of $v_z$ at the boundaries $z=0,1$. 
As shown in Figure \ref{vel_norm1}, when $\xi_\pm \ll 1$, the advection velocity is steep when $z\sim 0,1$ and flatter elsewhere. 
This functional form is similar to that proposed in~\cite{hodgkinson}.
For larger values of $\xi_\pm$, the variation is more gradual, with a single maximum (or minimum)
near $z\sim0.5$ (see Figure \ref{vel_norm2}). The exponents $\omega_\pm$ allow us to tune the symmetry/asymmetry in $v_z$ and also to modulate the flux at the boundaries (see Figure \ref{vel_prof}). For example, if $\omega_+=2$,  then $v(0)=\partial_z v(0)=0$ which means that CSCs will be less likely to differentiate compared to the case $\omega_+=1$.
In the absence of experimental data with which to specify the parameters in the phenotypic drift velocity, we consider combinations of the following parameter sets: 
\begin{itemize}
    \item $V_\pm \in \left\{2,4,8\right\} \times 10^{-4} \left[hr^{-1}\right]$, \item $\xi_\pm\in\left\{0.05,0.1,0.5\right\}$, and 
    \item $\omega_\pm\in\left\{1,2\right\}$.
\end{itemize}

In summary, our phenotype-structured model for the growth and response to radiotherapy of a solid tumour is defined by Equations~(\ref{mixedmodel})-(\ref{eq_ad}). A list of the model parameters and estimates of their values can be found in Table~\ref{param_set} in~\ref{AppendixA}. 
\section{Population Dynamics in the Absence of Treatment}
\label{notreatment}

\noindent In this section, we present numerical solutions of  Equations~(\ref{mixedmodel})-(\ref{apoptosis}) and~(\ref{eq_ad}) showing how, in the absence of treatment, the tumour cell distribution along the stemness axis evolves under normoxia and hypoxia. Our numerical solutions are generated using the method of lines, with discretisation performed in the $z$-direction. In more detail and following \cite{Gerisch2006}, we use a finite volume scheme, opting for a Koren limiter to control the advection component of the structural flux. In this way, we reduce~(\ref{mixedmodel}) to a system of time-dependent, ordinary differential equations which can be solved in MATLAB using \textit{ode15s}, an adaptive solver for stiff equations.
The numerical simulations are validated in Section \ref{LSA} where we perform a linear stability analysis. The associated eigenvalue problem is solved numerically using MATLAB's \textit{chebfun} package \cite{Driscoll2014}.

\subsection{Normoxic Conditions}
\label{oxygen_env}
\noindent In well-oxygenated environments, the advection velocity is positive and cells are driven towards a terminally differentiated phenotype, with $z=1$. Depending on the balance between the advective flux and cell renewal (i.e., Darwinian selection and Lamarckian induction), the model predicts a variety of long-time behaviours: the system relaxes to its steady state via damped fluctuations or monotonically. We start by considering symmetric velocity profiles (see Figure \ref{vel_norm1}). 
As summarised in Figures \ref{norm_sym} and  \ref{steady_state}, as the magnitude of the advection velocity, $V_+$, and its steepness, $\xi_+$, are varied, the system exhibits different long time behaviours, even though the dynamics at early times are similar for all parameter sets considered (see Figure \ref{norm_sym}). If simulations are initialised with a small population of  cells with $z\sim 0.5$, then the dynamics are initially dominated by proliferation. Over time, as $\phi$ increases, competition slows the cell proliferation rate and phenotypic advection becomes more important. As the cells mature, they accumulate near $z=1$, and the rate of natural cell death exceeds the rate of cell proliferation. From this time onwards, the growth curves corresponding to different parameter sets start to deviate. 

\begin{figure}[h!]
	\centering
	\includegraphics[width=\textwidth]{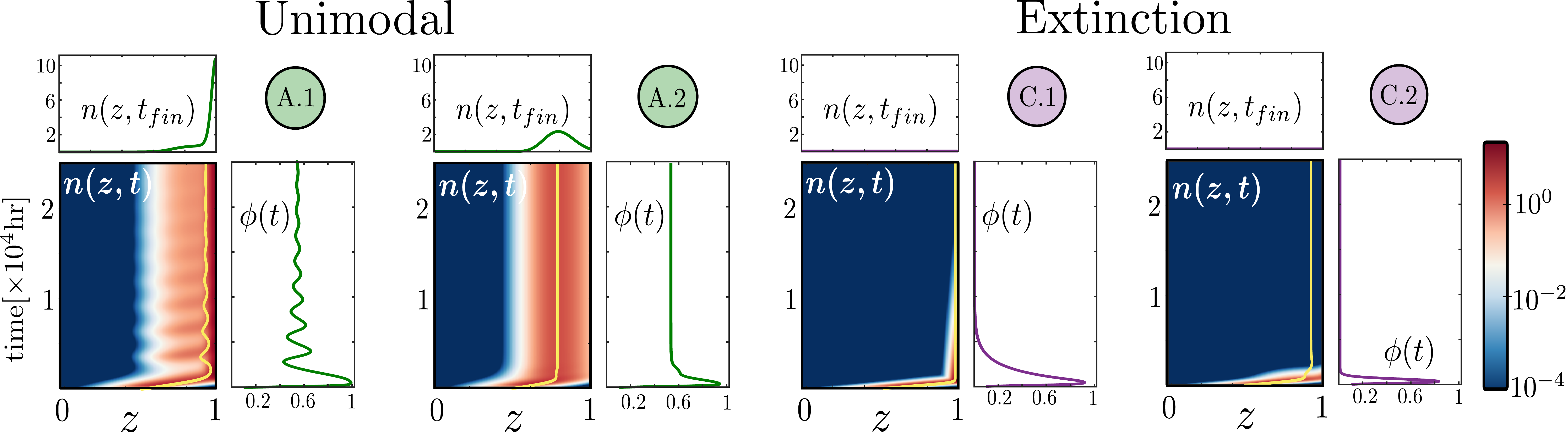}
	\caption{
	Results from a series of numerical simulations of 
	Equations~(\ref{mixedmodel})-(\ref{apoptosis}) and~(\ref{eq_ad}), showing how the cell distribution, $n(z,t)$, the phenotypic mean, $\mu(t)$, and the cell density, $\phi(t)$, change over time when we use a symmetric velocity profile (i.e., $\omega_+=1$ in Equation~(\ref{eq_ad})). 
	As $V_+$ increases and $\xi_+$ decreases, the system can be driven to extinction. 
	See Figure \ref{steady_state} for the values of the other model parameters.}
	\label{norm_sym}
\end{figure}

For example, in case A.2, the system rapidly relaxes to a non-zero steady state distribution characterised by cells with medium clonogenic capacity (i.e., a mix of highly proliferating and terminally differentiated cells or TDCs). 
Similarly, for cases C.1 and C.2, the cell density, $\phi(t)$, decays exponentially to extinction at a rate dictated by $d_f$. In other parameter regimes, the relaxation phase is characterised by damped fluctuations. In case A.1, for example, fluctuations are driven by the interplay between apoptosis, competition and advection. As TDCs are eliminated, the reduction in competition allows re-growth of highly proliferative cancer cells (i.e., $z\sim0.55$). 
As these cells proliferate, competition slows growth and advection becomes dominant, resulting in the alternating pattern of red and white stripes observed in the surface plot for $n(z,t)$ shown in Figure \ref{norm_sym} for case A.1. Over time, the fluctuations decay and the system relaxes to its steady state distribution. 
In Section~\ref{LSA}, we present a complementary 
investigation of this behavior, relating the damped oscillations 
to a complex eigenvalue in the linearisation about the
equilibrium solution. 

\begin{figure}[h!]
	\centering
	\includegraphics[width=0.9\textwidth]{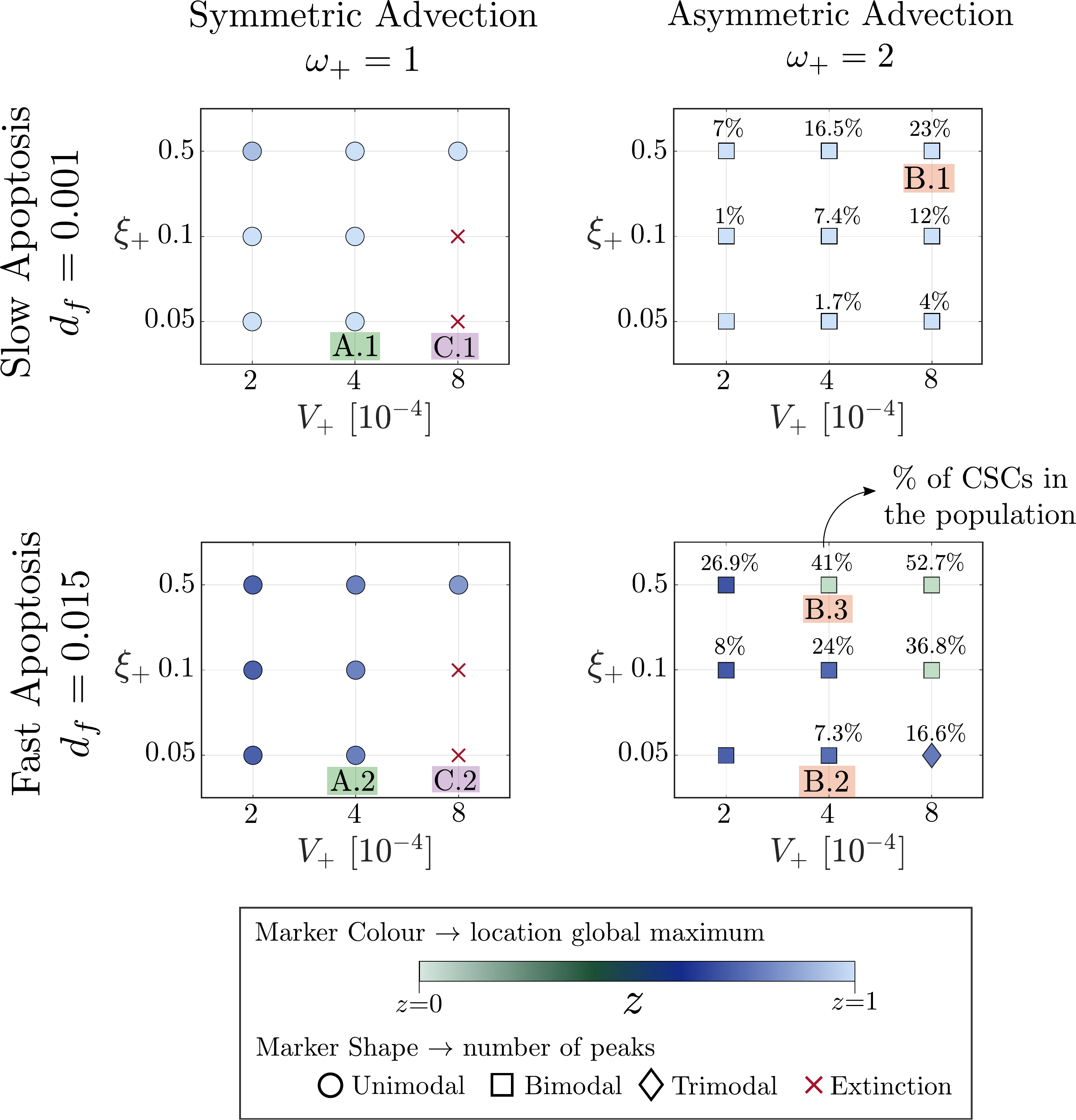}
	\caption{Series of phase diagrams characterising the steady state distribution predicted by the model as properties of the advection velocity, $v_z$, vary
		(i.e., for different values of the parameters $V_+$, $\xi_+$ and $\omega_+$), and the rate of apoptosis, $d_f$. At each point in $(V_+,\xi_+)$ parameter space, 
		we characterise the equilibrium distribution based on the number of peaks and the dominant phenotype (i.e., the $z$-locations of the local maxima) for different values of the parameters $\omega_+$ and $d_f$. For parameter sets that give rise to a significant fraction of CSCs 
		(i.e., $\%$ CSCs $\geq 1\%$), the value of $\phi_{CSC}(0.3,t_{\infty})$, as defined by Equation~(\ref{cum_dist}), is also indicated.}
	\label{steady_state}
\end{figure}

Focusing on the long time behaviour, the symmetric advective profile gives rise to a population with a unimodal equilibrium distribution where the location of the peak is dictated by the values of the other parameters. For example, for small values of the maximum death rate, $d_f$ (see case A.1), the distribution is skewed towards $z=1$, while for higher values of $d_f$ the peak is shifted towards the centre of the domain. These observations are summarized in Figure \ref{steady_state}, where we have further analysed how the properties of the equilibrium distribution depend on other parameters in the model. We note that as the advective velocity increases (i.e., larger $V_+$)  the value of $\xi_+$ determines whether total extinction occurs. This suggests that there is a bifurcation 
as $V_+$ and $\xi_+$ vary, with the system transitioning from a trivial to a non-zero steady state (this behaviour will be
investigated in Section~\ref{LSA}).

By contrast, the equilibrium distribution for an asymmetric velocity profile (i.e., $\omega_+=2$, as in Figure \ref{vel_norm2}), has a multimodal distribution, typically characterised by two peaks. In this case, since the CSCs have a lower propensity to mature, they accumulate and persist in the population, even under normoxia. The second column of Figure \ref{steady_state} shows that the proportion of CSCs at long time increases as the death rate, $d_f$, the steepness parameter, $\xi_+$, and the maturation velocity, $V_+$, increase, until the CSCs become the dominant subpopulation (see, for example, Case B.3 in Figure \ref{norm_asym}). 
Varying the death rate, $d_f$, does not significantly affect whether extinction occurs; rather, it determines the location of the maximum peak in the equilibrium distribution (see, for example, case B.2 in Figure \ref{norm_asym}). For low death rates, cells are predominantly 
in a terminally differentiated state. As the death rate increases, the peak moves to the left, producing an equilibrium distribution in which a higher proportion of rapidly proliferating cells balances the high death rate.
Figure \ref{norm_asym} shows how the system relaxes to its steady state when $\omega_+=2$. Comparison with Figure \ref{norm_sym} reveals that in this case the dynamics are characterised by secondary regrowth, driven by the accumulation of CSCs. For example, in case B.1, phenotypic diffusion enables the cancer cells to de-differentiate, acquire a stem-like phenotype and, therefore, contribute to population growth.

\begin{figure}[h!]
	\centering
	\includegraphics[width=\textwidth]{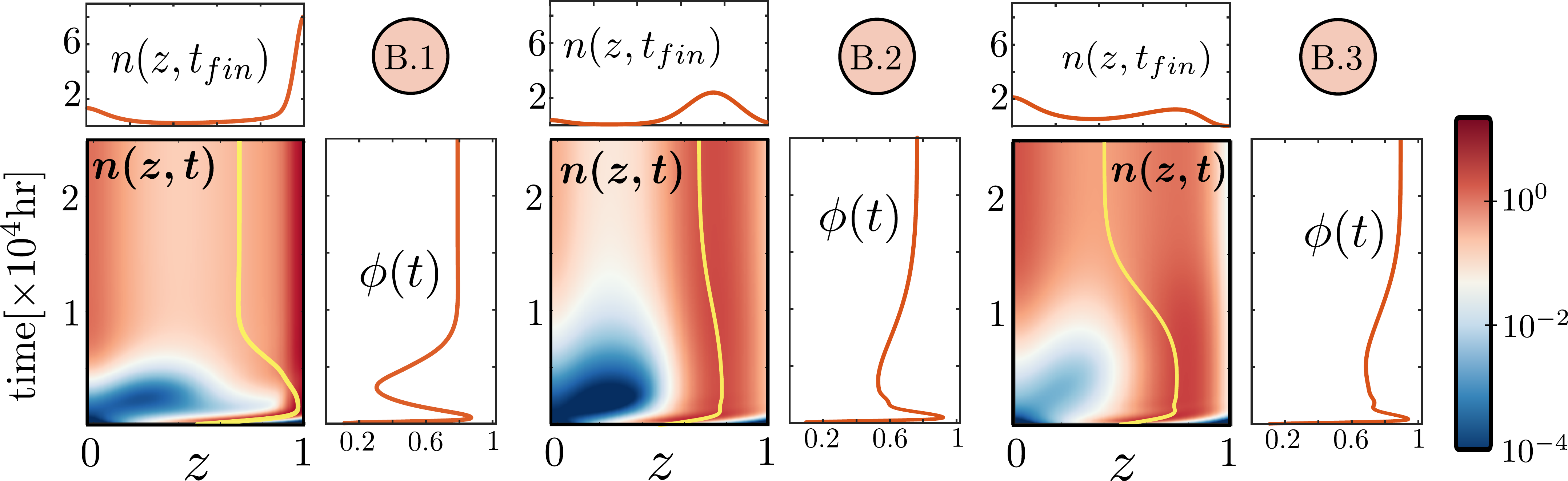}
	\caption{Results from a series of numerical simulations of 
	Equations~(\ref{mixedmodel})-(\ref{apoptosis}) and~(\ref{eq_ad}), showing how the cell distribution, $n(z,t)$, the phenotypic mean, $\mu(t)$, and the cell density, $\phi(t)$, change over time. For these results, we use an asymmetric velocity profile (i.e., $\omega_+=2$ in Equation~(\ref{eq_ad})).  See Figure \ref{steady_state} for the values of the other model parameters.}
	\label{norm_asym}
\end{figure}

To summarise, the properties of the advection velocity $v_z$, determine whether the model predicts extinction or persistence of CSCs, regardless of whether they are present initially. When $\omega_+=2$, random mutations (i.e., diffusion), may dominate the advective force near $z=0$, allowing CSCs first to form, then to proliferate and ultimately to comprise a significant proportion of the equilibrium population. CSCs have been observed in normoxic regions; for example, they have been found in perivascular tumour regions, where endothelial cells secrete factors that inhibit CSC maturation \cite{Calabrese2007}. By contrast, when $\omega_+=1$
(i.e., for symmetric velocity profiles), all cells mature over time, leading to the eventual extinction of CSCs. This behaviour could describe that of tumours which lack CSCs, or the effect of drugs which induce stem cell differentiation and, thereby, reduce the incidence of resistance to other treatments, such as radiotherapy. 
We conclude that targetting $V_+$ and $\xi_+$ may be effective for eliminating CSCs, increasing tumour sensitivity to treatment and, in certain scenarios, driving tumour extinction. 

\subsection{Hypoxic Conditions}
\label{hyp cond}
\noindent Under hypoxia, the advection velocity in our model is negative and cells will be driven to de-differentiate. In this case, the equilibrium distribution is unimodal, with the dominant phenotype at $z=0$. Although varying the death rate $d_f$ does not effect the equilibrium distribution (compare cases H3 and H4 in Figure \ref{hyp1}), the values of $\omega_-$ and $\xi$ influence the width of the peak 
(compare cases H1 and H2 in Figure \ref{hyp1}) and, therefore, the variability in the population.

\begin{figure}[h!]
\begin{subfigure}{\textwidth}
	\includegraphics[width=\textwidth]{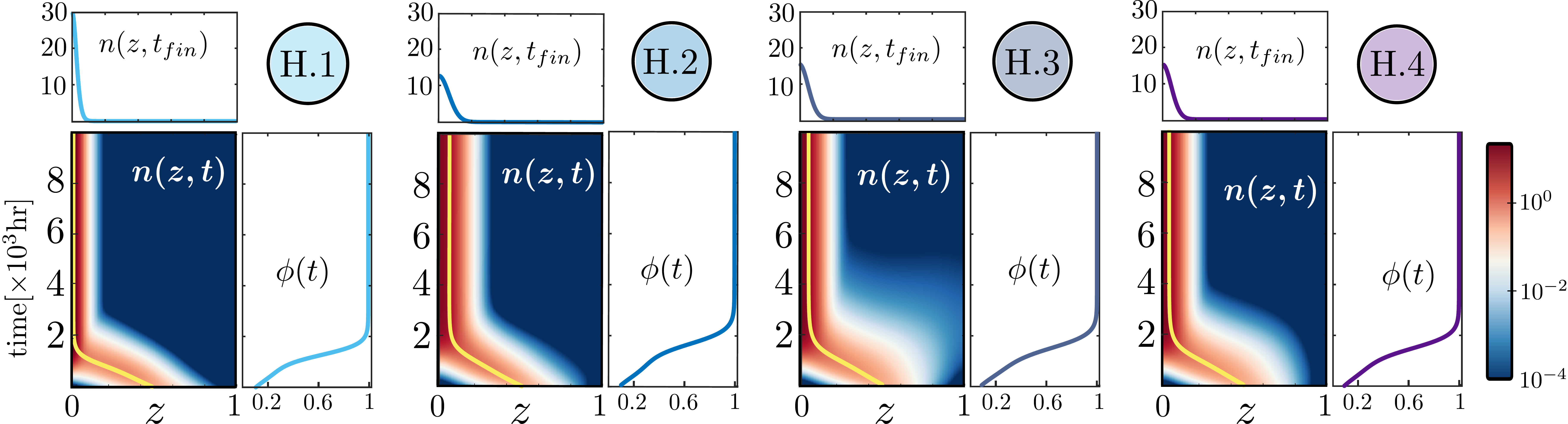}
	\caption{}
	\label{hyp1_cond1}
\end{subfigure}
\begin{subfigure}{\textwidth}
	\centering
	\includegraphics[width=\textwidth]{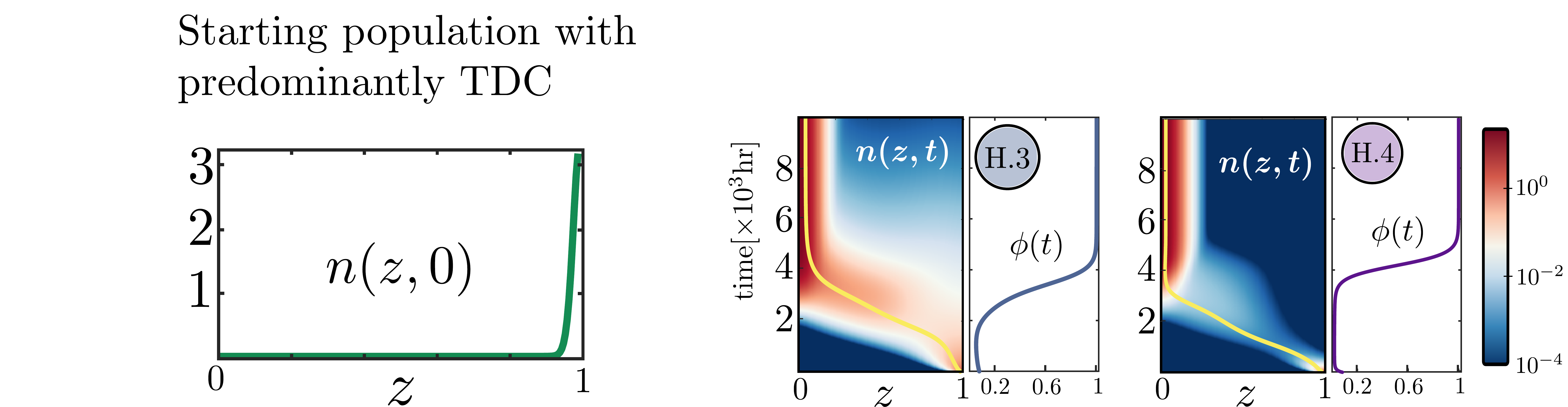}
	\caption{}
	\label{hyp1_cond2}
\end{subfigure}
\caption{Numerical results under hypoxic condition for four parameter sets, all with $V_-$= $4\times 10^{-4}$.
	In (a) we use the standard initial condition defined by Equation~(\ref{initial_cond}) while in (b) the population is centred around $z=1$. The other parameter values are as follows: 
	(H1) $\xi_-$= 0.05, $\omega_-$= 1 and $d_f$= 0.001; 
	(H2)  $\xi_-\mbox{= }0.5$, $\omega_-\mbox{= }1$ and $d_f\mbox{= }0.001$; 
	(H3) $\xi_-\mbox{= }0.5$, $\omega_-\mbox{= }2$ and $d_f\mbox{= }0.001$; 
	(H4) $\xi_-\mbox{= }0.5$, $\omega_-\mbox{= }2$ and $d_f\mbox{= }0.015$. }
\label{hyp1}
\end{figure}

Differences in the system dynamics also arise as the initial conditions $n_0(z)$ vary.
The results in Figure \ref{hyp1_cond1} indicate little variation in the system dynamics
when the initial conditions from Section \ref{oxygen_env} are used. By contrast, in
Figure \ref{hyp1_cond2} we observe marked differences when the initial conditions are centred around the TDCs. 
In this case, population regrowth is delayed, the delay depending on the choice of parameter values. 
For example, when $\omega_-=2$, the velocity in a neighbourhood of $z=1$ is smaller than when
$\omega_-=1$. Consequently, cells de-differentiate more slowly, delaying tumour regrowth. 
Similarly, increasing the death rate, $d_f$, reduces the number of cells that can de-differentiate and, subsequently, delays regrowth. Therefore, while $d_f$ does not affect the equilibrium distribution, it influences the system dynamics.
These results show how the formation of hypoxic regions can shape the development of a tumour. In particular, the emergence of hypoxia maintains and enhances the
pool of CSCs, preventing population extinction (see, for example, scenario D in Section \ref{oxygen_env}).

\subsection{Linear Stability Analysis}
\label{LSA}

\noindent We now validate some of the above numerical results by performing a linear stability analysis which enables us to characterise the equilibrium states. We denote by $\bar{n}=\bar{n}(z)$ a steady state for the (untreated) system~(\ref{mixedmodel})-(\ref{eq_ad}), with a total cell density $\bar{\phi}=\int_0^1 \bar{n}(z) dz$ and let $\delta n$ represent a small perturbation to this solution. Then we can approximate the solution $n$ in a neighbourhood of $\bar{n}$ as:
\begin{equation}
n(z,t)= \bar{n} + \delta n(z,t),\quad \|\delta n\| \ll 1
\quad \,\forall t>0.
\end{equation}
Substituting this ansatz into~(\ref{mixedmodel}) and retaining linear terms, we obtain the following equation for $\delta n$:
\begin{subequations}
	\begin{align}
\begin{aligned}
\frac{\partial \delta n}{\partial t}=
\mathcal{M}\delta n,
\end{aligned}
\label{defM}
\\[2pt]
\frac{\partial \delta n}{\partial z} = 0, \qquad z = 0, 1,\\
\delta n(z,0)\equiv 0,\label{neu}
\end{align}
\end{subequations}
where $\mathcal{M}$ is the following integro-differential operator
\begin{equation}
 \mathcal{M}\delta n \equiv \frac{\partial }{\partial z} \left(\theta \frac{\partial \delta n}{\partial z}-  v_z\delta n\right)+ \left[p\left(1-\bar{\phi}\right) - f \right] \delta n  - p\bar{n}\int_0^1 \delta n dz. \label{Mdef}
\end{equation}
The solution $\bar{n}$ is \textit{spectrally}
stable if the spectrum of the operator, $\sigma(\mathcal{M})$, does
not contain eigenvalues with positive real part, i.e.,
\begin{equation}
\sigma(\mathcal{M}) \bigcap \left\{\lambda \in \mathbb{C}: \Re(\lambda)>0\right\}=\emptyset.
\end{equation}
Moreover, the dynamics of the system will be dominated by the fastest growing mode (i.e., the eigenfunction corresponding to the eigenvalue with the largest real part, $\lambda_0$). 

In \ref{AppendixB} we transform the above eigen-problem so that it does not include any first order derivatives. For a non-zero steady state, we retain a non-local term in the eigenvalue problem and this can give rise to a spectrum with a pair of complex eigenvalues. Recalling case A.1 from Section \ref{oxygen_env} (see Figure \ref{norm_sym}), the numerically estimated value of $\lambda_0$ is indeed complex ($\lambda_0=-1.535 \times 10^{-4} \pm i  \, 2.24\times 10^{-3}$, where $i^2 = -1$). This result, in turn, explains why damped fluctuations are observed in the numerical simulations.

By contrast, when considering the trivial steady state, $\bar{n}\equiv 0$, which is always a fixed point for the system, the non-local term vanishes and we obtain the standard form analysed by 
Sturm-Liouville theory. Using well known results, we can identify sufficient conditions for the stability/instability of the trivial steady state (see Lemmas \ref{lemma1}-\ref{lemma3} in \ref{AppendixB}). Under hypoxia, where $v_z<0$,
we find that the trivial steady state is unstable (for the parameter sets in Table \ref{param_set}) and the system evolves to a non-zero distribution, which is consistent with the numerical results from Section \ref{hyp cond}. We note that the results relate only to the behaviour of the fitness function and advection velocity near the boundary $z=0$, suggesting that the most relevant parameters are $p_0^{max}$, $V_-$, $\theta$ and $\xi_-$.
By contrast, under normoxia, and for the range of parameter considered here, the system undergoes a bifurcation. For sufficiently small $V_+$, the trivial steady state is unstable; for sufficiently large $V_+$ and for large values of the death rate, $d_f$, the trivial steady state is stable, (see, for example, case C2 in Section~\ref{oxygen_env}). To investigate other parameter regimes that we can not tackle analytically, we rely on numerical estimation of the largest eigenvalue, $\lambda_0$. As shown in Figure \ref{xi_crit}, it is possible to identify the boundary of the region of stability in $(\xi_+,V_+)$ space. This
diagram does not change significantly as the death rate varies in the range from $d_f=0.001$ to $d_f=0.015$ (results not shown).
However, the results are highly sensitive to the value of $\omega_+$. Comparing Figures \ref{xi_crit1} and \ref{xi_crit2}, we see that setting $\omega_+=2$ favours the formation of a non-trivial equilibrium distribution, with the curve shifting to the far right of the parameter space (i.e., small values of $\xi_+$ and large
values of $V_+$). In the latter case, this implies that even higher velocities $V_+$ are needed to stabilise the tumour elimination solution. This is consistent with the numerical results in Section~\ref{oxygen_env}, where setting $\omega_+=2$ (see scenario B in Section~\ref{oxygen_env}) favoured the accumulation of CSCs which acted a reservoir for tumour cells. 

\begin{figure}
	\centering
	\begin{subfigure}{0.49\textwidth}
		\centering
		\includegraphics[width=\textwidth]{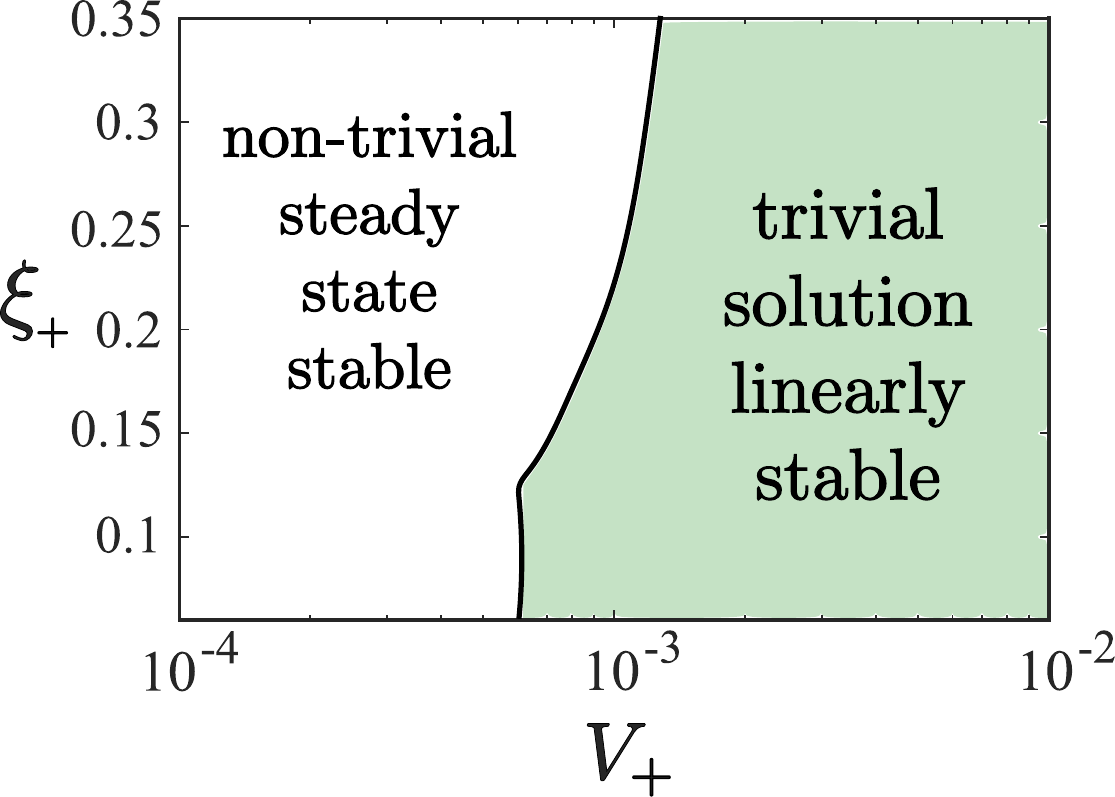}
		\caption{$\omega_+=1$}
		\label{xi_crit1}
	\end{subfigure}
	\begin{subfigure}{0.49\textwidth}
		\centering
		\includegraphics[width=\textwidth]{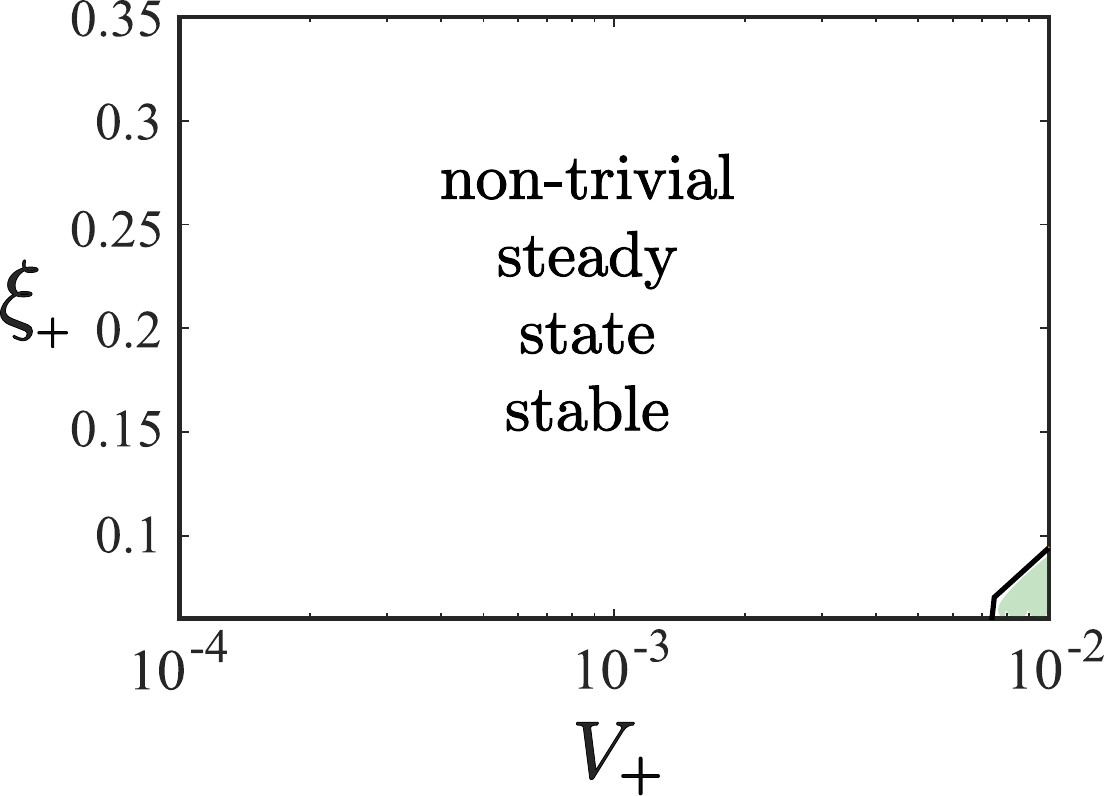}
		\caption{$\omega_+=2$}
		\label{xi_crit2}
	\end{subfigure}	
	\caption{Series of phase diagrams partitioning the $(V_+,\xi_+)$ parameter space into regions where the trivial steady state is linearly stable (green regions) and unstable (white regions). 
		The diagrams are obtained for $d_f=0.001$. 
		We note that changing $\omega_+$ has a significant impact on the size of the region of $(V_+,\xi_+)$ parameter space in which the non-trivial steady state is stable (compare (a) and (b)).}
	\label{xi_crit}
\end{figure}

\section{Population Dynamics in the Presence of Treatment}
\label{radio_result}

\noindent In the previous section, we found that the system possesses a 
stable steady state to which the dynamics converge
for the range of parameter values considered. 
Therefore, we anticipate that, while treatment can perturb the system from its equilibrium, it will eventually relax to its stable steady state once treatment ends. Thus we expect extinction to occur for parameter values lying in the stability region of the trivial steady state (see Figure \ref{xi_crit}).
From this point of view, we are interested in understanding how different environmental conditions (i.e. normoxia and hypoxia), different treatment protocols and different tumour compositions affect the relaxation phase and, in particular, the time to recurrence. 

To account for variability in tumour responses, we consider the different advection velocities used in our earlier analysis (see Table \ref{tab_rad3}). 
Starting from the initial condition~(\ref{initial_cond}), cells follow different pre-treatment protocols as specified in Table \ref{tab_rad3}.  Without loss of generality, we shift time so that $t=0$ corresponds to $24$ hours before treatment begins. While attention will focus on tumour responses in constant environmental conditions, we also consider briefly treatment responses in changing environments.
For each scenario, we simulate the response to treatment for the range of values of the radiation parameters listed in Table \ref{rad_tab2}. 
We denote by $n^{(S1,R1)}(z,t)$ the solutions corresponding to scenario $S1$ from Tables \ref{tab_rad3} and radio-sensitivity parameter set $R1$ from Table \ref{rad_tab2}.
\begin{table}
	\hspace{-4mm}
	\begin{tabular}{c c c c}
		\toprule[1.5pt]\addlinespace[2pt]
		Scenario  & Protocol &Parameters&	Subsection\\[2pt]
		\hline\addlinespace[6pt]	
		\multirow{4}{*}{$S1$} &\multirow{2}{*}{\parbox[c]{120pt}{\includegraphics[width=0.3\textwidth]{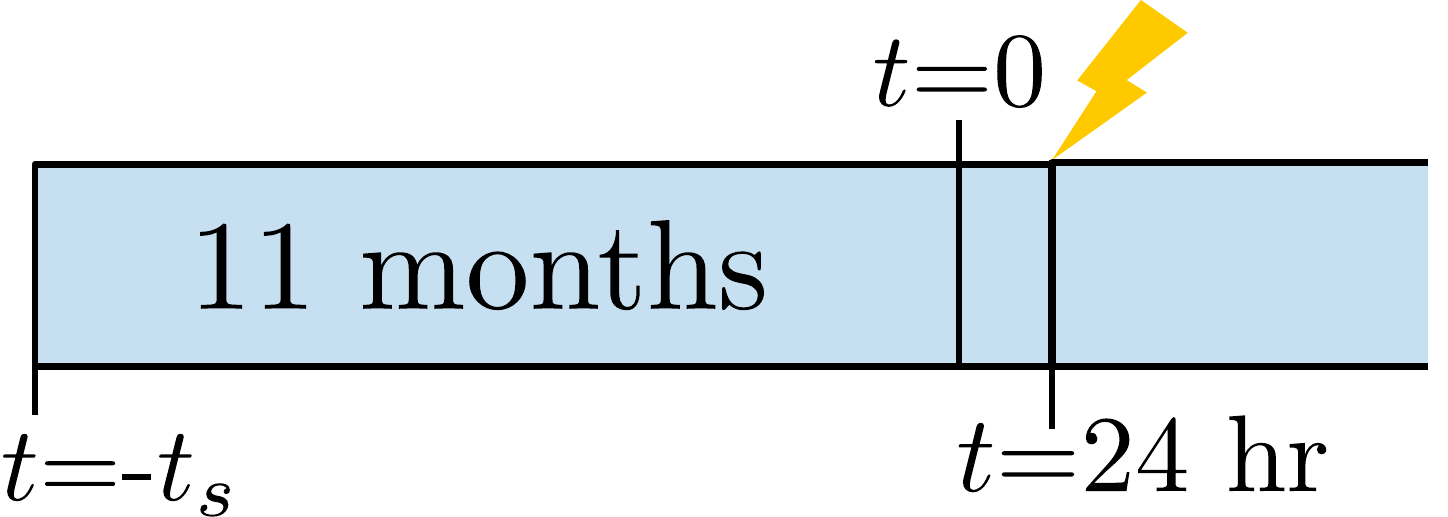}}} &&\multirow{12}{*}{\ref{radio_norm}} \\
		&&$(V_+[10^{-4}],\xi_+,\omega_+,d_f)$&\\
		& &  =$\left(4,0.05,2,0.015\right)$&\\
		&&&\\[2pt]
		\multirow{4}{*}{$S2$}&\multirow{2}{*}{\parbox[c]{120pt}{\includegraphics[width=0.3\textwidth]{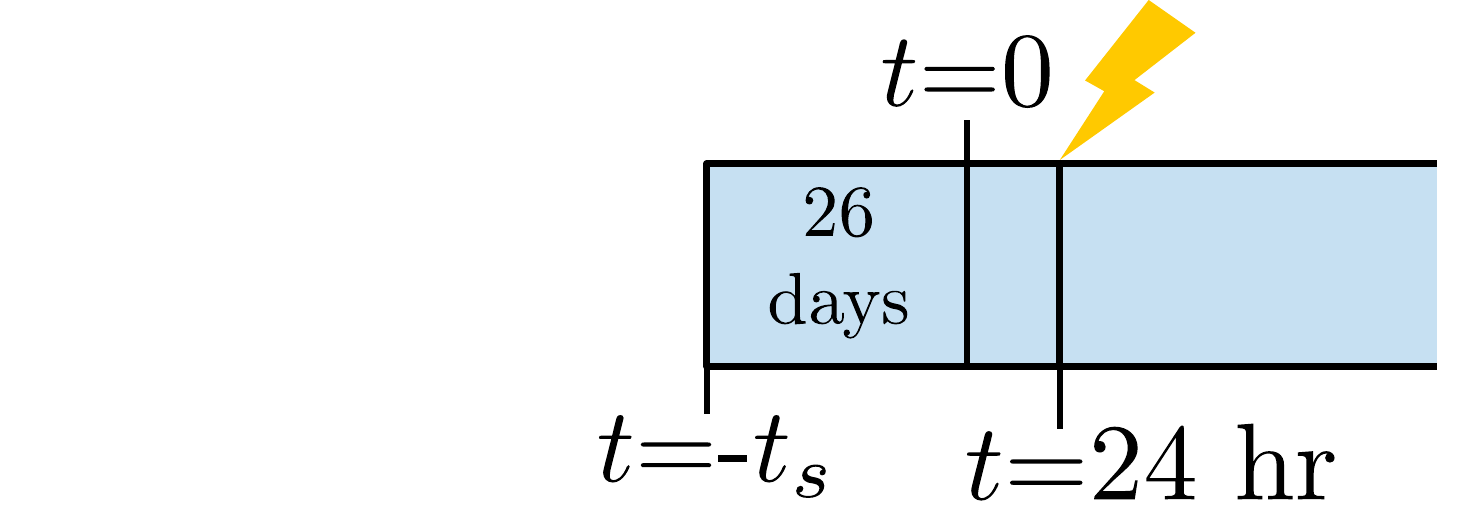}}}&&\\
		&&$(V_+[10^{-4}],\xi_+,\omega_+,d_f)$&\\
		&&=$\left(8,0.05,1,0.001\right)$&\\
		&&&\\[2pt]
		\multirow{4}{*}{$S3$}& \multirow{2}{*}{\parbox[c]{120pt}{\includegraphics[width=0.3\textwidth]{figures/radio/growthcurves/protocols/protocols2}}}&&\\
		&&$(V_+[10^{-4}],\xi_+,\omega_+,d_f)$&\\
		&&=$\left(8,0.05,2,0.001\right)$&\\
		&&&\\[2pt]
			\hline\addlinespace[4pt]
		\multirow{4}{*}{$S4$}&\multirow{2}{*}{\parbox[c]{120pt}{\includegraphics[width=0.3\textwidth]{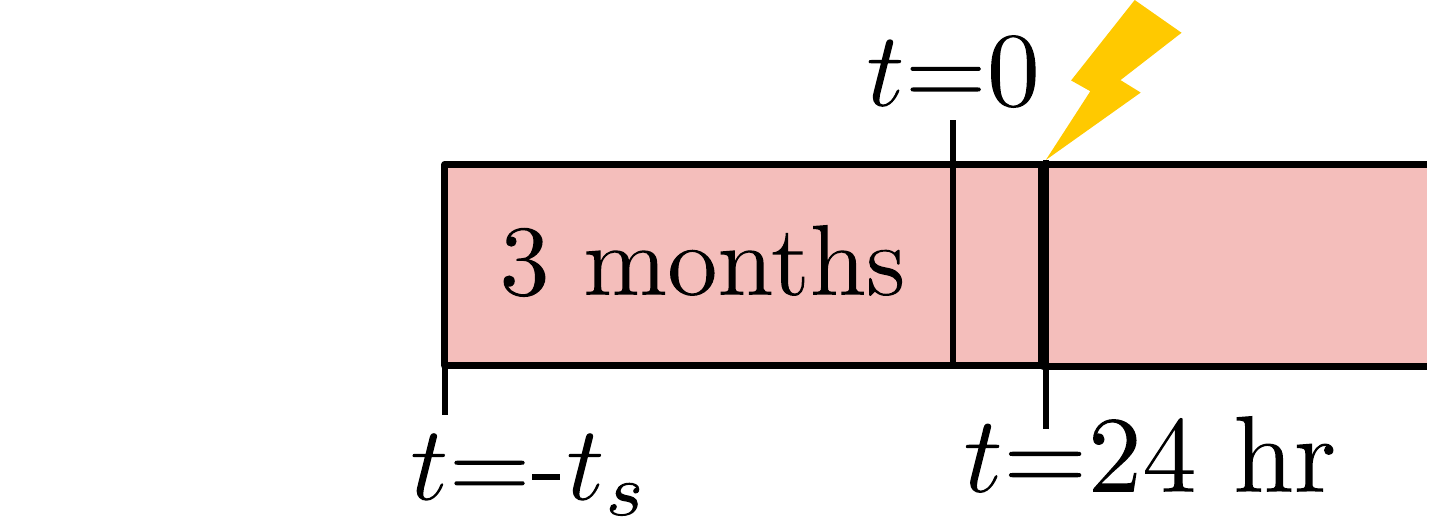}}}&&\multirow{4}{*}{\ref{radio_hyp}}\\
		&&$(V_-[10^{-4}],\xi_-,\omega_-,d_f)$&\\
		&&=$(2,0.5,2,0.001)$&\\
		&&&\\[2pt]
			\hline\addlinespace[8pt]
		\multirow{4}{*}{$S5$} &\multirow{2}{*}{\parbox[c]{120pt}{\includegraphics[width=0.3\textwidth]{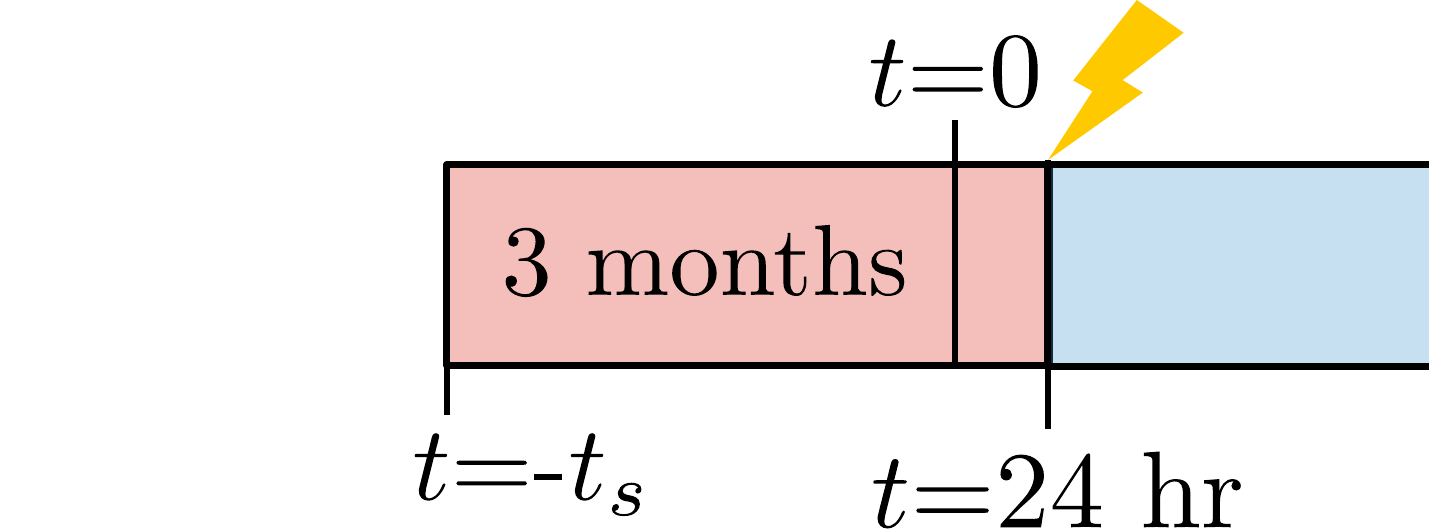}}}&&\multirow{8}{*}{\ref{radio+change}}\\
		&&$(V_\pm[10^{-4}],\xi_+,\xi_-,\omega_+,\omega_-,d_f)$&\\
			&&=$(8,0.05,0.5,1,2,0.001)$&\\
		&&&\\[2pt]
		\multirow{4}{*}{$S6$} &\multirow{2}{*}{\parbox[c]{120pt}{\includegraphics[width=0.3\textwidth]{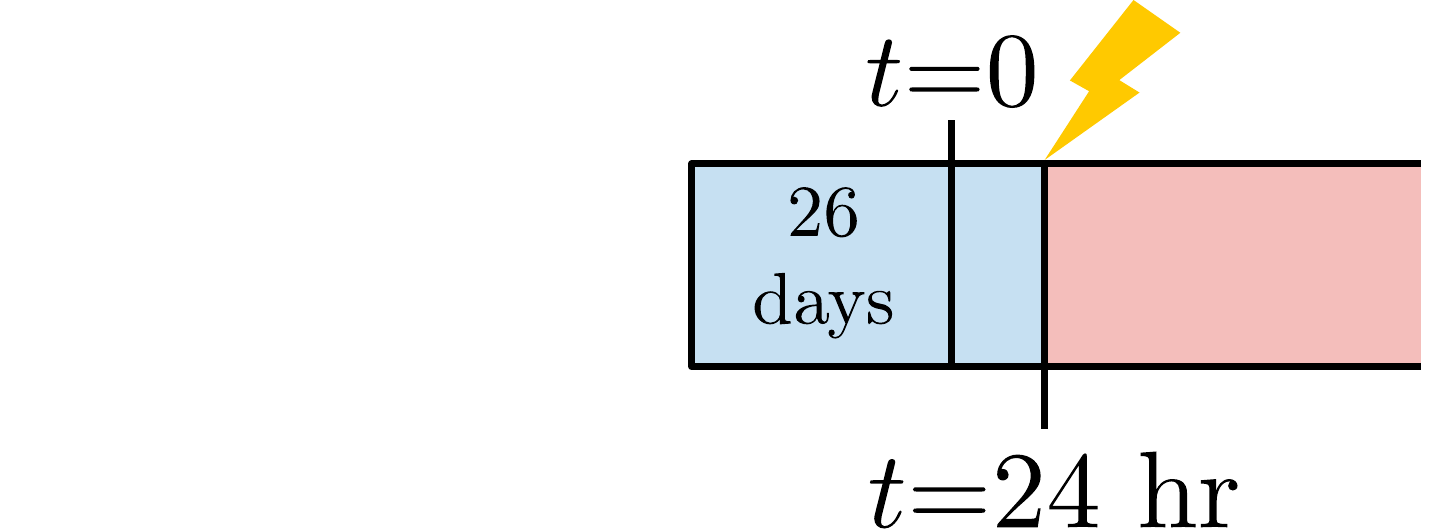}}}&&\\
			&&$(V_\pm[10^{-4}],\xi_+,\xi_-,\omega_+,\omega_-,d_f)$&\\
		&&=$(8,0.05,0.5,1,2,0.001)$&\\[2pt]
		&&&\\
		\bottomrule[1.5pt]
	\end{tabular}
	\vspace{2mm}
	\caption{Parameter sets used to generate the numerical simulations in Section~\ref{fit}, together with the corresponding environmental conditions pre- and post-treatment (blue: normoxia, red: hypoxia). Simulations are initialised using equation \ref{initial_cond} at different times $t=-t_s$ as indicated in the second column. Radiotherapy is administered at time $t=24$ hours. The parameter values have been chosen to illustrate the range of qualitative behaviours that the model exhibits.}
	\label{tab_rad3}
\end{table}

\subsection{Treatment Response in Normoxic Conditions}
\label{radio_norm}
\noindent The simulation results presented in Figure \ref{radio_norm_growth} illustrate the different regrowth dynamics that can arise when well-oxygenated tumour cells are exposed to a single dose of RT. We identify three distinct behaviours: instantaneous regrowth (S1), decay and extinction (S2) and initial remission with subsequent regrowth (S3). While the cell survival fraction immediately post-treatment depends on the parameter values used in the LQ-model (see Equation~(\ref{radiosensitivity})), the qualitative population regrowth dynamics post-treatment do not depend on these values. 
 
In more detail, for scenario S1, the cell density increases rapidly after treatment, driving the system towards its (asymptotic) equilibrium. By contrast, for scenarios S2 and S3, the growth curves initially decrease at similar rates until about $40$ days after treatment. Thereafter, for scenario $S3$ the tumour exhibits rapid regrowth to the equilibrium distribution, whereas for scenario $S2$, the tumour continues to shrink, until it is eventually eliminated.  

\begin{figure}[h]
	\centering			
	\includegraphics[width=1\textwidth]{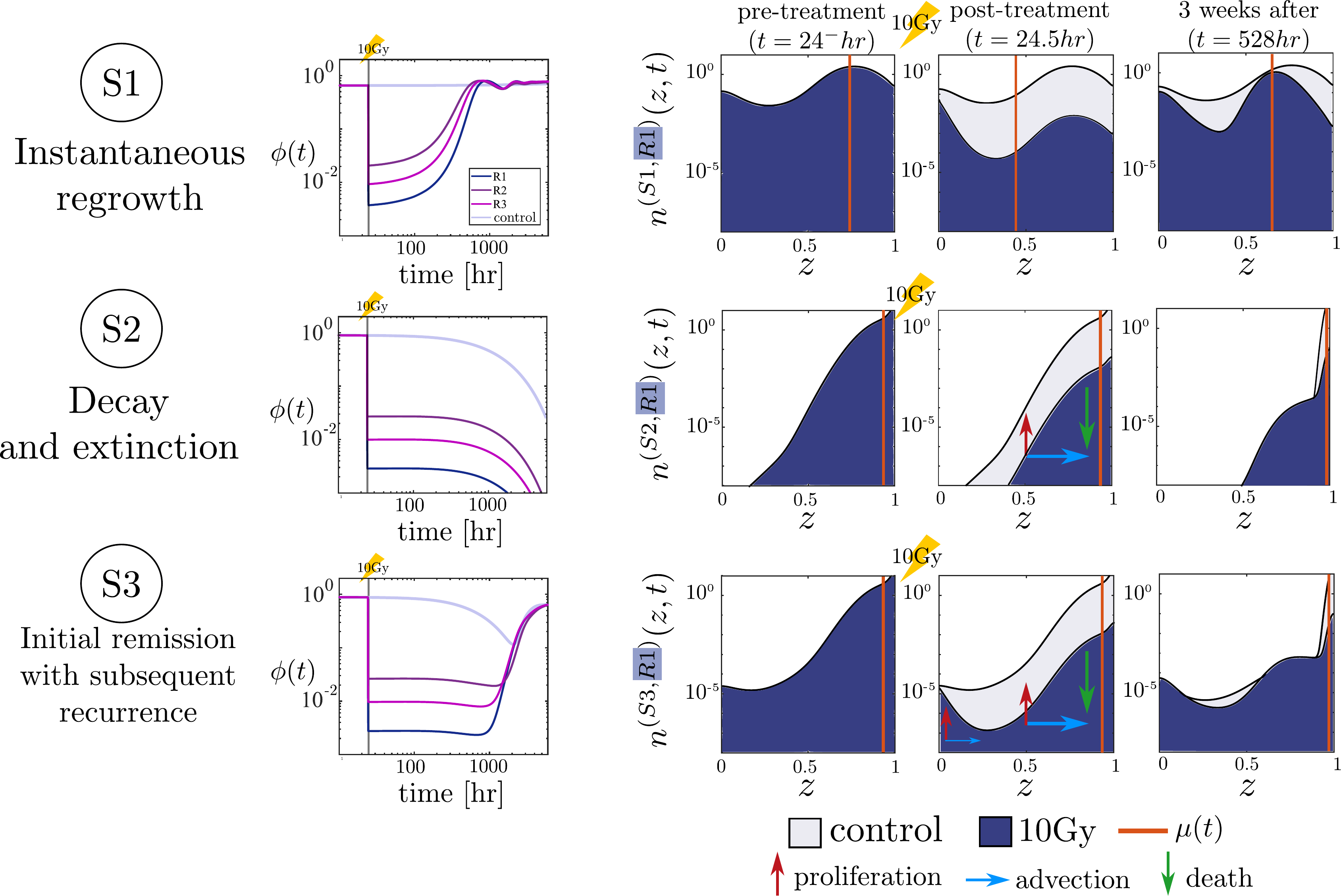}
	\caption{Different treatment outcomes under normoxia. For each scenario S1, S2 and S3 (see Table \ref{tab_rad3}) we consider the dynamics of the total cell number, $\phi(t)$, and compare the responses for the radio-sensitivity parameter sets R1, R2 and R3 (see Table \ref{rad_tab2}) to the control, untreated case. For each scenario we also present plots of the phenotypic cell distribution, $n(z,t)$, at different times for radiotherapy protocol R1. The vertical line indicates the time of irradiation, while a line is also shown that follows the evolution of the control (i.e., in the absence of treatment).}
	\label{radio_norm_growth}
\end{figure} 

The origin of such differences can be understood from the time evolution of $n(z,t)$ post-radiotherapy. Figure \ref{radio_norm_growth} shows that for
case R1 of Table~\ref{rad_tab2}, the balance between cell proliferation and advection drives the system dynamics. 
The reduction in the cell density $\phi(t)$ post-radiotherapy reduces intra-population competition and allows the cells to resume
proliferation. Depending on the magnitude of the advection velocity (which is positive), 
the cells either regrow ($S3$) or they are driven
to a terminally differentiated state and, thereafter, become extinct
($S2$). For scenario $S3$, the
presence of radioresistant CSCs post treatment and a small positive velocity at $z=0$ 
together drive regrowth. As the CSCs start to
mature, there is a continuous source of highly
proliferative cells which, in turn, drive rapid regrowth of the
tumour. As the total cell number increases, intra-population competition slows cell proliferation until eventually advection becomes dominant, driving the cells to de-differentiate. By contrast, for scenario $S2$, advection dominates proliferation along the entire phenotypic axis. Additionally, CSCs are absent so that all cells are rapidly terminally differentiated and, thereafter, undergo cell death.

Comparison of scenarios S2 and S3 reveals how different phenotypic compositions can generate treatment responses which are initially qualitatively similar, but differ markedly at long times. 
This finding is reinforced in Figure  \ref{radio_mean_norm} where we plot the mean phenotypes, $\mu = \mu(t)$, as defined by Equation~(\ref{mean}). 
For scenarios S2 and S3, the dynamics of the mean phenotype are indistinguishable at short times and do not start to diverge until approximately 20 days after treatment. 

\begin{figure}[h]
	\includegraphics[width=0.95\textwidth]{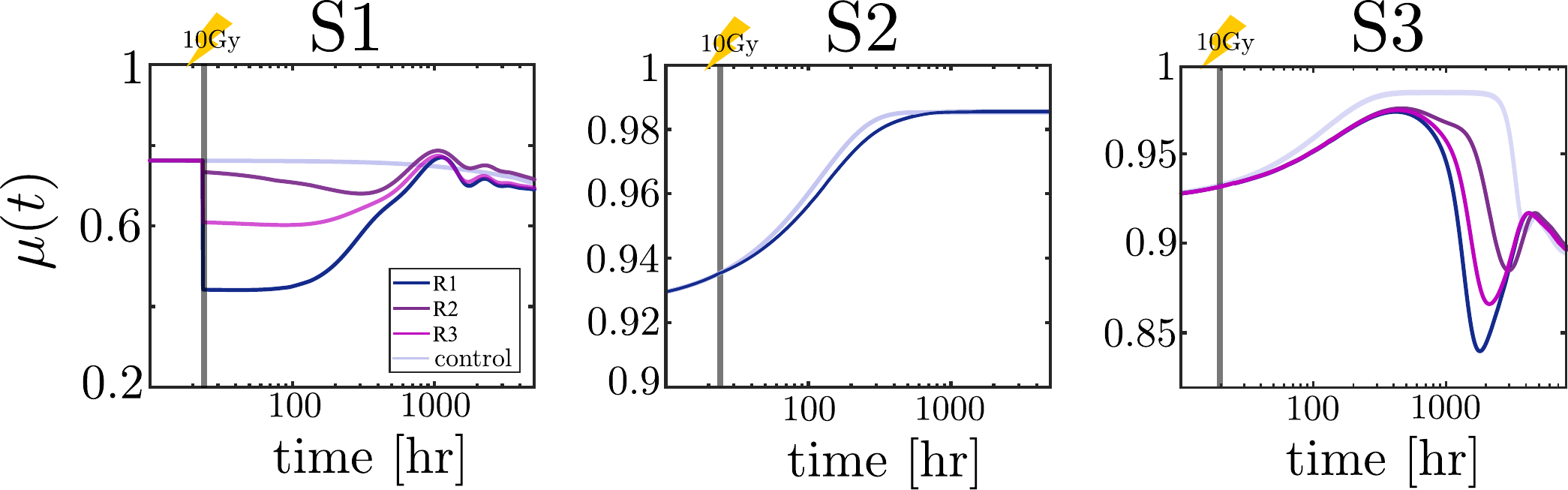}
	\caption{Series of plots showing the evolution of the phenotypic mean, $\mu(t)$, for scenarios S1, S2 and S3 (see Figure \ref{radio_norm_growth}). We note that the scales used on the vertical axes are different.}
	\label{radio_mean_norm}
\end{figure}

More generally, the results presented in Figure \ref{radio_mean_norm} reveal three characteristic behaviours for the evolution of the phenotypic mean following radiotherapy. The dynamics of $\mu$ may be the same as those prior to treatment, with negligible deviation from the control (see scenario S2). A discontinuity in $\mu$ may be induced by radiotherapy (see scenario $S1$). In this case, CSCs comprise a significant proportion of the population prior to RT and the effect of radioresistance is pronounced (see Figure \ref{radio_norm_growth}). As CSCs
are more likely to survive radiotherapy than more mature cells, we observe an ``instantaneous'' shift in $\mu$ towards
less mature phenotypes.
The size of the discontinuity depends on the relative sensitivity of CSCs and TDCs to RT, or, using the terminology introduced in Section \ref{fit}, the selective power of RT. Since we are considering high radiation dosages, the discontinuity is determined by the ratio $\beta_{min}/\beta_{max}$. In order for the selective pressure of treatment to be apparent, CSCs must comprise a significant fraction of the population prior to treatment. This explains why, for scenario S3, 
there is an initial transient period during which, as for 
scenario S2, there is no discernible deviation from the control. Only at later times does the difference in the evolution of $\mu(t)$ for the different parameter sets become apparent. 

\begin{figure}[h!]
    \centering
    \includegraphics[width=0.9\textwidth]{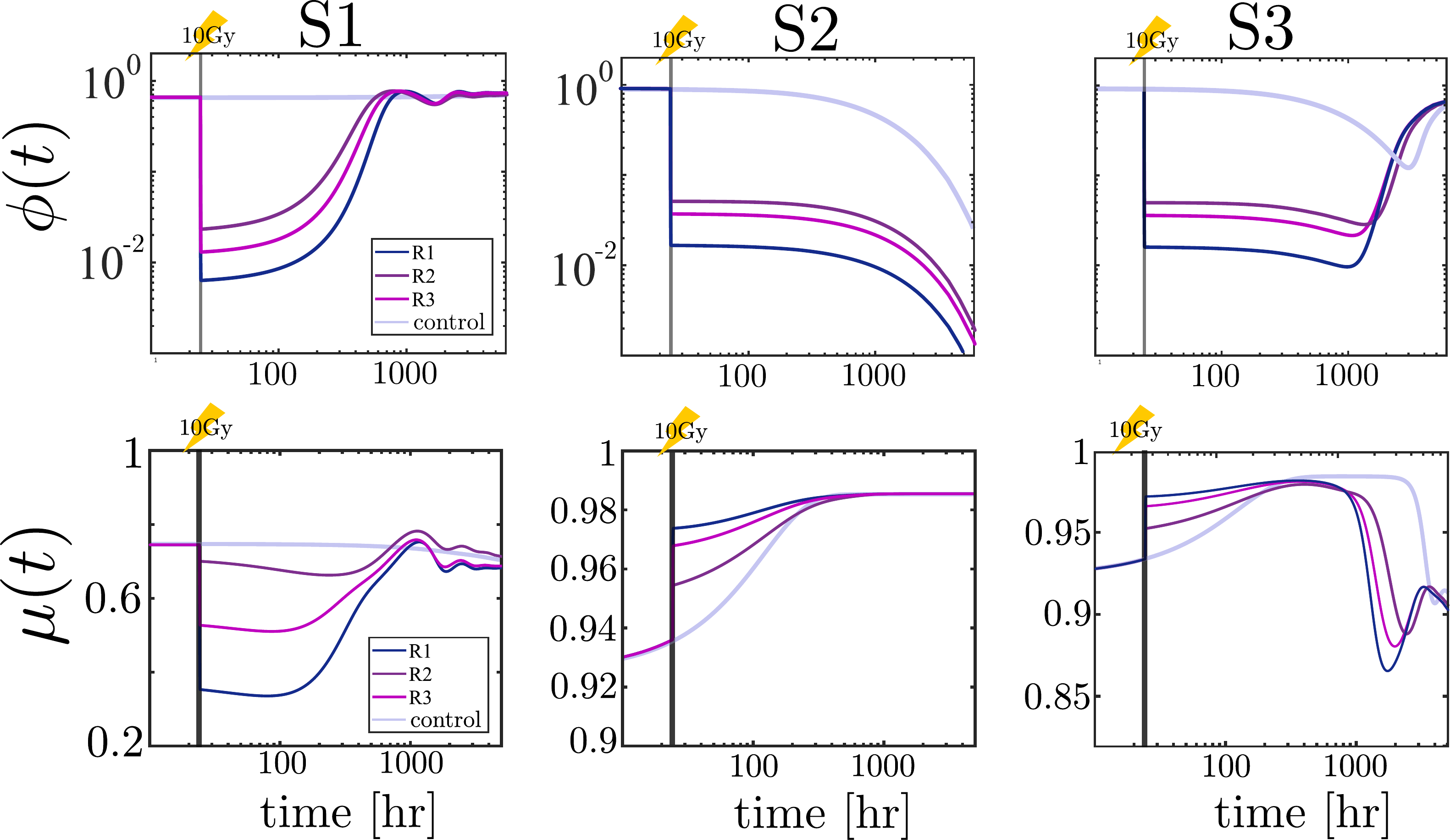}
    \caption{Series of numerical results showing how the growth dynamics and the phenotypic mean evolves following exposure to a single dose of radiotherapy when cell radio-sensitivity is a non-monotonic function of cell phenotype. The simulations are analogous to those presented in Figure~\ref{radio_norm_growth} and \ref{radio_mean_norm}, except that Equations~(\ref{eq:new_al_bet}) are used in place  of Equations~(\ref{rad_alpha})-(\ref{rad_beta}). }
    \label{fig:App_radio1}
\end{figure}

We note that other factors, in addition to stemness, influence cell radio-sensitivity. 
It is natural to expect
cells that have permanently exited the cell-cycle will be less radio-sensitive than cycling cells, as the DNA damage response may already be active in such cells \cite{Lee2014}. 
The functional forms for $\alpha$ and $\beta$ defined by 
Equations~(\ref{rad_alpha})-(\ref{rad_beta}) assume that radio-sensitivity increases monotonically with cell phenotype, $z$.
In order to investigate situations in which TDCs have lower radio-sensitivity than proliferating cancer cells, we now
the following, non-monotonic functional forms:
\begin{subequations}
\begin{align}
    \alpha(z)=\alpha_{min}+(\alpha_{max}-\alpha_{min})\tanh\left(\frac{z}{\xi_R}\right)H_{0.075}(1-z),\\[2pt]
	\beta(z)=\beta_{min}+(\beta_{max}-\beta_{min})\tanh\left(\frac{z}{\xi_R}\right)H_{0.075}(1-z),
\end{align}\label{eq:new_al_bet}%
where $H_\epsilon$ is defined in \S \ref{vz_sec}, and we arbitrarily fix $\epsilon=0.075$ (all other parameters are as defined in \S\ref{fit}). 
\end{subequations}

When the single dose experiment is repeated with the new radio-sensitivity profile, we observe an overall increase in the population survival fraction (compare Figures \ref{fig:App_radio1} and \ref{radio_norm_growth}) and changes in the dynamics of the population mean $\mu(t)$ (compare Figures \ref{fig:App_radio1} and \ref{radio_mean_norm}).  The differences are most pronounced for scenarios $S2$ and $S3$ where TDCs, localised near $z=1$, are dominant in the population prior to treatment. The qualitative growth dynamics (i.e., $\phi(t)$) is similar for both cases. 
Further investigation of these differences is beyond the scope of the current study and is postponed for future work.

\begin{figure}[h]
\centering
\includegraphics[width=0.7\textwidth]{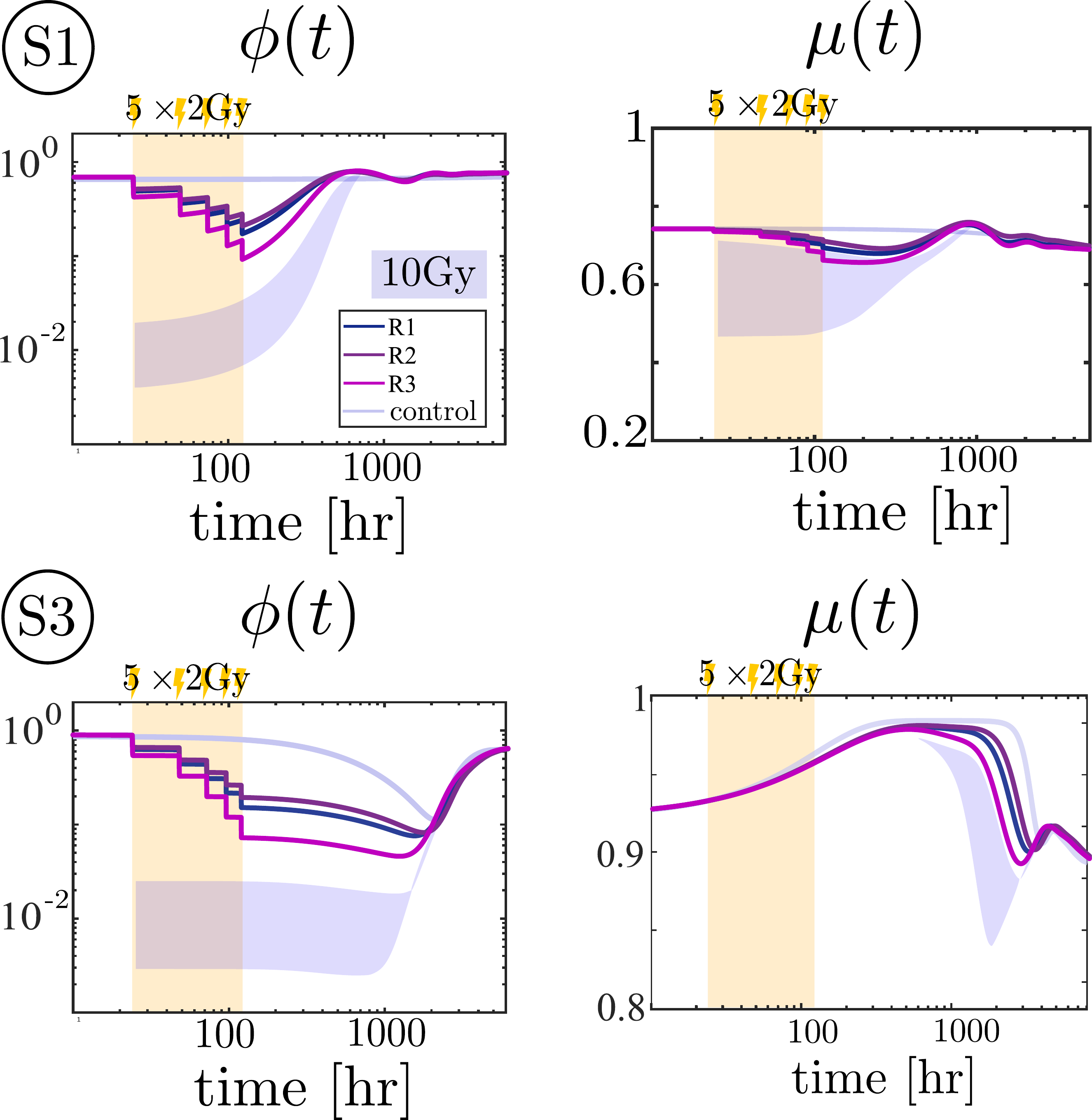}	
\caption{Simulation results for fractionated radiotherapy protocols, showing how the total cell number $\phi(t)$ and the phenotypic mean 
$\mu(t)$ evolve for scenarios S1 and S3 (see Figure \ref{radio_norm_growth} for details). In all plots, the light purple shaded area indicates the variability in responses when a single dose of 10 Gy is administered and is included for comparison with the fractionated treatments (see Figure \ref{radio_mean_norm}). The yellow shaded area indicates the duration of the treatment for the fractionated case.} 
\label{frac_growth}
\end{figure}
In practice, delivery of a single (high) dose of 10 Gy may not be practical for treating patients, due to adverse side effects \cite{Taylor2011}. Therefore, we now consider tumour responses to fractionated RT protocols.
The trends for fractionated RT are similar to those for single doses for all scenarios in Table \ref{tab_rad3}. Typically, the proportion of cells that survive fractionated therapy is larger than for the single-dose case, by a factor of about 100. 
Consequently, for scenarios $S1$ and $S3$, the time to return to the equilibrium population distributions is reduced. For S2, while treatment causes a monotonic decrease in the cell density $\phi$, since more cells survive fractionated RT, it takes longer for the cell population to become extinct. 
For scenarios $S1$ and $S3$, we recall that for high doses of RT, the phenotypic mean  was markedly affected by the specific LQ model parameters considered; this is not the case when lower doses are applied (see Figure \ref{frac_growth}).

\begin{figure}[h!]
	\centering	
	\includegraphics[width=0.95\textwidth]{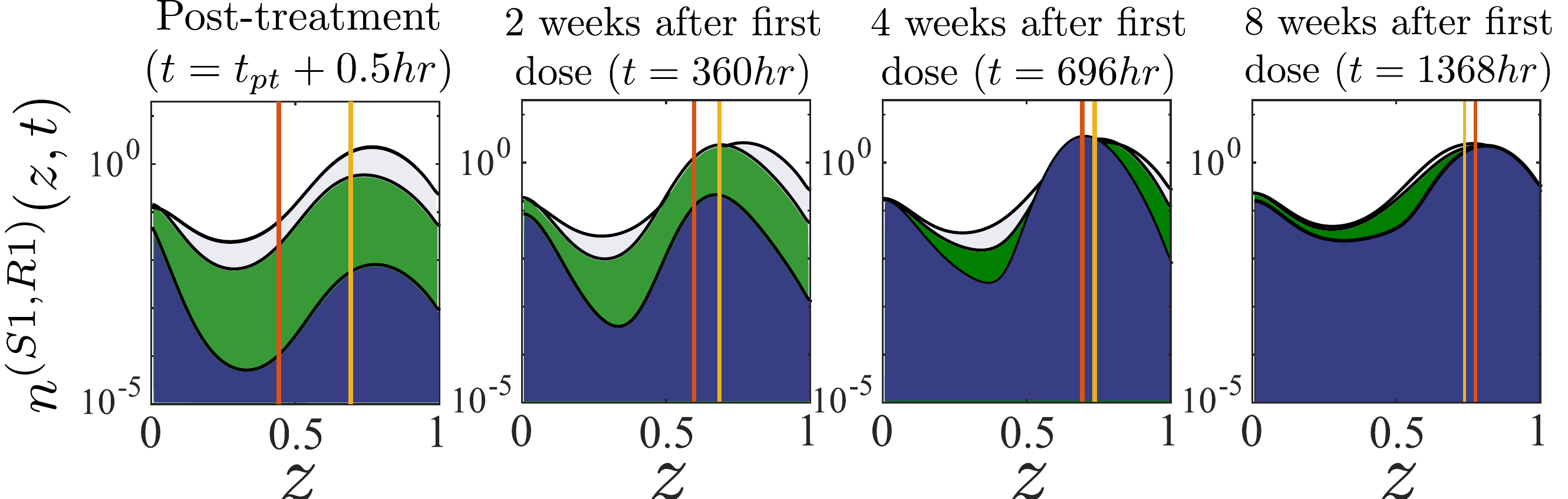}
	\caption{Phenotypic distribution $n^{(S1,R1)}(z,t)$ for the control (light blue), the colony exposed to a single dose (dark blue) and the one treated with fractionated dose $2$ Gy
	$\times5$ (green). The orange and yellow lines indicate the phenotypic mean for the single dose (orange) and fractionated (yellow) therapy respectively. Note that the first panel corresponds to the end of the treatment so that $t_{pt}$ is 24 hr and 120 hr for the 10 Gy and fractionated protocol respectively. On the other hand, the remaining panels are measured relative to the beginning of the treatment, which is at $t=24$ hr for both protocols.}
	\label{distFd}
\end{figure}

The variability in responses for scenarios S1 and S3 following a single dose of radiotherapy can be attributed to the temporary advantage CSCs have post treatment. When using a fractionated protocol, intra-population competition is maintained at the cost of fewer cells being killed. This is apparent when we compare the phenotypic distribution at different times for the two treatment protocols (see Figure \ref{distFd}). When $10$ Gy is administered in one dose (first panel, dark blue region), the peak of the distribution is at $z=0$. On the other hand, after 5 doses of $2Gy$ per day (first panel, green region), the proportions of differentiated and cancer stem cells are approximately equal. 
Given that the former proliferate faster than the latter, the differentiated cells quickly become the dominant phenotype. Consequently, one month after treatment ends (third panel in Figure \ref{distFd}), the proportion of CSCs in the population is the same for both protocols. 
We conclude further that the single dose protocol outperforms the fractionated protocol when we compare the total number of cells (the blue curve is below the green one for all values of $z$).

\subsection{Treatment Response in Hypoxic Conditions}
\label{radio_hyp}
\noindent Cell populations that are continuously exposed to hypoxia, exhibit instantaneous re-growth following RT, as shown in Figure \ref{rad_hypoxia_growth}. Compared with the treatment outcome under normoxia, a higher percentage of cells survive radiation, because there is a larger proportion of radio-resistant cells in the population under hypoxia. Even though a smaller fraction of cells are killed, re-growth is also usually slower under hypoxia than under normoxia. 
We note also that, following exposure to the single and fractionated protocols, the phenotypic mean $\mu(t)$ shifts toward $z=0$ under hypoxia, favouring CSCs as the dominant phenotype (see Figure \ref{rad_hypoxia_growth}). The drift in $\mu$ is less pronounced for the fractionated case, suggesting the latter protocol is less favourable for the immediate accumulation of resistant subpopulation of CSCs than the single dose.

\begin{figure}[h]
	\centering
	\includegraphics[width=0.75\textwidth]{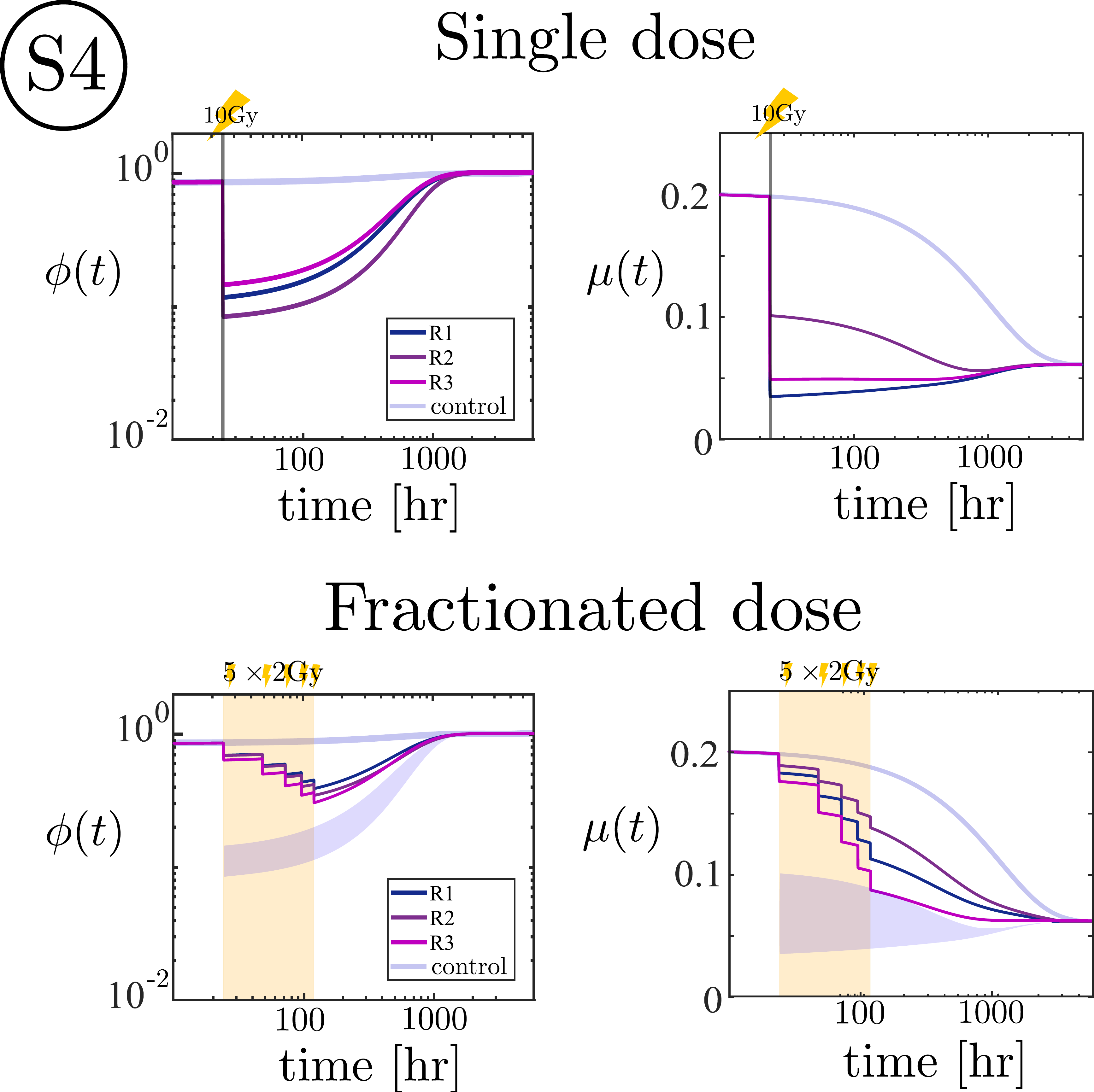}	
	\caption{Comparison of the tumour cell responses to single and fractionated radiotherapy protocols under hypoxia for scenario S4
		(See Table \ref{tab_rad3}). Simulation results showing the time evolution of the cell density, $\phi(t)$, and phenotypic mean, $\mu(t)$, are presented. For comparison, the light purple shaded areas in the fractionated plots indicate the variability in the response when a single dose of $10$ Gy is administered. The yellow shaded areas indicate the duration of treatment for the fractionated case.} 
	\label{rad_hypoxia_growth}
\end{figure}

Taken together, our simulation results suggest that, under hypoxia, RT may accelerate the accumulation of resistant cells, while significantly reducing the overall growth rate of the population. 
 
\subsection{Treatment Response in a Changing Environment}
\label{radio+change}
\noindent Thus far we have assumed that the oxygen concentration remains constant throughout treatment. 
While this may accurately describe RT responses for cells cultured \textit{in vitro}, such control is likely to be absent \textit{in vivo} \cite{Arnold2018,Fenton2001,Kempf2015}. There is currently no consensus about the impact of RT on tumour vasculature and, hence, tissue re-oxygenation. On the one hand, high doses of radiotherapy may damage the vasculature~\cite{Hormuth}, and decrease nutrient availability post radiotherapy. 
On the contrary, moderate RT may transiently increase tissue oxygenation
by \emph{normalising} the tumour vasculature 
(vessel normalisation is a phenomenon that has been observed when tumours are exposed to vascular-targetting agents which destroy some of the blood vessels in a way that increases blood flow through the network and, thereby, tissue oxygen levels \cite{Carmeliet2011,Jain2014}). 

Moreover, as tumour cells are killed, the pressure on immature vessels, not damaged by the radiation, decreases, and oxygen supply to the surviving cells may increase. Equally, hypoxic regions may form at later times as the tumour regrows. From this point of view, radiotherapy may impact both the phenotypic distribution of the cell population (and, thereby, its radio sensitivity),  and oxygen levels post-treatment.
We can use our mathematical model to investigate these scenarios, by assuming that oxygen levels change post radiotherapy.  

Based on the results presented in Sections \ref{radio_norm} and \ref{radio_hyp}, we anticipate that reoxygenation of a hypoxic tumour will be beneficial in certain cases, driving CSC maturation, and even leading to tumour eradication. The results presented in Figure \ref{posthypo_2} show that the long-term tumour regression is preceded by an initial phase of regrowth during which CSCs that survive treatment de-differentiate and proliferate. Such a treatment might initially be considered unsuccessful, although
the stability of the trivial steady state upon 
re-oxygenation leads to extinction at longer times.

\begin{figure}[h]
	\begin{subfigure}{0.4\textwidth}
		\begin{subfigure}{\textwidth}
		\includegraphics[width=0.9\textwidth]{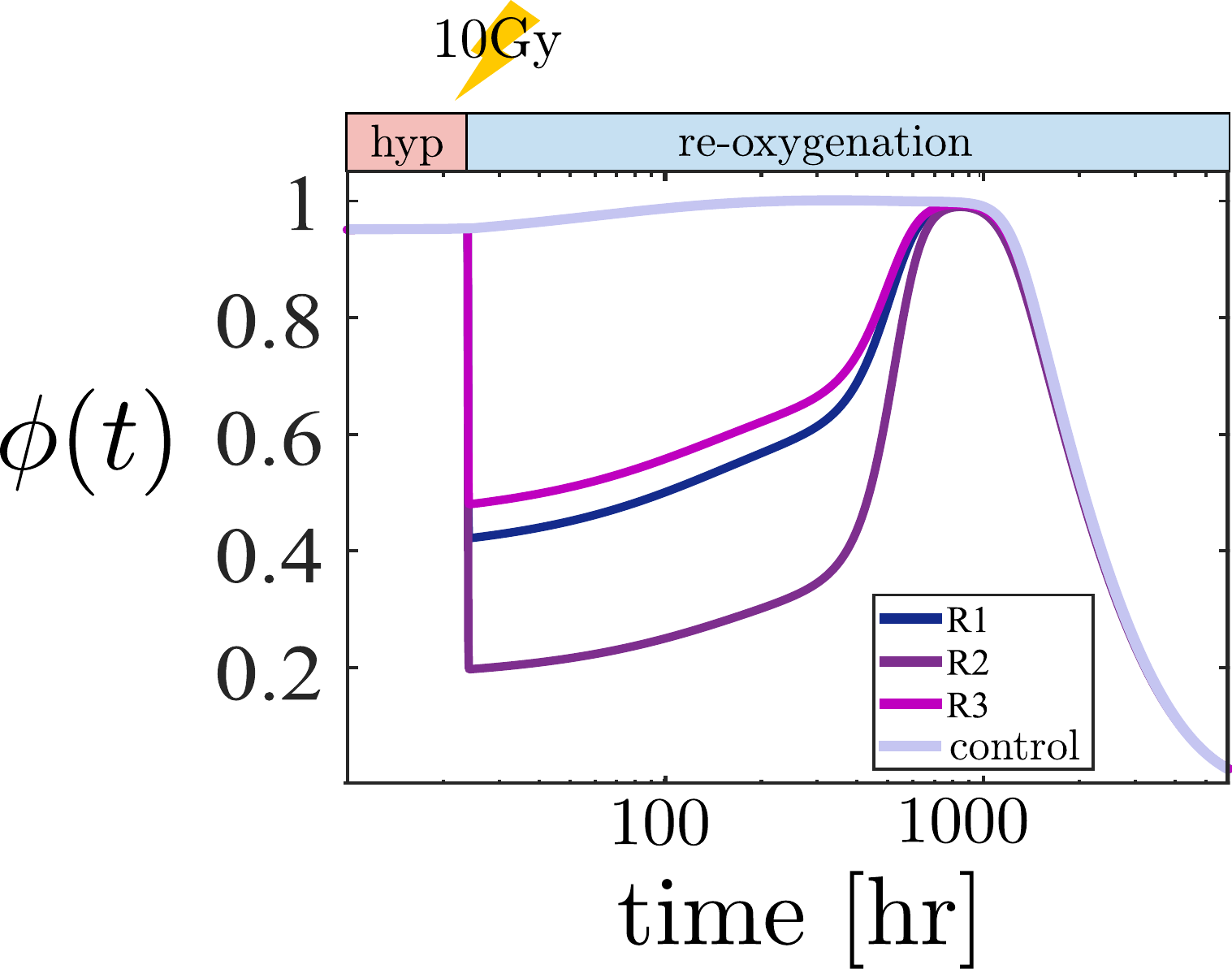}
		\caption{S5}
	\end{subfigure}	

	\begin{subfigure}{\textwidth}
		\includegraphics[width=0.9\textwidth]{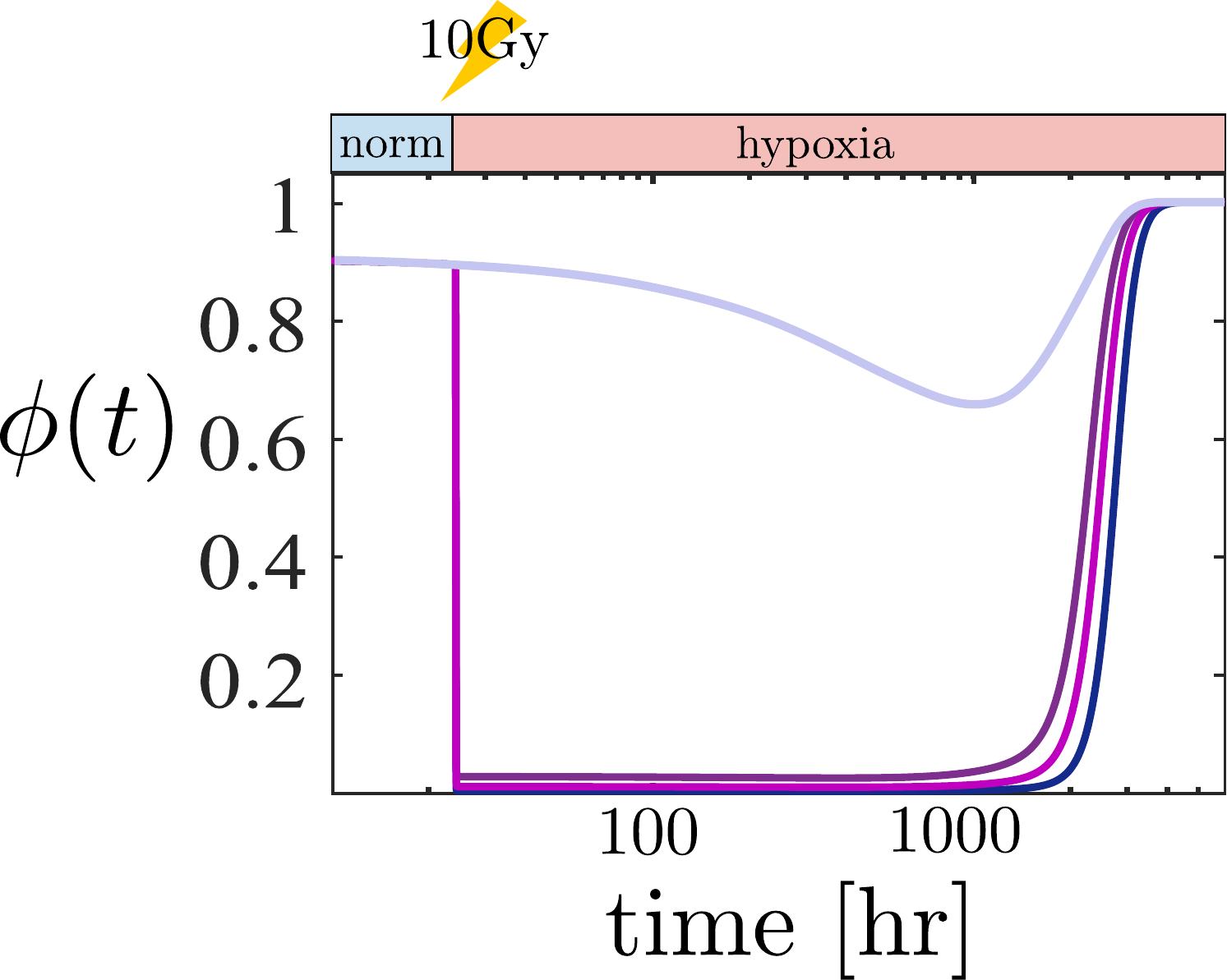}
		\caption{}
		\label{posthypo_2b}
	\end{subfigure}	
	\end{subfigure}
	\hspace{5mm}
	\begin{subfigure}{0.5\textwidth}
	\includegraphics[width=0.9\textwidth]{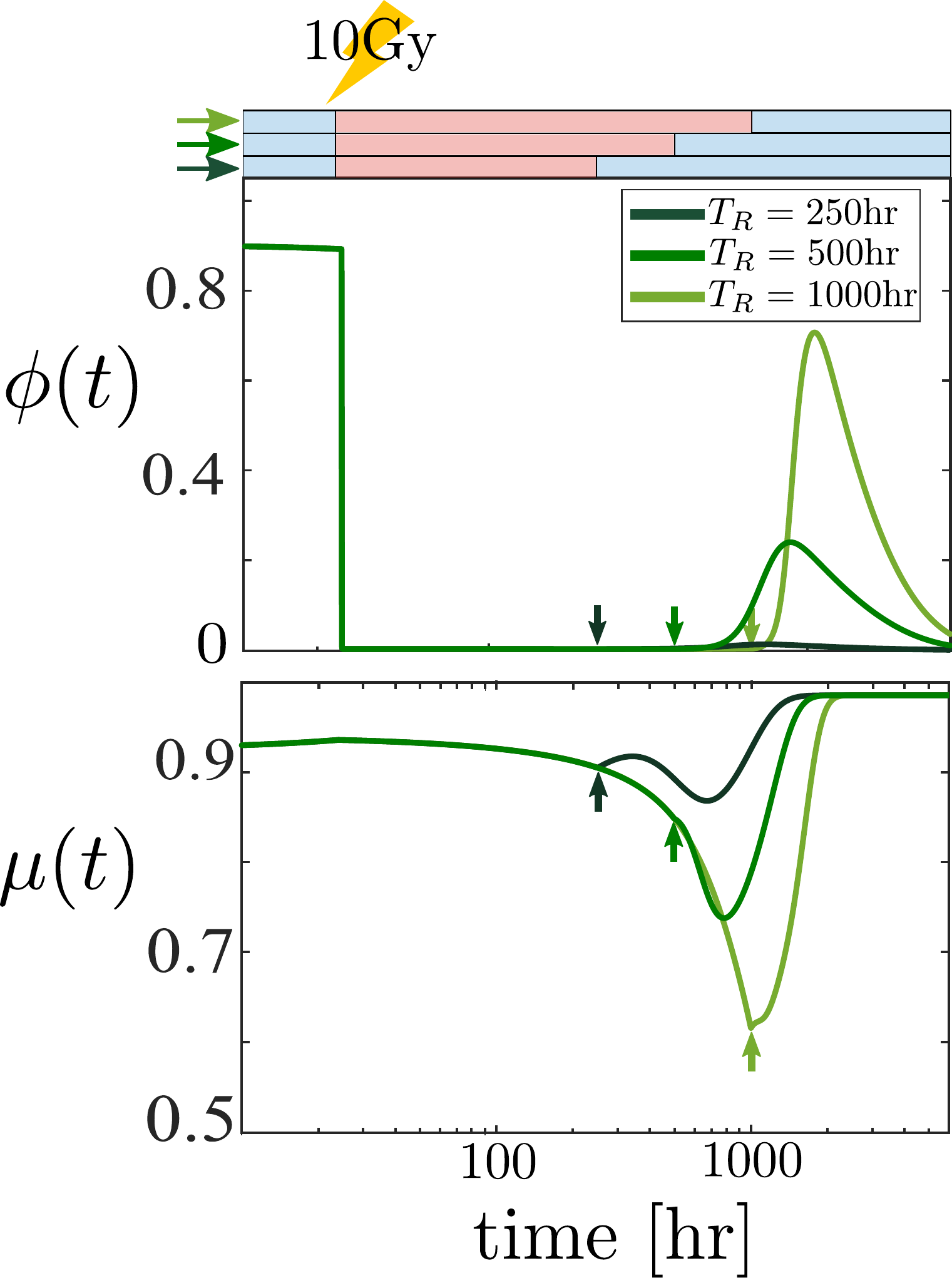}
	\caption{S6}
	\label{posthypo_2c}
	\end{subfigure}	
\caption{Growth curves for changing environmental conditions: (a) re-oxygenation and (b) post-radiation hypoxia for the parameter values S5 and S6 in Table \ref{tab_rad3}, respectively. Different response to treatment are compared based on parameter values from Table~\ref{rad_tab2}. (c) Growth curve $\phi(t)$ and phenotypic mean $\mu(t)$ evolution for model $R1$ from Table~\ref{rad_tab2}, when exposed to transient post-treatment hypoxia. We denote by $T_R$ the time at which re-oxygenation occur (indicated by the arrows in the plot). If $T_R$ is sufficiently small than re-oxygenation does not drive re-growth of the cells population. If we waited for a sufficiently long time (as in case $T_R=1000$) then re-oxygenation would first drive regrowth. Areas in blue and pink correspond to intervals of normoxia and hypoxia respectively.}
\label{posthypo_2}
\end{figure}

As mentioned previously, when high radiation doses are applied \textit{in vivo}, it is likely that the vessel network is also damaged, potentially inducing hypoxia \cite{Arnold2018}. Figure \ref{posthypo_2b} shows that 
such environmental changes may negatively impact the outcome. 
The formation of an hypoxic region favours the development and 
maintenance of radioresistant CSCs, reducing the treatment efficacy and making it more difficult to eradicate the tumour. At the same time, environmental changes may be transient: damaged blood vessels are likely to be replaced by new vessels which form via angiogenesis and re-oxygenate 
the damaged regions. As shown in Figure \ref{posthypo_2c}, depending on the time-scale required for vessel regrowth (indicated by $T_R$), different behaviours may arise. If the duration of RT-induced periods of hypoxia is sufficiently short, then the size of the cell population remains low. 
By contrast, if there is sufficient time for cells to de-differentiate (see $T_R=1000$), then re-oxygenation leads to a rapid increase in cell number, although 
eventually the cells die out. These results highlight the complex interplay between tumour growth and treatment response \textit{in vivo} and the importance of environmental factors in 
determining the eventual outcome of radiotherapy treatment. 

\section{Conclusion and Future Challenges}
\label{conclusion}
\noindent We have developed a structured model to investigate how clonogenic heterogeneity affects the growth and treatment response of a population of tumour cells. Cell heterogeneity is incorporated via an
independent and continuous structural variable which represents \emph{stemness}. As proposed by \cite{Scott169615,Chisholm2016}, we view stemness as a plastic trait, with cells becoming more, or less, stem-like depending on their environmental conditions. 
Our mathematical model accounts for cell proliferation and apoptosis,
inter-cell competition, and phenotypic movement along the stemness axis, via diffusion and advection.

Studies of the population dynamics in the absence of treatment revealed that, under normoxia, a variety of qualitative behaviours may arise
depending on the functional forms used to represent the structural flux and fitness landscape. When advection dominates movement along the stemness axis, its magnitude, relative to the rates of proliferation and cell death, determines whether the population is driven to extinction. Multimodal distributions,  which allow for the formation and maintenance of CSCs pools, are observed for asymmetric velocity profiles. Under hypoxia, the population distribution is unimodal and skewed toward stem-like phenotypes, with little intra-population variability. 
The resulting cell distribution is highly
resistant to radiotherapy, the tumour will typically
regrow following treatment. By contrast, under
normoxia (or re-oxygenated hypoxia), and 
for suitable parameter values, the tumour may become extinct following radiotherapy.

There are many ways in which the work presented in this paper could be extended. 
A first, natural extension would be to incorporate structural and spatial heterogeneity (i.e., both phenotypic and spatial dimensions) \cite{hodgkinson}. 
This would enable us to consider \textit{in vivo} situations, where spatial gradients in oxygen levels emerge naturally, due to oxygen consumption by the cells as it diffuses from blood vessels. As outlined in~\ref{AppendixA}, in such a model oxygen consumption rates may vary with cell phenotype, and spatial fluxes may account for random movement of the cells. 
Preliminary results for such a model are presented in Figure~\ref{space1}. 
We consider a 1D Cartesian geometry and focus on a tumour region of  
width $L$, in which a blood vessel located at $x=0$ provides a continuous 
supply of oxygen to the tissue. If the tumour initially comprises a spatially homogeneous distribution of terminally differentiated cells (see Equation~(\ref{initial_cond})), then the oxygen rapidly relaxes to a steady state and a hypoxic region forms at distance from $x=0$. In contrast to the well-mixed model, cells are now able to move, by random motion, between normoxic and hypoxic regions. 
While terminally differentiated cancer cells are dominant in the well-oxygenated region, a small fraction persists in the hypoxic region (in particular, near the boundary of the hypoxic region, orange line in the plots in Figure~\ref{space1}). 
This is due to the influx of 
cells from the well-oxygenated portion of the domain. Similarly, CSCs are dominant in the hypoxic region, but a small fraction of hypoxic CSCs migrate towards $x=0$, where re-oxygenation induces their maturation, creating a differentiated and highly proliferative cell phenotype, alongside terminally differentiated cancer cells.
These results illustrate how the interplay between space, resources
and phenotypic adaptation may give rise to complex behaviours; their investigation is the focus of ongoing work.
\begin{figure}[h!]
	\centering
	\includegraphics[width=\textwidth]{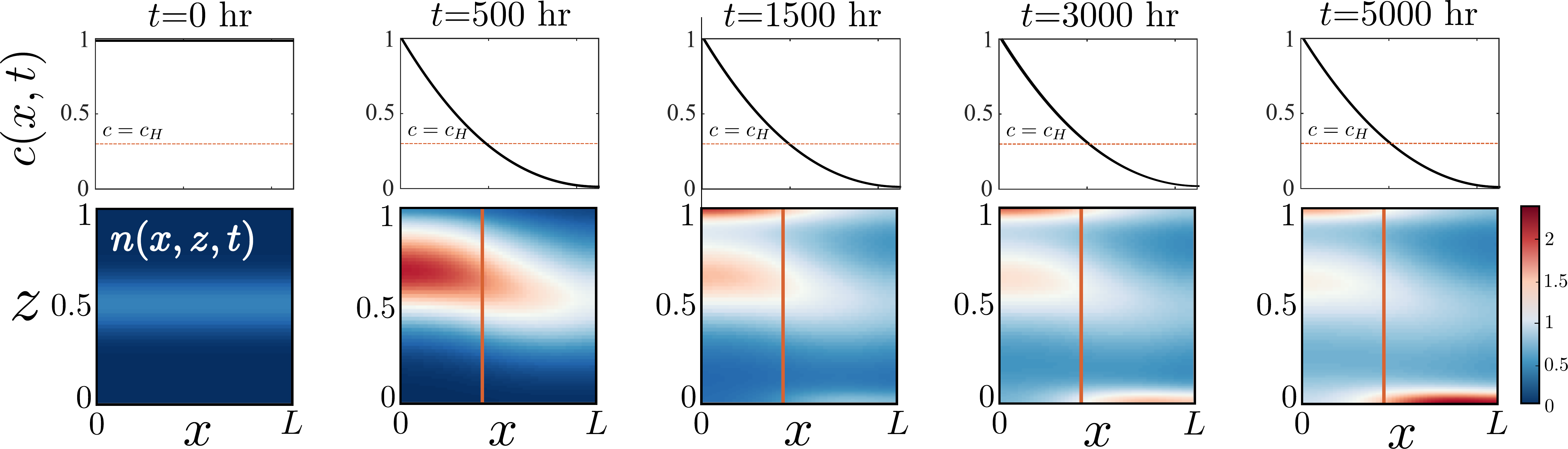}
	\caption{Series of plots showing how, in the absence of treatment, the cancer cell population $n(x,z,t)$ and the oxygen concentration $c(x,t)$ change over time $t$ when we account for spatial and phenotypic variation (see Equations~(\ref{spatial_mod})). 
	We indicate the threshold $c=c_H$ which defines the boundary of the hypoxic region with a horizontal red line in the upper plots and with a vertical orange line in the lower plots. We fix $V_\pm=4\times10^{-4}$, $\xi_\pm=0.1$, $\omega_+=1$, $\omega_-=2$ and $d_f=0.001$, while the remaining model 
	parameters are fixed at the values stated in Table \ref{param_set}.}
	\label{space1}
\end{figure}

A significant challenge of the modelling approach presented in this paper is the determination of model parameters and functional forms. 
In the longer term, techniques such as \textit{single-cell RNA sequencing}~\cite{tirosh,venteicher} will make it be possible to 
quantify specific aspects of
our model, such as the dependence of the proliferation and apoptosis rates on cell stemness and the dependence on 
the tumour micro-environment of the (phenotypic) advection velocity 
associated with cell maturation and de-differentiation. 
In spite of their current limitations, we believe that studies of
such models can increase understanding of the ways in which specific 
physical processes may influence 
the phenotypic distribution of cell populations in different environments. 
At the same time, we acknowledge that 
it remains a matter of debate as to whether asymmetric cell distributions are driven by micro-environmental signals (as in the model presented here), asymmetric division, or a combination of the two~\cite{Roeder2006}. By using a non-local proliferation kernel to account for asymmetric division, we could investigate these alternative hypotheses and identify conditions under which they lead to different outcomes. 

A important feature of our model is the way in which the response to radiotherapy (RT) varies with cell stemness (i.e., $z$). 
Our analysis shows how the functional forms used to describe the advection velocity and fitness functions can affect the system dynamics post-RT. 
While unimodal phenotypic distributions lead to monotonic growth curves post-treatment, more complex behaviour is observed when heterogeneous populations, with a pool of CSCs, are considered. 
For example, under normoxia, the presence of radio-resistant CSCs can drive recurrence, despite an initial phase of tumour regression. As the CSCs mature into highly-proliferating cancer cells, rapid re-growth is accompanied by re-sensitisation of the population to RT. 
Under hypoxia, CSCs maintain their stemness, leading to a slowly growing, 
radio-resistant cell population. 
More complex outcomes arise when we consider the effect that treatment might have on the environment.
As noted in~\ref{radio+change}, changes in the vasculature induced by radiotherapy can result in either post-treatment re-oxygenation or hypoxia. While re-oxygenation increases the radio-sensitivity of the population, hypoxia increases their radio-resistance. In practice, such environmental changes are likely to be transient. Even in an untreated tumour, 
fluctuations in oxygen levels can occur. Consider, for example, 
cells in a neighborhood of immature blood vessels. As the cells proliferate, they exert mechanical pressure on the vessels, causing them to collapse and local oxygen levels to fall. Under hypoxia, the tumour cells stimulate the growth of new blood vessels from pre-existing ones, via angiogenesis. In this way, tumour regions may cycle between periods of hypoxia and normoxia. 
It would be of interest to extend the model to account explicitly for the tumour vasculature and its interaction with tumour cells. This could be achieved at a ``high level'' of description, via
simple ODE models such as \cite{Hahnfeldt1999, Stamper2010}, or via more complex, multi-phase \cite{Hubbard2013} or multi-scale approaches \cite{Byrne2010,Macklin2009,Vavourakis2017,Walpole2013}. 

This would enable us to better capture the different time-scales on which the oxygen dynamics and cell adaptation velocity change. As shown in 
Figure \ref{space2}, variations in oxygen levels emerge naturally within spatially-resolved models. Here, cell killing leads to tissue re-oxygenation which, in turn, disrupts the CSC niche. Depending on the time scale over which the cells adapt to their new environmental conditions, this may increase the overall radio-sensitivity. Understanding and accounting for such phenomena is particularly relevant for predicting responses to RT and comparing alternative treatment protocols.

\begin{figure}[h!]
	\centering
	\includegraphics[width=\textwidth]{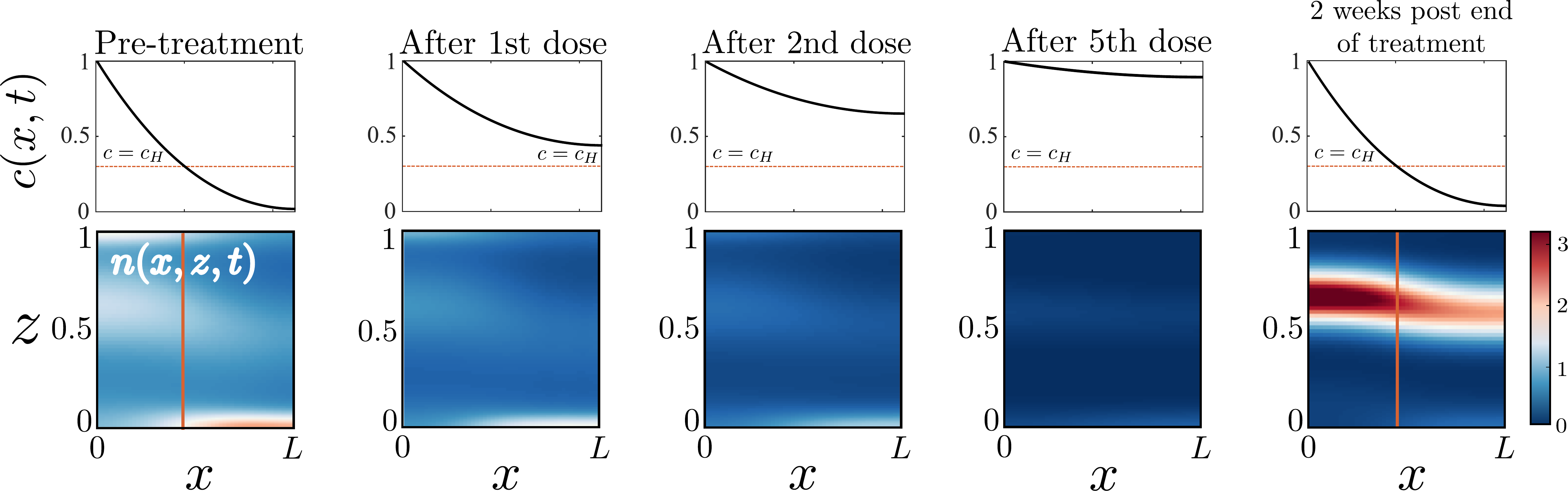}
	\caption{Evolution of the population $n(x,z,t)$ in the spatial and phenotypic dimensions following a cycle of fractionated radiotherapy 
	(5 $\times$ $2$ Gy). The parameter values are the same as those used in Figure \ref{space1} and the initial cell distribution is the same as the final distribution in Figure \ref{space1}. For the LQ-model we used parameter set $R3$ in Table \ref{rad_tab2}.}
	\label{space2}
\end{figure}

In the extinction scenario, or post administration of high RT doses, the number of cells in the population can become low and our continuum model may cease to be valid. In such conditions, stochastic effects which are neglected herein may become important. As in \cite{Ardaseva2020_2,Franz2013,Spill2015}, stochastic and mean field approaches may be combined with hybrid discrete-continuum techniques to account for small population effects and to study their impact on the probability of tumour extinction. 

In this paper, we considered only single dose and fractionated treatment protocols. In future work, we could investigate alternative strategies, such as \textit{adaptive therapeutic} protocols \cite{Gatenby2009} and/or multi-drug treatments, which have been proposed as an effective way to overcome radio-resistance. From this point of view, considerable efforts have been invested  in designing treatments that exploit features of CSCs, such as their
metabolic plasticity~\cite{frontiers}. 
Motivated by recent metabolically-structured models~\cite{Ardaseva2019,hodgkinson, Villa2019}, a natural extension of our model would be to include a 
``metabolic dimension'' in order to investigate the interplay between stemness, metabolic switching and
resistance. A biologically informed model that incorporates
metabolic and phenotypic effects, together with the tumour micro-environment
and vascular remodelling lies at the heart of a
mathematical program that would enable systematic comparison with
{\it in vivo} observations. The framework and results outlined in this work represent a first step towards achieving this long-term goal.

\appendix
\section{Spatial Model}
\label{AppendixA}
\noindent We outline here  the set up for the 1D simulations presented in Section \ref{conclusion}. As a full description of the spatial model goes beyond the scope of the present work, we focus on the main changes to~(\ref{mixedmodel})-(\ref{eq_ad}). 
We now view the oxygen concentration $c$ as a dependent variable, rather than a prescribed function. We suppose that oxygen is supplied to the region by blood vessels on the domain boundary $\partial \Omega_2$ (see Figure~\ref{schematicspatial}). Oxygen diffuses from the boundary into the tissue where it is consumed by the tumour cells at rates which depend on their phenotype and the local oxygen concentration. The evolution of the dimensionless cell density,  $n=n(\vec{x},z,t)$, is driven by a phenotypic flux of the same form as in Equation~(\ref{mixedmodel}) but a spatial flux is included to account for random motion in the spatial dimension. 

\begin{figure}[h!]
	\centering
	\includegraphics[width=0.85\textwidth]{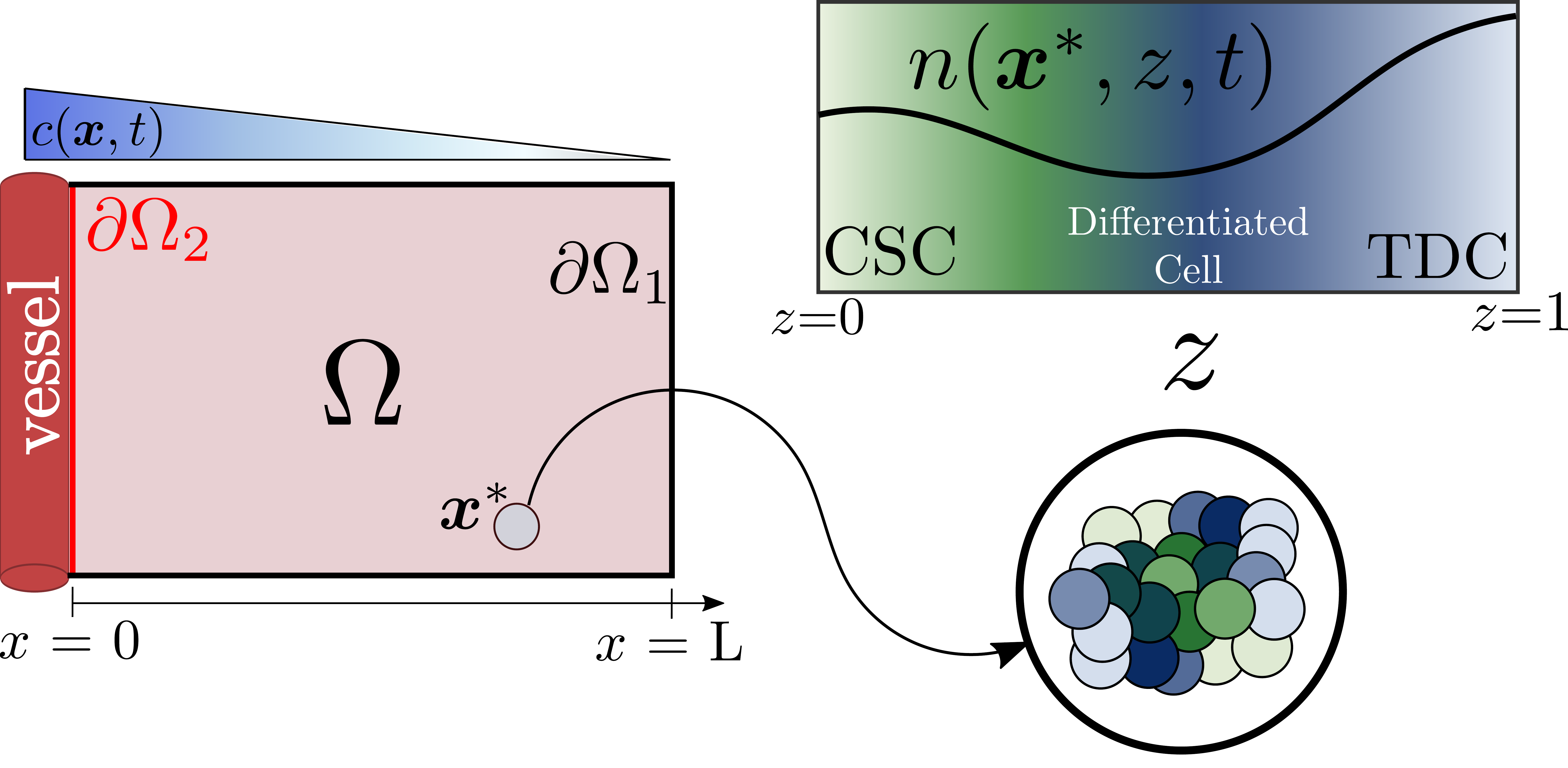}
	\caption{Schematic representation of the phenotypic and spatial model.}
	\label{schematicspatial}
\end{figure}

As shown in Figure \ref{schematicspatial}, we consider a fixed tissue slice where the oxygen supply (i.e. vasculature) is confined to one of the tissue boundaries. Given the assumed symmetry of the problem, we can consider a 1D Cartesian geometry with $x\in[0,L]$. The spatial model is defined by the following system of coupled PDEs:
\begin{subequations}
	\begin{align}
	\hspace{-10mm}
	  \frac{\partial n}{\partial t}=\underbrace{ D_N \frac{\partial^2 n}{\partial x^2}}_{spatial \hspace{1mm} flux}+ \frac{\partial }{\partial z} \left(\theta \frac{\partial n}{\partial z}-n v_z(z,c)\right)+ F(z,c,\phi,t) n,\\
	\frac{\partial c}{\partial t} = D_C\frac{\partial^2 n}{\partial x^2}-\Gamma(t,x,c),\label{ox}\\
	\theta \frac{\partial n}{\partial z}-n v_z = 0, \qquad z\in\left\{0,1\right\},\, x\in [0,L],\, t>0,\\
	\left.\frac{\partial n}{\partial x}\right|_{x=0}=\left.\frac{\partial n}{\partial x}\right|_{x=L}=0, \quad z\in(0,1),\, t>0,\\
	\left.\frac{\partial c}{\partial x}\right|_{x=L}=0, \quad c(0,t)=c_{\infty}, \quad t>0,\\[2pt]
	n(x,z,0)= n_0(x,z) \quad x\in[0,L],\, z\in(0,1),\\[2mm]
	c(x,0)= c_0(x) \quad  x\in[0,L],\\
	\phi(x,t)=\int_0^1 n(x,z,t) \, dz,\\
	\Gamma(t,x,c)=\int_0^1 \gamma(z,c) n(x,z,t) \, dz,\\
	\begin{aligned}
	F(z,c,\phi,t)= p(z,c)\left(1-\phi\right) -f(z) - \underbrace{g H(c_N-c)}_{necrosis}\\-\sum^{N}_{i} \log\left(\frac{1}{SF(z,c)}\right)\delta(t-t_i).
	\end{aligned}
	\end{align}\label{spatial_mod}
\end{subequations}
In Equation~(\ref{spatial_mod}), $D_N$ and $D_C$ are the assumed constant spatial diffusion coefficient for the cells and oxygen, respectively, while $\gamma$ denotes the rate at which cells of phenotype $z$ consume oxygen and $\Gamma$ the net rate of oxygen consumption at position $x$ and time $t$. The advection velocity $v_z$ is as defined by Eq.~(\ref{eq_ad}), while the fitness function $F$ is analogous to that defined in Section \ref{fit}, with an additional term to account for necrosis. The latter is assumed to occur at a constant rate $g \geq 0$, independent of cell phenotype, when the oxygen concentration falls below a threshold value, $c_N \geq 0$. We also modify the definition of the survival fraction $SF$ given in \S\ref{fit} (see Equation~(\ref{lq_model})) to account for the \textit{oxygen-enhancement ratio} (OER)~\cite{lewin,lewin2}. According to the \textit{oxygen fixation hypothesis} \cite{HallEricJ2012Rftr}, part of the biological damage induced by radiation is  indirect, being mediated by the presence of free radicals. Thus, when oxygen is limited, radio-sensitivity is accordingly reduced. Based on experiments, the range of oxygen concentrations at which this effect is relevant corresponds to more severe levels of hypoxia (where $c\sim 0.5\%$ or lower). We do not consider such situations for the well-mixed model, where we consider mild hypoxia. However, accounting for the OER will be important for the spatially extended model. Recall from Section \ref{hyp cond} that hypoxia is a favourable \textit{niche} for CSCs. Therefore the OER will endow them with additional protection from radiation. 
Denoting by $c^R_H$ the oxygen threshold at which the OER becomes active, we use the following functional form for the survival fraction when simulating the spatially-extended model: 
\begin{align}
SF(z,c)=\begin{cases}
\exp\left[-\alpha(z)d-\beta(z) d^2\right] \quad c>c^R_H\\[2mm]
\exp\left[-\myfrac[2pt]{\alpha(z)}{OER}d-\myfrac[2pt]{\beta(z)}{OER^2} d^2\right] \quad c<c^R_H.
\end{cases}\label{SF_space}
\end{align}
In Equation (\ref{SF_space}), $\alpha$ and $\beta$ are defined by~(\ref{radiosensitivity}).
We note that in the main text, we consider $c=1$ (normoxia) and $c=0.2$ (hypoxia), so that the OER does not impact cell responses to RT.  


For the well-mixed model, the oxygen concentration is typically maintained at a prescribed, constant value. By contrast, for the spatially extended model, we suppose that the tumour cells consume oxygen at a rate $\gamma$ which depends on their phenotype, $z$. As mentioned previously, stem cells are known to have a glycolytic metabolism and, thus, we assume that they consume less oxygen than cancer cells. Consequently, we consider $\gamma$ to be a monotonically increasing function of the phenotypic variable $z$ which asymptotes to its maximum value for $z>0.5$:
\begin{equation}
\gamma(z,c)= H(c-c_N)\left[\gamma_{max} -\frac{\gamma_{max}}{2} e^{-k_{\gamma} z}\right].\label{gamma}
\end{equation}   
In Equation (\ref{gamma}), $H=H(x)$ is the Heaviside function 
(i.e. $H(x)=1$ if $x>0$ and $H(x)=0$ if $x\leq 0$).
In order to continue their normal function, glycolytic cells consume oxygen, albeit at a lower rate. Motivated by results presented in \cite{consumption}, we assume that glycolytic CSCs consume oxygen at approximately half the rate of terminally differentiated cancer cells.

\subsection{Parameters}
\label{App_param}
\begin{table}[h!]
	\centering\footnotesize  
	\begin{tabular}{l c@{\hspace{0.5cm}} l@{\hspace{0.5cm}} l@{\hspace{0.5cm}} c c}
		\toprule[2pt]\addlinespace[2pt]
		&Parameter &  Value &  Units &  Reference & Label\\[2pt]
		\toprule\addlinespace[4pt]
		Phenotypic Diffusion&$\theta$ & $5\times 10^{-6}$& $hr^{-1}$ & -& \\[3pt]
		\hline\addlinespace[3pt]
		\multirow{3}{20mm}{\centering Advection velocity $v_z$ Eq~(\ref{eq_ad})}& $V_\pm$ & $\left\{2,4,8\right\}\times 10^{-4}$& $hr^{-1}$ & -& \\
		& $\xi_\pm$ & $\left\{0.05,0.1,0.5\right\}$& - & -& \\
		& $\omega_\pm$ & $\left\{1,2\right\}$& - & -& \\	
		\toprule\addlinespace[2pt]
		\multirow{9}{25mm}{\centering Fitness $F$ Eq~(\ref{eqnetprol})-(\ref{apoptosis})} &$p^{max}_0$ & 0.005& $hr^{-1}$&\cite{Sweeney1998} &\\[2pt]
		&$K_{H,0}$ & $0.05$ & -&-&\\[2pt]
		&$g_0$ &0.01&-&-&\\[2pt]
		&$p^{max}_1$ & 0.02& $hr^{-1}$&\cite{Sweeney1998}&\\[2pt]
		&$K_{H,1}$ & $0.3$ & -&-&\\[2pt]
		&$g_1$ &0.04&-&-&\\[2pt]
		&$d_f$ & $\left\{0.001,0.015\right\}$ &$hr^{-1}$ &-&\\[2pt]
		&$k_f$ & $10$ & -&-&\\[2pt]
		& $\Phi_{max}$ & $10^8$ & cell/cm$^3$&\cite{DelMonte2009}&\\[5pt]
		\multirow{4}{25mm}{\centering Survival Fraction $SF$ Eq~(\ref{radiosensitivity})/Eq~(\ref{SF_space})}& $\alpha_{min,max}$& Table \ref{rad_tab2}& Gy&\cite{Saga}&\\
		& $\beta_{min,max}$& Table \ref{rad_tab2}& Gy$^{-2}$&\cite{Saga}&\\[2pt]
		& $\xi_R$ & 0.2& - & -&\\[3pt]
		& OER & 3 &- & \cite{lewin2}& S\\[2pt]
		\toprule\addlinespace[2pt]
		\multirow{2}{30mm}{\centering Initial phenotypic distribution $n_0$}&$\phi_0$ & 0.4& $hr^{-1}$&-&\\[2pt]
		&$\sigma$& 0.1 & -&- & \\[2pt]
		\toprule[1.5pt]\addlinespace[3pt]
		\centering Spatial Diffusion &$D_N$ &$1.25\times 10^{-4}$& mm$^2$hr$^{-1}$&&S\\[2pt]
		\centering Domain Size & L & 0.45 & mm&-&S\\[3pt]
		\toprule		
		Oxygen Diffusion & $D_c$ & $6.84\times 10^{-1}$&mm$^2$hr$^{-1}$&-&S\\[3pt]
		\multirow{2}{25mm}{\centering Consumption $\gamma$ Eq~(\ref{gamma})}&$\gamma_{max}$ &$3.11\times 10^{-12}$&g(cell hr)$^{-1}$&\cite{Boag1970}&S\\
		&$k_\gamma$ & $10$ & -&-&S\\[3pt]	
		\multirow{3}{25mm}{\centering Oxygen thresholds}&$c_\infty$ & 1& -&\cite{Lewin2018}&S\\[2pt]
		&$c_H$ & 0.3&  -& \cite{lewin,ester2}&S\\[2pt]
		&$c_N$ & 0.0125 &  -&\cite{ester2}&S\\[2pt]
		\bottomrule[2pt]\addlinespace[2mm]
	\end{tabular}
	\caption{List of the parameters values in model~(\ref{mixedmodel})-(\ref{eq_ad}) and/or its spatial extension~(\ref{spatial_mod})-(\ref{gamma}). Where the parameters are free, we list the set of values considered in the paper. We further label with (S) those parameter that are only present in the spatial model.}
	\label{param_set}
\end{table}  
The model contains a large number of parameters, most of which will vary in value between tumours and patients. The main focus of this work is to study the role played by phenotypic advection (as it interacts with cell proliferation and apoptosis, as well as competition mechanisms). On this basis, we decided to perform a parameter sweep for parameters associated with the advection velocity, while holding all other model parameters fixed at values previously reported in the literature, where such values exist. 
The main challenge is to identify the phenotypically dependent parameters, such as the growth rate in Equation (\ref{pO2}). As most data reported in the literature refer to processes, such as cell proliferation, at the population/cell colony and do not account for phenotypic variation, it was difficult to estimate parameters that characterise the phenotypic variation in these processes. 

We based our estimates of the proliferation rate on the doubling times reported by \cite{Sweeney1998} for two breast cancer cell lines, MCF-7 and BT-549. The former belong to the class of \textit{laminal}-like cells which are characterised by low stemness levels \cite{Ricardo2011} and high proliferation rates (doubling time $1.8$ days, i.e., growth rate $0.016$ hr$^{-1}$). On the other hand, BT-549 belong to the class of \textit{triple-negative} cells whose population is dominated by highly aggressive but slowly proliferating stem-like cells \cite{Ricardo2011} (doubling time $3.7$ days \cite{Sweeney1998}, i.e., growth rate $0.008$ hr$^-1$). Given the variability in the phenotypic distribution of these cell lines, we have rounded the values to those presented in Table \ref{param_set}.

As is common in the literature, we have chosen the source of oxygen (i.e. $c_\infty$) to be at a pressure of $100$ mmHg \cite{Lewin2018}. Given that atmospheric pressure corresponds to $760$ mmHg with  $21\%$ O$_2$, the oxygen tension corresponding to $c_\infty$ is about $8\%$\, O$_2$. The hypoxic and necrotic thresholds ($c_H$ and $c_N$) are then equivalent to oxygen pressures of $2.5\% \, O_2$ and $0.1\% \, O_2$ in line with \cite{ester,ester2}. These values can be converted into oxygen concentrations by use of Henry's law \cite{Lewin2018}, see Table \ref{param_set}.

\section{Linear Stability Analysis.}
\label{AppendixB}

\noindent As mentioned in Section \ref{LSA}, in order to compute the largest eigenvalue $\lambda_0$ numerically we rely on the \textit{Chebfun} package for MATLAB \cite{Driscoll2014}. In order to solve the eigenvalue problem we first 
make the following substitution in Equation~(\ref{neu}):
\begin{equation}
\delta{n}=y(z) \exp\left[\frac{1}{2\theta}\int^z v_z(s) \,ds\right].
\end{equation} 
It is straightforward to show that the function $y$ satisfies the following eigenvalue problem:
\begin{eqnarray}
\begin{aligned}
\theta \frac{d^2 y}{d z^2}+q(z;c,\bar{\phi})y-p(z;c)\bar{n}\int_0^1 y(s)k(s,z) \,ds=\lambda y
\end{aligned}\label{eq:ApBeig}\\
\mbox{where} \quad  q(z;c,\bar{\phi})=p(z;c)(1-\bar{\phi})-f(z)-\frac{1}{2}\frac{d v_z}{d z}-\frac{1}{4}\frac{v^2_z}{\theta},\label{eq:q}\\
\mbox{and} \quad k(s,z)=\exp\left[\frac{1}{2\theta}\int_s^z v_z(p) \,dp\right],\\
\frac{d y}{dz}=0\quad \mbox{at} \; z=0,1.
\end{eqnarray}

Note that the integral in Equation~(\ref{eq:ApBeig})is of the form of a Fredholm integral which is built in the \textit{Chebfun} package \cite{Driscoll2014}. The above differential equation for $\bar{n}=0$
corresponds to the standard form of a
Schr{\"o}dinger-type, \textit{Sturm-Liouville} eigenvalue problem, where  the \textit{Hermiticity} of the differential operator
implies the existence of purely real eigenvalues.
In the case of the null steady state, the eigenvalue problem simplifies to:
\begin{subequations}
\begin{align}
\begin{aligned}
\theta \frac{d^2 y}{d z^2}+\tilde{q}(z;c)y=\lambda y
\end{aligned}
\label{neu_cond0}
\\
\frac{d y}{dz}=0 \quad \mbox{at} \; z=0,1. \label{neu_cond}
\end{align}\label{eq:eigenSL}%
\end{subequations}
where $\tilde{q}(z;c)=q(z;c,0)$ as defined in~(\ref{eq:q}). 
Therefore, applying the \textit{Sturm Oscillation Theorem}~\cite{coddingtonlevinson}
to~(\ref{eq:eigenSL})
we deduce that $\sigma(\mathcal{M})$ has infinitely many simple and real eigenvalues which can be enumerated in strictly decreasing order:
\begin{equation}
\lambda_0>\lambda_1>\ldots, \, \lim\limits_{n\rightarrow \infty} \lambda_n=-\infty.
\end{equation}
We conclude that the trivial steady state is either a stable node (if $\lambda_0<0$) or a saddle (if $\lambda_0>0$). 

In addition to numerical estimation of $\lambda_0$, 
analytical approximations and bounds can be obtained via the so-called  \textit{Rayleigh quotient} $R(y)$. 
If we multiply Equation~(\ref{neu_cond0}) by $y$ and integrate by parts, then we obtain:
\begin{subequations}
\begin{align}
R(y)=\frac{1}{\|y\|^2_{L^2}} \: 
\int_0^1 \left \{ \theta y\frac{d^2 y}{d z^2}+\tilde{q}(z;c)y^2 \right \}  dz,
\end{align}
where $y$ also satisfies the Neumann boundary conditions \ref{neu_cond}. 
We deduce that the following therefore holds:
\begin{align}
\lambda_0 =\sup_{y\in E, \ y\neq 0} R(y)
\end{align}
\end{subequations}
where $E$ is the set of twice differentiable functions that satisfy condition~(\ref{neu_cond}). 
\begin{lemma}
	If the function $\tilde{q}$ is such that $\max\limits_{z\in(0,1)} \tilde{q}<0$ then the null steady state is stable.\label{lemma1}
\end{lemma}
\begin{proof}
Consider the numerator of the quotient defining $R(y)$:
\begin{subequations}
\begin{align}
\begin{aligned}
\int_0^1\left \{ \theta y\frac{d^2 y}{d z^2}+\tilde{q}(z;c)y^2 \right \} dz = - \cancelto{0}{\left[y\frac{d y}{d z}\right]_{0}^1} + \int_0^1\tilde{q}(z;c)y^2-\theta \left(\frac{d y}{d z}\right)^2 dz \\
\leq \int_0^1\tilde{q}(z;c)y^2 dz. \qquad \qquad
\end{aligned}
\end{align}
We deduce that
\begin{align}
R(y) \leq \int_0^1 \tilde{q}(z;c) \frac{y^2}{\|y\|_2^2} dz=R_{up}(y).
\end{align}
It is therefore apparent that if the function $q$ is negative throughout the domain, then $R_{up}$ is negative for any choice of $y\in E$. In such a case, we have that: 
\begin{align}
\lambda_0=\sup_{y\in E, \ y\neq 0} R(y)<\sup_{y\in E, \ y\neq 0} R_{up}(y)<0.
\end{align}
\end{subequations} 
\end{proof}

We now show that under normoxia, $q < 0$ if the death rate is high and 
the magnitude of the phenotypic advection velocity is sufficiently large. 
\begin{lemma}
	If the model proliferation rate, apoptosis rate and phenotypic 
	advection velocity and diffusion coefficient are chosen such that: 
	\begin{equation}
	 \int_0^1 \left \{ p(z,c)-f(z)- \frac{v_z^2}{4\theta} \right \} dz>0,\label{cond_ins}
	\end{equation}
	then the trivial steady state is unstable. \label{lemma2}
\end{lemma}
\begin{proof}
	Consider $y_0\equiv 1$, then $y_0\in V$ and $\lambda_0=R(y_0)$ where:
	
	\begin{equation*}
	R(y_0)=\int_0^1 \left \{ p(z,c)-f(z)- \frac{v_z^2}{4\theta} \right \} dz>0.
	\end{equation*}
	
	Consequently, $\sup_{y\in V} R(y)\geq R(y_0)>0$, and our steady state is unstable.
	\end{proof}
\begin{remark}
	Note that for~(\ref{cond_ins}) to hold we require $\int_0^1 (p-f) dz>0$ so that cell proliferation dominates apoptosis.
	Based on the functional form defined in Section~\ref{fit}, we have that:
	\begin{equation}
	\begin{aligned}
	\hspace{-8mm}I(c;d_f)&=\int_0^1 (p-f ) dz\\
	&= \sqrt{g_1} p_1(c) \left[\mathcal{Z}\left(\frac{z-0.55}{\sqrt{g_1}}\right)
	+\sqrt{g_0} p_0(c) \mathcal{Z}\left(\frac{z}{\sqrt{g_0}}\right)+\frac{d_f}{k_f}e^{-k_f z}\right]_{z=0}^{z=1}\\
	&\sim \frac{\sqrt{g_0} p_0(c)}{2}+\sqrt{g_1} p_1(c)-\frac{d_f}{k_f}
	\end{aligned}
	\end{equation} 
    where $\mathcal{Z}$ denotes the cumulative distribution function for the normal distribution. We note that $I(1;d_f)>0$ while $I(0.2;d_f)<0$ for all values of the parameters listed in Table \ref{par_fitness}. We conclude that under normoxia there is a threshold $\mathcal{V}_+(\xi_+,\omega)$ such that the system is unstable for all choices of $V_+<\mathcal{V}_+(\xi_+,\omega)$:
	\begin{subequations}
		\begin{align}
	\mathcal{V}_+=\sqrt{\frac{2I(1;d_f)\theta}{I_v(\xi_+,\omega_+)}},\\
	\mbox{where} \; I_v(\xi_+,\omega_+)=\int_0^1 \left(\frac{1}{V^*_+}\tanh\left(\myfrac[3pt]{z^{\omega_+}}{\xi_+}\right)\tanh\left(\myfrac[2pt]{(1-z)}{\xi_+}\right)\right)^2 dz.
	\end{align}
	We note also that higher values of $\theta$ favour instability of the trivial solution as $\mathcal{V}_+$ increases with $\theta$. By inspecting Figure \ref{vel_prof}, we note qualitatively that $I_v$ is expected to decrease for increasing values of $\xi_+$ and $\omega_+$. 
	\end{subequations} 
\end{remark}

\begin{figure}[h!]
	\begin{subfigure}{0.49\textwidth}
		\centering
		\includegraphics[width=0.8\textwidth]{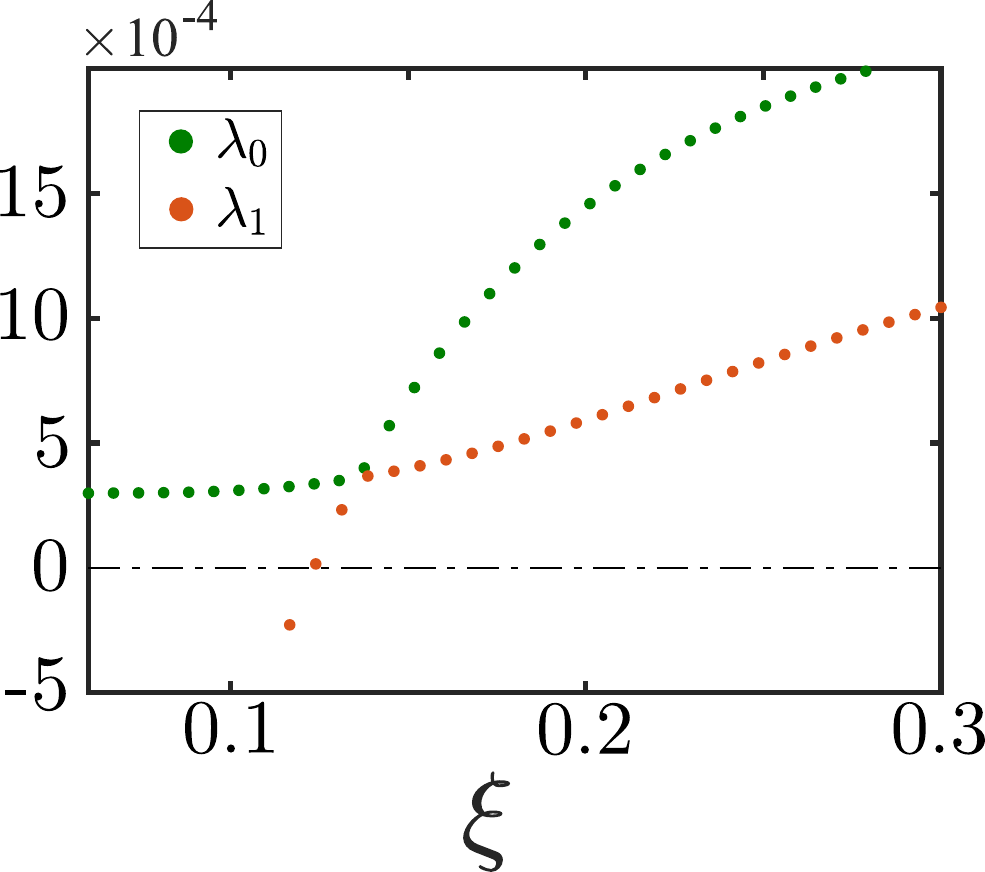}
		\caption{}
		\label{lsa1_a}
	\end{subfigure}	
	\begin{subfigure}{0.49\textwidth}
		\centering
		\includegraphics[width=0.8\textwidth]{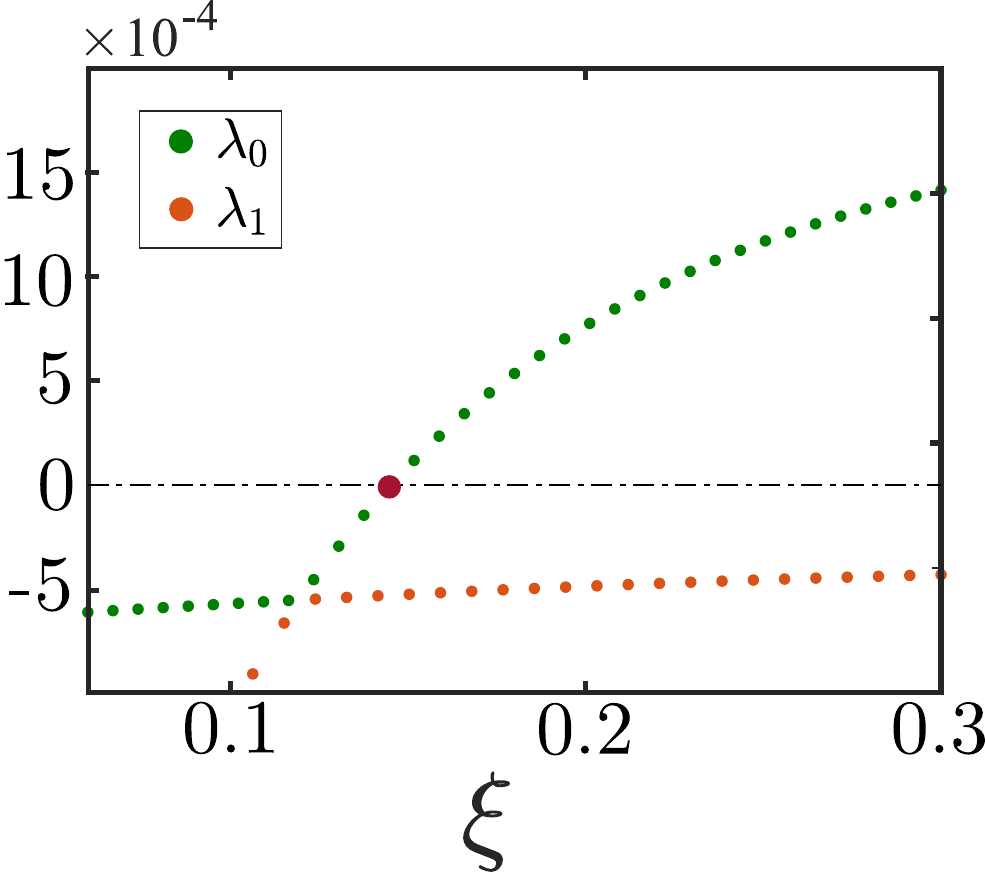}
		\caption{}
		\label{lsa1_b}
	\end{subfigure}	
	\caption{Linear stability analysis 
	of the trivial solution: plot of the two largest
		eigenvalues $\lambda_0(\xi)$ and $\lambda_1(\xi)$ for (a) $V_+=6\times 10^{-4}$, $\omega_+=1$ and $d_f=0.001$, and (b) $V_+=8\times 10^{-4}$, $\omega_+=1$, $d_f=0.001$. In (a), $\lambda_0 > 0$ for all values of $\xi$. In (b), $\lambda_0$ changes sign as $\xi$ increases and we can 
		identify a critical value of $\xi$ at which 
		the trivial solution loses stability, favouring the emergence of a nontrivial, phenotypic cell distribution.}
	\label{lsa1}
\end{figure}

To analyse other regions of parameter space, where neither of the sufficient conditions holds, we rely on numerical estimates of the eigenvalue $\lambda_0$. As shown in Figure \ref{lsa1}, and as expected based on the above findings, when the magnitude of the velocity $V_+$ is small,
$\lambda_0>0$ for all $\xi$ and the trivial solution is unstable.
By contrast, as the magnitude of the advection velocity
increases, its steepness, $\xi$, determines
the stability of the trivial solution. Using this estimate, we can identify the region of stability of the trivial steady state (see Figure \ref{xi_crit} in Section \ref{LSA}).
We remark that the boundary between the regions of stability is non-smooth. This is because $\lambda_0 = \lambda_0(\xi)$ plateaus as $\xi\ll1$ (see Figure \ref{lsa1}). By computing the second largest eigenvalue, $\lambda_1(\xi)$, we observe that the sharp change in the profile of $\lambda_0$ as $\xi$ decreases
occurs where $|\lambda_0-\lambda_1|$ attains its minimum value. 
It is possible to show that the two eigenvalues do not cross, as expected by the Sturm Oscillation theorem. A similar phenomenon occurs in quantum physics~\cite{cohent} where it is known as \textit{avoided crossing}.

Finally, we consider the stability of the trivial solution in an hypoxic environment. We confirm the numerical simulations from \S\ref{hyp cond} by showing that, under hypoxia, the trivial solution is always unstable.
\begin{lemma}
	Under hypoxia (i.e. when $c=0.2$), and for the parameter values listed in Table \ref{param_set}, the trivial steady state is always unstable. \label{lemma3}
\end{lemma}
\begin{proof}
Let us consider as a trial function:
\begin{equation}
y=\frac{1}{(\pi \kappa^2)^{1/4}}\exp\left(-\frac{z^2}{2\kappa^{2}}\right) +Az^2,
\end{equation}
where a small parabolic correction is added to the standard Gaussian, the constant $A$ being chosen to ensure
that the boundary condition at $z=1$ is satisfied:
\begin{equation}
 A=\frac{e^{-\frac{1}{2\kappa^2}}}{2\pi^{1/4}\kappa^{5/2}};
\end{equation}
the derivative $y'$ at $z=0$ vanishes, by construction.  We now want to show that the Rayleigh quotient is positive for such a choice of the test function $y$ which implies that the trivial steady state is unstable.

Given that the denominator of $R(y)$ is always positive, its sign will be determined by the numerator $R_n(y)$ that is:
\begin{equation}
R_n(y)=\int_0^1 \left(p-f-\frac{v_z^2}{4\theta}\right) y^2 dz -\cancelto{0}{\left.\frac{v_z y^2}{2}\right|_0^1} +\int_0^1 v_z y\frac{d y}{d z} - \theta \left(\frac{d y}{d z}\right)^2 dz
\end{equation}
Computing the derivative of $y$ and denoting the Gaussian by $y_0$, we obtain:
\begin{subequations}
\begin{align}
    y^2= y_0^2 +2Az^2y_0+A^2z^4,\\
    y'^2= \frac{z^2}{\kappa^4} y_0^2 -\frac{4Az^2}{\kappa^2}y_0+4z^2A^2,\\
    yy'= -\frac{z}{\kappa^2}y_0^2+Az\left(2-\frac{z^2}{\kappa^2}\right)y_0+2A^2z^3.
\end{align}
Recalling that the  constant $A$ is exponentially small in $\kappa$ while $y_0$ grows only as a power law of $\kappa^{-1}$, the terms multiplied by $A$ will be negligible and the sign of $R_{n}(y)$ will be determined also by the leading term:
\begin{align}
R_n(y)= I_0 + \mathcal{O}(A),\label{eq:Rnexp}\\
I_0=\int_0^1 \left(p-f\right)y_0^2 dz-\theta\int_0^1 m^2(z) y_0^2 dz,\label{eq:I0first}\\
\quad{where} \quad m(z)=\frac{v_z}{2\theta}+ \frac{ z}{\kappa^2}.
\end{align}
\end{subequations}
Proving instability therefore reduces to show that $I_0$ is positive for the range of parameters and functional forms considered in hypoxic condition. We do so finding a lower bound on the value on $I_0$, exploiting the quick decay of the function $y_0$, whose mass is concentrated in a neighborhood of $z=0$. Given that $p(0)-f(0)>0$ and $m(0)=0$, provided that $m$ does not grow too fast near $z=0$, we can intuitively see that the major contribution to the integral $I_0$ will be positive. We will now expand this intuitive argument with a more rigorous calculation. 

We first focus on $I_0^{(1)}=\int_0^1(p-f)y_0^2 dz$, the contribution in~(\ref{eq:I0first}) due to cell proliferation.
We can compute this integral exactly as the integrand comprises products of exponentials, that can be re-written as the integral of Gaussian distribution:
\begin{subequations}
\begin{align}
  I_0^{(1)}= &\left[\frac{p_1(c)\sqrt{2}\zeta_1 }{\kappa}e^{-\frac{0.55^2}{g_1}+\frac{c_1^2}{2\zeta_1^2}}\mathcal{Z}\left(\frac{z-c_1}{\zeta_1}\right)\right.\\
 &\left.+\frac{p_0(c)\sqrt{2}\zeta_0}{\kappa}\mathcal{Z}\left(\frac{z}{\zeta_0}\right)-d_f e^{-k_f+\frac{c_f^2}{\kappa^2}}\mathcal{Z}\left(\frac{\sqrt{2}(z-c_f)}{\kappa}\right)\right]_{0}^1,
 \end{align}\label{eq:I_0(1)}%
 \end{subequations}
 where $2\zeta_{0,1}^2=(\kappa^2 g_{0,1})/ (g_{0,1}+\kappa^2)$, $c_1=0.55(2\zeta^2_1)/g_1$ and $c_f=k_f\kappa^2/2$, while $\mathcal{Z}$ is again the normal cumulative distribution function as in Lemma \ref{lemma2}.

We now focus on the term in~(\ref{eq:I0low}) which depends on $m$. In this case the integral can not be computed exactly and we will therefore find a lower bound for its contribution instead. This is achieved by decomposing the full domain $[0,1]$ into two three sub-domains. This will allow us to balance the rapid growth of the function $m$ with the quicker decay of $y_0$ away from $z=0$:

\begin{align}
    \int_0^1 m^2y^2_0 dz  =\int_0^{z_0\kappa} m^2y^2_0dz + \int_{z_0\kappa}^{z_1\kappa}m^2 y_0^2dz+\int_{z_1\kappa}^1 m^2y_0^2 dz.\label{eq:mint1}
\end{align}
where $z_{0,1}$ are positive constants such that $0 < z_0 < z_1 < \kappa^{-1}$. Note that we have the freedom of choosing their values with the aim of making the quantity in~(\ref{eq:mint1}) as small as possible. 
It is straightforward to see that $m$ attains its maximum value at $z=1$ as both $v_z$ and $z/\kappa^2$ attain maxima there. We now choose the value of $\kappa$ so that the derivative of $m$ at $z=0$ vanishes:
\begin{subequations}
\begin{align}
m'(z)= \left(\frac{v'_z(z)}{2\theta}+\frac{1}{\kappa^2}\right) \quad \Rightarrow \quad \kappa = \sqrt{\frac{2\theta}{|v'_z(0)|}}.\label{eq:kappa}
\end{align}
However, by definition (see Equation~(\ref{eq_ad_hyp})), under hypoxia, the advection velocity $v_z(z)=v_z^-(z)$ is such that $|v'_z(z)|\leq |v'_z(0)|$ for all $z\in(0,1]$, with equality only if $\omega_-=2$. Consequently, we have that $m(z)$ is a non-decreasing function of $z$, i.e. $m'(z)\geq 0$.


Given the above, we can now construct an upper bound for the integral in~(\ref{eq:mint1}):
\begin{align}
\begin{aligned}
   \hspace{-8mm} \int_0^1 m^2 y^2_0 dz  &\leq m^2(z_0 \kappa)\int_0^{z_0 \kappa} y^2_0 \: dz + m^2(z_1 \kappa)\int_{z_0 \kappa}^{z_1 \kappa} y^2_0 \: dz + m^2(1)\int_{z_1 \kappa}^1 y_0^2 \: dz\\
    &=m^2(z_0 \kappa) \left[\mathcal{Z}\right]_0^{\sqrt{2}z_0}+m^2(z_1 \kappa) \left[\mathcal{Z}\right]_{\sqrt{2}z_0}^{ \sqrt{2}z_1}+ \frac{1}{\kappa^4} \left[\mathcal{Z}\right]_{\sqrt{2} z_1}^{\frac{\sqrt{2}}{\kappa}}\\
    &\leq \frac{m^2(z_0 \kappa)}{2}+m^2(z_1 \kappa) \left[\mathcal{Z}\right]_{\sqrt{2}z_0}^{ \sqrt{2}z_1}+ \frac{1}{\kappa^4} \left[\mathcal{Z}\right]_{\sqrt{2} z_1}^{\infty}
\end{aligned}
\end{align}
Let us reiterate that we want $z_0$ and $z_1$ to be such that $m^2(z_0\kappa)$ and $m^2(z_1\kappa)$ are not too large while $\left[\mathcal{Z}\right]_0^{\sqrt{2}z_0}$ and $\left[\mathcal{Z}\right]_{\sqrt{2} z_1}^{\infty}$ are sufficiently small. In this way, the growth of $m$ is balanced by the exponential decay of the Gaussian function $y^2_0$. In particular, we choose $z_0=\sqrt{2}$ and $z_1=5/\sqrt{2}$. 
\end{subequations}
Combining the above with the estimate from Equation~(\ref{eq:I_0(1)}), we obtain:
\begin{eqnarray}
I_0  >  I_0^{(1)}- \frac{\theta m^2(z_0\kappa)}{2} -\theta m^2(z_1\kappa)\left[\mathcal{Z}\right]_{\sqrt{2}z_0}^{ \sqrt{2}z_1}-\frac{v'_z(0)^2}{4\theta} \left[\mathcal{Z}\right]_{\sqrt{2} z_1}^{\infty}=I^{low}_0.\label{eq:I0low}
\end{eqnarray}
\begin{figure}[h!]
    \centering
    \includegraphics[width=0.95\textwidth]{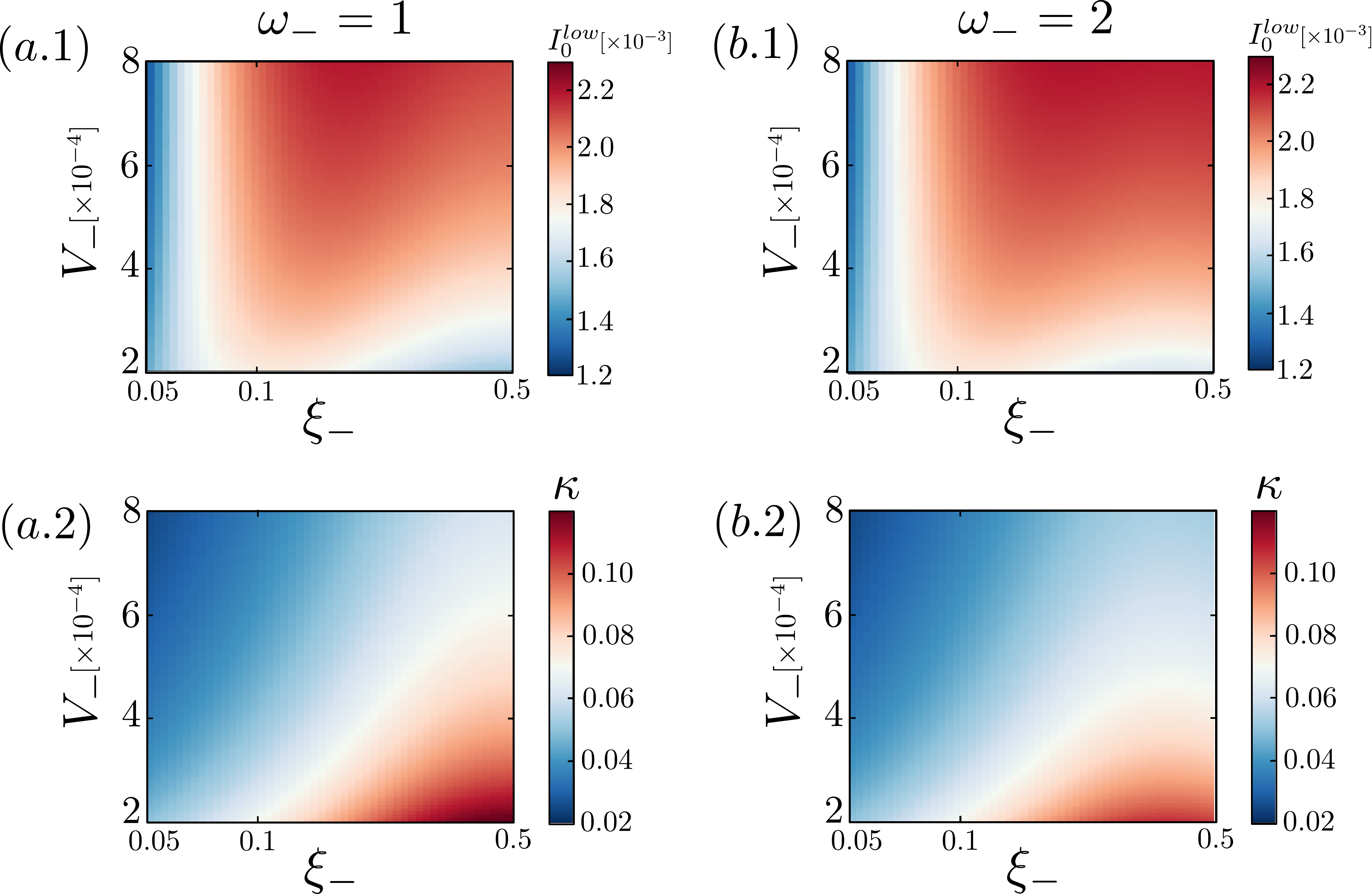}
    \caption{Plot of the lower bound $I_0^{low}$ and the standard deviation $\kappa$ as defined by~(\ref{eq:I0low}) and~(\ref{eq:kappa}) respectively for parameter regime considered in the paper (note that $d_f$ is fixed to its maximum values $0.015$ as this gives the smaller bound $I_0^{max}$).}
    \label{fig:my_label}
\end{figure}
We can compute the values of $\kappa$ and $I_0^{low}$ associated with the value of the  magnitude $V_-$ and steepness $\xi_-$ considered in the paper (without loss of generality, we only consider $d_f=0.015$ as $I^{low}_0$ decreases with $d_f$). As shown in Figure \ref{fig:my_label}, for all such values, we have that $I_0^{low}>0$. As $I_0>I_0^{low}$, we therefore have that generically $I_0$
is also positive. We estimate $A \leq O(10^{-13})$ which justifies us dropping the $O(A)$ in~(\ref{eq:Rnexp}). Consequently, we conclude that $R_n(y)$ is positive and so is the quotient $R$. Hence, in hypoxia, the trivial steady state is always unstable.

\end{proof}

\section*{Acknowledgements}
\noindent The authors wish to thank Professor Philip K. Maini for helpful comments and feedback on the manuscript.~GC is supported by from EPSRC and MRC Centre for Doctoral Training in Systems Approaches to Biomedical Science and Cancer Research UK. C.Z. acknowledges Breast Cancer Research Foundation (BCRF).
P.G.K.  acknowledges support from the Leverhulme
Trust via a Visiting Fellowship and thanks the Mathematical
Institute of the University of Oxford for its hospitality
during part of this work.

\bibliographystyle{elsarticle-num} 
\bibliography{main}


\end{document}